%% file: triangle_counting_algorithm.tex
\pdfoutput=1
\documentclass[11pt]{article}
\PassOptionsToPackage{hidelinks}{hyperref}
\usepackage{amsmath,amssymb,amsthm,algpseudocode,algorithm,float,tkz-berge,caption,mathtools}
  
\usepackage[margin=1in]{geometry}
\usepackage{multirow}
\usepackage{subcaption}
\usepackage{afterpage}
\usepackage{tikz}
\usetikzlibrary{calc}
\usetikzlibrary{graphs,graphs.standard}

\usepackage{hyperref}

\DeclareMathOperator*{\var}{\textrm{Var}}

\DeclareMathOperator*{\lo}{\textproc{Light}}

\DeclareMathOperator*{\instopt}{\textproc{InstOpt}}

\input{libtheorems.tex}

\newcommand{\Os}{O^*}
\newcommand{\Thetat}{\widetilde{\Theta}}
\newcommand{\wt}[1]{\widetilde{#1}}
\newcommand{\abs}[1]{|#1|}
\newcommand{\eps}{\epsilon}
\newcommand{\E}{\mathbb{E}}

\newtheorem{theorem}{Theorem}
\newtheorem*{theorem*}{Theorem}
\newtheorem{lemma}[theorem]{Lemma}

\newtheorem*{definition*}{Definition}
\newtheorem*{lemma*}{Lemma}
\newtheorem{corollary}[theorem]{Corollary}
\newtheorem*{corollary*}{Corollary}

\newtheorem*{claim*}{Claim}

\title{A Hybrid Sampling Scheme for Triangle Counting}
\date{}
\author{%
  John Kallaugher\\
  \texttt{jmgk@cs.utexas.edu}\\
  The University of Texas at Austin
  \and
  Eric Price \\
  \texttt{ecprice@cs.utexas.edu}\\
  The University of Texas at Austin
}

\begin{document}

\begin{titlepage}
  \maketitle
  \thispagestyle{empty}

  \begin{abstract}
    We study the problem of estimating the number of triangles in a
    graph stream.  No streaming algorithm can get sublinear space on
    all graphs, so methods in this area bound the space in terms of
    parameters of the input graph such as the maximum number of
    triangles sharing a single edge.  We give a sampling algorithm
    that is additionally parameterized by the maximum number of
    triangles sharing a single vertex.  Our bound matches the best
    known turnstile results in all graphs, and gets better performance
    on simple graphs like $G(n, p)$ or a set of independent triangles.

    We complement the upper bound with a lower bound showing that no
    sampling algorithm can do better on those graphs by more than a
    log factor.  In particular, any insertion stream algorithm must
    use $\sqrt{T}$ space when all the triangles share a common vertex,
    and any sampling algorithm must take $T^{1/3}$ samples when all
    the triangles are independent.  We add another lower bound, also
    matching our algorithm's performance, which applies to \emph{all}
    graph classes.  This lower bound covers ``triangle-dependent''
    sampling algorithms, a subclass that includes our algorithm and
    all previous sampling algorithms for the problem.

    Finally, we show how to generalize our algorithm to count arbitrary
    subgraphs of constant size.
  \end{abstract}
\end{titlepage}

\section{Introduction}

Estimating the number of triangles in a graph is a fundamental graph
problem.  It is the smallest graph structure that cannot be counted
directly from local information at each node, so it is a useful
benchmark for understanding the power of various models of graph
computation.  At the same time, the number of triangles is itself a
useful piece of information, and one of the standard features for
understanding networks~\cite{ML12}.

We consider the problem of triangle counting on graphs, under three
different models of streaming access to the graph.  In increasing
order of restrictiveness, these are: \emph{insertion-only} streams,
where the edges $E$ of $G$ are added in adversarial order but are not
deleted; \emph{turnstile} streams, where edges may be inserted and
deleted, and is equivalent (with some restrictions) to linear
sketches~\cite{LNW14}; and nonadaptive \emph{sampling}, where the
algorithm specifies a distribution over sets $S$ and receives $E \cap
S$, for which it pays a ``sample complexity'' of $\E \abs{E \cap S}$.
Sampling is a more restrictive model than turnstile, because it can be
implemented using linear sketches by composing an $O(\abs{E \cap S})
\times E$ noiseless sparse recovery sketch with a diagonal $E \times
E$ subsampling matrix.  At the same time, many turnstile algorithms
for graphs can be cast as sampling algorithms, and sampling algorithms
have some advantages over general linear sketches, such as the ability
to filter the graph after receiving the sketch, so we find it useful
to distinguish the two models.

In the triangle estimation problem, one observes a graph $G$ with $T$
triangles and would like to output $\overline{T}$ such that $\abs{T -
  \overline{T}} \leq \eps T$ with probability $1-\delta$, using as
little space/sample complexity as possible.

In the following discussion, we suppose that algorithms using a
parameter know the optimal value of that parameter up to constant
factors.  We also use $\Os(f)$ to hide factors of $\eps$, $\log
\frac{1}{\delta}$, and $\log n$.

\paragraph{Streaming algorithms for triangle counting.} The problem of
estimating triangles from a graph stream was introduced
in~\cite{BKS02}, which gave an $\Os(\left(\frac{mn}{T}\right)^3)$
space algorithm based on estimating frequency moments in the
insertion-only model.  This was improved in~\cite{BFLMS06} to
$\Os(\frac{mn}{T})$ in the same model. An incomparable algorithm was
given in~\cite{JG05} that uses $\Os(\frac{md^2}{T})$ space in the
insertion-only model, where $d$ is the maximum degree, based on
subsampling the edges at rate $d/T$ and storing all subsequent edges
that touch the sampled edges.

None of the above algorithms are sublinear in $m$ for all graphs.
Unfortunately, as shown in~\cite{BOV13}, this is inherent: even in the
insertion model, no streaming algorithm can distinguish $0$ triangles
from $T$ triangles in $o(m)$ space for all graphs.  The hard case
involves a graph where all the triangles use a single edge.  Since
graphs of interest typically don't involve that structure, a natural
question arose of finding sublinear algorithms when $\Delta_E$, the
maximum number of triangles sharing a common edge, is $o(T)$.

The first algorithms to achieve this,~\cite{TKMF09,TKM11}, work by
subsampling the edges with uniform probability $p$.  The expected
number of triangles is $p^3T$, so one needs $p > \frac{\Delta}{T} +
\frac{1}{T^{1/3}}$ in order to (a) sample heavy edges with good
probability and (b) sample at least one triangle in expectation.  In
fact this is sufficient, giving an algorithm with space
$\Os(m\left(\frac{\Delta_E}{T} + \frac{1}{T^{1/3}}\right))$.

In order to improve this, one would like a different sampling scheme
that increases the chance of sampling edges in a triangle at the same
time.  In~\cite{PT12}, this is done by randomly coloring the vertices
of the graph with $1/p$ colors, and keeping the monochromatic edges.
This now samples $p^2T$ triangles in expectation, and indeed they show
that $\Os(m\left(\frac{\Delta_E}{T} + \frac{1}{T^{1/2}}\right))$ samples
suffice.

\paragraph{Our results.}
First, we show a distribution on graphs with $\Delta_E = 1$ for which
$\Omega(\frac{m}{\sqrt{T}})$ space is required to estimate the number
of triangles to within $25\%$, even in the insertion-only model.  This
shows that the~\cite{PT12} bound is tight in general.  However, just
as the $\Omega(m)$ lower bound in~\cite{BOV13} involved unrealistic
graphs where all triangles shared a single edge, our
$\Omega(m/\sqrt{T})$ lower bound involves unrealistic graphs where all
triangles share a single vertex.  So we consider algorithms
additionally parameterized by $\Delta_V$, the maximum number of
triangles sharing a common vertex. 

We give a sampling algorithm using expected
$\Os(m\left(\frac{\Delta_E}{T} + \frac{\sqrt{\Delta_V}}{T} +
  \frac{1}{T^{2/3}}\right))$ samples.

\define[Triangle estimation]{thm:main}{Theorem}{%
  If $\wt{T} \leq T$, our algorithm obtains an
  $(\epsilon,\delta)$ approximation to $T$ while keeping $O\left(
    \frac{m \log
      \frac{1}{\delta}}{\epsilon^2}\left(\frac{1}{\wt{T}^{2/3}}
      + \frac{\sqrt{\Delta_V}\log \wt{T}}{\wt{T}}
      +\frac{\Delta_E\log \wt{T}}{\wt{T}}\right)\right)$
  edges. If $\wt{T} > T$, the algorithm either determines that 
  $\wt{T} > T$ or obtains a $(1 \pm \epsilon)$ approximation
  to $T$ with probability $1 - \delta$.
    }
    \state{thm:main}

    \begin{figure}
      \centering
      \renewcommand{\arraystretch}{1.5}
      \begin{tabular}{|l|l|l|}
        \hline
        Paper & Space used & Model\\
        \hline\hline
        \cite{BKS02} & $(\frac{mn}{T})^3$ & Insertion\\
        \cite{BFLMS06} & $\frac{mn}{T}$ & Insertion\\
        \cite{JG05} & $\frac{md^2}{T}$ & Insertion\\
        \hline
        \cite{KP14,BFKP14} & $\frac{\sqrt{m}}{\alpha}$ & Turnstile\\
        \cite{TKMF09,TKM11} & $m\left(\frac{\Delta_E}{T} + \frac{1}{T^{1/3}}\right)$ & Sampling\\
        \cite{PT12} & $m\left(\frac{\Delta_E}{T} + \frac{1}{\sqrt{T}}\right)$ & Sampling\\
        \hline
        This work & $m\left(\frac{\Delta_E}{T} + \frac{\sqrt{\Delta_V}}{T} + \frac{1}{T^{2/3}}\right)$ & Sampling\\
        \hline
      \end{tabular}
      \caption{Algorithms for triangle estimation.  For simplicity, we
        drop dependencies on $\eps$, $\delta$, and logarithmic
        factors.  Here, $\alpha$ denotes the transitivity coefficient,
        $d$ denotes the maximum degree of any vertex, and
        $\Delta_E/\Delta_V$ denote the maximum number of triangles
        sharing a common edge/vertex.}
      \label{fig:existing}
    \end{figure}

    \input{graphcommands.tex}
    \begin{figure}
      \centering
      \renewcommand{\arraystretch}{1.5}
      \begin{tabular}{|l|l|c|c|c|c|}
        \hline
        & & Heavy edge & Hub & $G(n, p)$ & Independent\\
        Method&Model&\drawheavy&\drawhub&\drawclique&\drawindep\\
        \hline\hline
        \cite{BKS02} & Insertion & $n^3$ & $n^3$ & $1/p^6$ & $n^3$\\
        \cite{BFLMS06} & Insertion & $\mathbf{n}$ & $n$ & $1/p^2$ & $n$\\
        \cite{JG05} & Insertion & $n^2$ & $n^2$ & $n$ & $\mathbf{1}$\\
        \hline\hline
        \cite{KP14,BFKP14} & Turnstile & $n^{3/2}$ & $\mathbf{\sqrt{n}}$ & $n/\sqrt{p}$ & $\sqrt{n}$\\
        \cite{TKMF09,TKM11} & Sampling & $\mathbf{n}$ & $n^{2/3}$ & $n$ & $n^{2/3}$\\
        \cite{PT12} & Sampling & $\mathbf{n}$ & $\mathbf{\sqrt{n}}$ & $\sqrt{n/p}$ & $\sqrt{n}$\\
        \hline
        This paper & Sampling & $\mathbf{n}$ & $\mathbf{\sqrt{n}}$ & $\mathbf{1/p}$ & $\mathbf{n^{1/3}}$\\
        \hline\hline
        \multirow{2}{*}{Lower bounds} & Insertion & $n$~\cite{BOV13,BFKP14} & $\mathbf{\sqrt{n}}$ & ? & 1 \\
        \cline{2-6}
        & Sampling & same & same & $\mathbf{1/p}$ & $\mathbf{n^{1/3}}$ \\
        \hline
      \end{tabular}
      \caption{Results on specific graphs, ignoring logarithmic
        factors. For upper bounds, entries in bold are optimal in the
        computational model of the row.  For lower bounds, entries in bold
        are new.  The lower bound instances are subsets of the illustrated
      graphs containing a constant fraction of the triangles.}
      \label{fig:existingspecific}
    \end{figure}

    \paragraph{Illustrative examples.} To compare our result with
    those in the literature, in Figure~\ref{fig:existingspecific} we
    specialize the bounds to some illustrative example graphs.  The
    \emph{heavy edge} example is the lower bound from~\cite{BOV13}
    (and similar to one in~\cite{BFKP14}), where all triangles use a
    single edge.  Any streaming algorithm requires $\Omega(m)$ space
    in examples like this one, demonstrating the need for further
    parameterization, and most streaming algorithms match the bound.
    The \emph{hub} example has a single vertex involved in $n$
    disjoint triangles.  Here, we demonstrate that the $\sqrt{n}$
    achieved by~\cite{PT12} was optimal even for insertion-only
    streams.

    We then consider graphs where most triangles do not overlap.  One
    natural example is the Erd\"os-R\'enyi random graph $G(n, p)$,
    with $1 \geq p \gg \frac{1}{n}$ so the number of triangles is well
    concentrated.  Here, we use $\Os(1/p)$ samples, and show that this
    is optimal for any sampling method.  Previous algorithms in the
    turnstile/sampling models took $\Omega(\sqrt{n})$ space for dense
    graphs, and even in the insertion-only model they took
    $\Omega(1/p^2)$ space; our bound implies $\Os(1/p)$.  The next
    example we consider is a collection of $n$ independent triangles.
    Here, our algorithm takes $\Os(n^{1/3})$ samples, which we show is
    optimal for any sampling algorithm, and improves upon the previous
    bound of $\sqrt{n}$.  It is an interesting question whether
    $n^{1/3}$ space is necessary in less restrictive streaming models.

    \paragraph{Instance lower bounds.}
    It takes some care to define lower bounds for instances in a
    meaningful way.  Let us consider trying to solve triangle counting
    for a specific class of graphs $G_X$, where $X$ is some set of
    parameters (such the number of triangles in the hub graph). We
    would like to avoid ``cheating'' algorithms; for instance, a hub
    graph with $m$ edges has $m/3$ triangles, so one could estimate
    the number of triangles in the class of hub graphs by simply
    counting the edges, which takes $1$ word in the turnstile model.

    One attempt to avoid this would be to show the difficulty of
    distinguishing each graph $G \in G_X$ from an alternative graph
    $G'$ with a significantly different number of triangles.  This is
    too weak: $G'$ can introduce a difficult subproblem that is
    unnatural for the class $G_X$, so the lower bound doesn't really
    represent the difficulty of $G_X$.  For example, when $G$ has $n$
    independent triangles, we can have $G'$ be $G$ plus a clique on
    $n^{1/3}$ vertices.  A sampling algorithm using less than
    $n^{1/3}$ edges would entirely miss the clique and be unable to
    distinguish the two, giving a lower bound by this definition, but
    not a satisfactory one: the alternative hard instance $G'$ doesn't
    have the independent-triangle structural property we expect.  To
    preserve the structure of the graph class, we would like to only
    consider graphs $G'$ that are subgraphs of $G$.

    The broader question here is, given that existing algorithms
    out-perform the $\Omega(m)$ lower bound by adding the extra parameter
    $\Delta_E$, and we in turn out-perform these by adding the extra
    parameter $\Delta_V$, where should we draw the line? What are the
    correct set of parameters?

    We suggest that an algorithm can correctly count triangles for a
    graph $G$ with $S$ samples and error $\epsilon T(G)$, it should be
    able to correctly count triangles for any subgraph $G'$ of $G$
    with error $\epsilon T(G)$. (so additive error should be preserved,
    but not necessarily multiplicative error) Alternatively stated,
    the only thing that makes counting triangles in a larger graph
    easier is the fact that the larger graph may have more triangles,
    and thus have more tolerance for (additive) error.

    We bolster this intution by observing that existing algorithms 
    all depend on two sets of parameters: $T$ itself, in which their 
    complexity is decreasing, and a set of monotonic\footnote{Meaning,
    in this context, that if $A$ is a subgraph of $B$, $f(A) \leq f(B)$}
    graph functions, in which their complexity is increasing
    (e.g. $\Delta_V$ or $\Delta_E$)\footnote{The only apparent
      exception is those algorithms which depend on the transitivity
      coefficient $\alpha$, but it holds if we replace $\alpha$ with
      $T/P_2$, $P_2$ being the number of wedges in the
      graph.}.

    We also assume that a triangle countring algorithm should be resilient
    to vertex labels being arbitrarily permuted, as it should
    not depend on knowing beforehand that certain vertices are
    ``special.'' This gives us our definition of ``solving'' an
    instance.

    \define{dfn:solving}{Definition}{%
        Let $G$ be a graph. We say an algorithm \emph{solves}
        $G$ with $S$ space/samples and $(\epsilon,\delta)$ error if, 
        for any $G'$ isomorphic to some subgraph of $G$, the 
        algorithm returns $T(G') \pm \epsilon T(G)$ with $1 - \delta$
        probability, using no more than $S$ space/samples.
    }
    \state{dfn:solving}

    We now define the \emph{instance-optimum} for the space/sample
    complexity of solving $G$.
    
    \define{dfn:ilb}{Definition}{%
       For any given streaming model, $\instopt(G,\epsilon,\delta)$ is
       the least amount of samples/space such that some algorithm
       solves $G$ with $\instopt(G,\epsilon,\delta)$ samples/space.
    }
    \state{dfn:ilb}


    An instance lower bound, then, is a lower bound on $\instopt(G_X, 
    \epsilon,\delta)$.  The lower bounds in 
    Figure~\ref{fig:existingspecific} use this definition.

    It is an interesting question whether $\Omega(n^{1/3})$ space is
    necessary for $n$ independent triangles under turnstile or
    insertion streams.  Our hard instance consists of randomly
    coloring the vertices with two colors, and either choosing the
    monochromatic or dichromatic edges.  We do not know how to solve
    this instance with a turnstile streaming algorithm, but reductions
    from communication complexity seem to want three-party
    communication lower bounds, which are difficult to show.  One can
    easily solve this particular instance in insertion streams, but
    similar instances of independent triangles with extra edges may be hard.

    \paragraph{Instance-optimal lower bound.}
    The results in Figure~\ref{fig:existingspecific} show that our
    algorithm performs well on several natural graphs, but what about
    other graphs?  Is there a more refined parameterization that would
    again yield large improvements on another class of graphs?

    With some caveats, we show that our algorithm is optimal for
    \emph{all} graphs.  We need to refine our parameterization
    slightly, because currently the bound for a graph containing
    $(1-\eps/2)T$ independent triangles and a heavy edge with $\eps
    T/2$ triangles would depend on the heavy edge, even though an
    upper bound could just skip it and remain within $1\pm \eps$
    accuracy.  We therefore define $\Delta_{E,\eps}$ to be like
    $\Delta_E$ but with the $\eps T$ triangles contributing the most
    to $\Delta_E$ removed, and $\Delta_{V,\eps}$ similarly (for a
    precise definition, see Definition~\ref{dfn:dgepsilon}).  This
    turns out to be a sufficient parameterization.

    We show that any \emph{triangle dependent} sampling
    algorithm---one that depends only on the set of triangles it
    samples---must use a set of samples with the expected dependence
    on $\Delta_{E,2\eps}$ and $\Delta_{V,2\eps}$.  All
    existing sampling algorithms for triangle counting are triangle
    dependent, but we cannot rule out better non-triangle-dependent
    algorithms.

    \define[Triangle-dependent sampling bound]{thm:tdbound}{Theorem}{%
      For any constant $\epsilon$ and for any graph $G$, 
      \[
      \instopt(G,\epsilon,1/10) = \Omega\left(m \left(\frac{1}{T^{2/3}} + \frac{
            \sqrt{\Delta_{V,2\epsilon}}}{T} + \frac{
            \Delta_{E,2\epsilon}}{T}\right)\right)
      \]
      in the setting of triangle-dependent sampling algorithms. 
    }%
    \state{thm:tdbound}

    We then show that our upper bound can depend on
    $\Delta_{E,\eps/24}$ and $\Delta_{V,\eps/24}$.  Therefore, for 
    any constant $\epsilon$, our algorithm calculates an 
    $(\epsilon,1/10)$ approximation to the triangle count with 
    $O^*\left(\instopt(G,\epsilon/48,1/10)\right)$
    samples.

    \paragraph{Beyond triangles.} We also give a generalization of our
    upper bound to counting arbitrary subgraphs of constant size.  We
    give an algorithm to estimate $M$, the number of instances of a
    fixed size-$s$ subgraph, from a sample of the edges.

    \define[Subgraph estimation]{thm:general}{Theorem}{%
      Let $f_\ell$ be the fraction of pairs of subgraphs that
      intersect at $\ell$ vertices.  We show how to find a
      $1+\epsilon$ factor approximation to $M$ with probability
      $1-\delta$, using order
      \[
        m\frac{\log(1/\delta)}{\eps^2}\log M \left(\sum_{\ell=2}^{s} f_\ell^{2/\ell} + f_\ell^{\frac{1}{\ell-1}}f_1^{1 - \frac{1}{\ell-1}}\right)
      \]
      samples in expectation.%
    }%
    \state{thm:general}%
    For comparison, simple vertex sampling would replace the sum with
    $\sum_{\ell=1}^{s} f_\ell^{2/\ell}$; our bound is always better.
    In the context of triangles, this difference is why we get
    $\frac{\sqrt{\Delta_V}}{T}$ rather than
    $\left(\frac{\Delta_V}{T}\right)^2$, which was important in the
    hub case.

    We get this bound using a similar scheme to our triangle
    estimation algorithm, removing a direct dependence on $f_1$ by
    treating ``heavy'' vertices specially.  In the case of $s = 3$,
    this is equivalent to our theorem up to a log factor.

    \paragraph{Other related work.}

    Another line of work parameterizes the space complexity using the
    \emph{transitivity coefficient} $\alpha$, defined as the fraction
    of wedges that are completed into triangles.
    In~\cite{KP14,BFKP14} it is shown how to get
    $\Os(\frac{\sqrt{m}}{\alpha})$ space in the insertion model, for
    graphs without isolated edges.  In~\cite{JSP13} it was shown that
    $O(\frac{m}{\eps^2\sqrt{T}})$ space suffices in insertion streams
    to learn $\alpha$ to $\pm \eps$.  In fact, as we note in
    Appendix~\ref{app:transitivity}, both bounds for triangle counting
    are directly implied by the $\Os(m\left(\frac{\Delta_E}{T} +
      \frac{1}{\sqrt{T}}\right))$ bound of~\cite{PT12}. Since our
    bound improves upon~\cite{PT12}, it also implies these bounds.

    \cite{CJ14} shows multipass algorithms take $\Thetat(m/\sqrt{T})$
    space for arbitrary graphs, giving an algorithm for two passes and a
    lower bound for a constant number of passes.  \cite{KMPT12} shows a
    three pass streaming algorithm using $O(\sqrt{m} + m^{3/2}/T)$ space.

    \cite{ELRS15} considered the problem of triangle counting with
    query access to a graph.  Similar to our algorithm, a simpler
    algorithm is modified to handle the impact of many vertices
    intersecting at a single triangle on the variance. The main
    difference is that, in \cite{ELRS15}, these ``heavy'' vertices are
    discarded without damaging the accuracy of the estimate, whereas
    we spend the bulk of our effort on attempting to estimate the
    number of triangles intersecting at each ``heavy'' vertex.
 
    \paragraph{Running time.} An $O(m^{3/2})$ time algorithm was given
    in~\cite{IR78} to list all the triangles in a graph.  This was
    improved for graphs with small arboricity by~\cite{CN85}.  For
    \emph{counting} triangles,~\cite{AYZ97} gave a different algorithm
    that improves the time to $m^{\omega/(1+\omega)} \approx m^{1.41}$
    using matrix multiplication.  In~\cite{BPVZ14}, it was shown how
    to extend this to \emph{listing} triangles in $o(m^{1.5})$ time
    when $T = o(m^{3/2})$.  Other works, such as~\cite{SW05,L08}, have
    given combinatorial $O(m^{3/2})$ time triangle listing algorithms
    that are more efficient in practice.  Our algorithm's running time
    is dominated by listing triangles in the subsampled graph, which
    we can do either using one of the $O(m^{3/2})$ time algorithms or
    (in some cases) slightly faster via~\cite{BPVZ14}.  Because our
    algorithm improves upon the number of edges necessary to
    approximate the triangle count, it also implies a faster method
    for approximately counting the number of triangles in a given
    graph (as in, e.g., \cite{KMPT12}).

    \section{Overview of Techniques}
    \subsection{Triangle Counting Algorithm}

    This algorithm is a modification of simple vertex sampling, where
    we sample each vertex in the graph with probability
    $\frac{1}{\sqrt k}$, and keep any edge between two sampled
    vertices. The estimate of $T$ is $k^\frac{3}{2}$ times the number
    of triangles sampled, and the sample complexity is $m/k$. This is
    appealing, because as we show in our proof of Theorem
    \ref{thm:itlb}, \emph{any} algorithm that samples edges at rate
    $1/k$ has a constant chance of sampling less than
    $\frac{T}{k^{3/2}}$ triangles.

    In order to approximate the number of triangles well, we need an
    estimator with variance $O(T^2)$. The vertex sampling estimator
    has variance bounded by $k^\frac{3}{2}T + k\sum_eT_e^2 +
    \sqrt{k}\sum_v T_v^2$, by the fact that a pair of triangles
    intersecting at $l$ vertices is $k^\frac{l}{2}$ times more likely
    to be sampled than a pair of triangles sampled independently.
    Choosing $k$ small enough for this to be $O(T^2)$ gives sample
    rate
    \[
    \frac{1}{k} \eqsim \frac{1}{T^{2/3}} + \frac{\Delta_E}{T} + \left(\frac{\Delta_V}{T}\right)^2
    \]
    The first and second terms are optimal, as we see in the
    independent triangles graph and heavy edge graph, but the third
    can be improved.

    This is the term that makes vertex sampling fail on the hub case,
    when there is a single vertex $v$ s.t. $T_v = T$. The vertex
    sampling algorithm will consistently fail to estimate the triangle
    count accurately in this case, as $v$ will only be sampled with
    probability $1/\sqrt{k}$, and so usually we will miss all the
    triangles, and occasionally we will overestimate the triangle
    count by a factor of $\sqrt{k}$. More generally, our issue is
    vertices where $T_v$ is large, and in particular vertices where it
    is larger than $\frac{T}{\sqrt{k}}$, as the total contribution to
    the variance of vertices with $T_v$ smaller than this is
    $O\left(T^2\right)$.

    We can deal with the hub case by extended the sampling in the
    following way.  After sampling vertices, in addition to taking all
    the edges between sampled vertices, we can also take $1/\sqrt{k}$
    of the edges between sampled vertices and unsampled vertices.
    Now, even when the central vertex $v$ is not sampled, each
    triangle in the hub has a $1/k^2$ chance of being sampled (if both
    other vertices, and both edges between $v$ and those vertices, are
    picked).  This is independent for the different triangles in the
    hub, so we will find a triangle when $k \approx \sqrt{T}$.  The
    resulting $m/\sqrt{T}$ sample complexity is optimal by our lower
    bound.

    So we could handle our vertices one of two ways---for our
    ``light'' vertices, we will use vertex sampling, and for our
    ``heavy'' ($T_v \geq \frac{T}{\sqrt{k}}$) vertices, we will use
    the scheme above. In order to identify the lightest of the heavy
    vertices, as we are sampling their triangles at rate $k^{-2}$, we
    would need $k^2 = \frac{T}{\sqrt{k}}$ and so $k <
    T^\frac{2}{5}$. Can we do better?


    We can think of vertex sampling and the ``hub scheme'' as two ends
    of a continuum.  For a given ``weight'' $x$, we can take two
    samples $S_1$ and $S_2$ of $V$ where each vertex appears in $S_1$
    with probability $1/\sqrt{k}$ and appears in $S_2$ with
    probability $x/\sqrt{k}$, then sample each edge in $S_1 \times
    S_2$ with probability $1/x$.  The $x = 1$ case is vertex sampling,
    and the $x = \sqrt{k}$ case is our hub scheme.  The trade-off is
    that small $x$ makes it possible to completely miss important vertices,
    but high $x$ means that we sample fewer triangles overall.

    Consider the ``many-hubs'' case where we have $T/\Delta_V$ hubs
    involved in $\Delta_V$ triangles each.  We need $x \geq \sqrt{k}
    \Delta_v / T$ to sample them reliably in the weight $x$ scheme,
    and we will get $\frac{T}{k^{3/2}x}$ triangles in expectation.
    When $x$ is minimized according to the first constraint, this is
    $\frac{T^2}{\Delta_v k^2}$ triangles in expectation.  For the
    variance to be small, we need this to be at least a constant, as
    happens for $k \leq \frac{\sqrt{\Delta_V}}{T}$.  This explains the
    $\frac{\sqrt{\Delta_V}}{T}$ bound in our algorithm.

    The above discussion applies when $x$ is optimized for a given
    graph.  Our full algorithm runs the scheme for all $\log k$
    different scales of $x$ and combines the results.  This lets us
    improve the variance bound to order $T^2 + k\sum_eT_e^2 +
    k^2\frac{\sum_v T_v^2}{T}$\footnote{Note that $\sum_vT_v^2 \leq \Delta_V T$.}
    at the cost of sampling each edge with
    probability $\frac{\log k}{k}$ instead of $\frac{1}{k}$. In
    particular, we improve our performance on a hub graph from
    sampling at rate $1$ to rate $\frac{\log T}{\sqrt{T}}$, while on
    graphs with less triangle-heavy vertices, we can achieve the bound
    of a $\frac{1}{T^{3/2}}$ sampling rate.

    As we only want to use those sampling rates appropriate to the scales of vertex actually present in the graph (as our performance depends on the greatest scale), we parameterize our algorithm by $\omega$, the minimum weight we will put on a vertex, which is $\min \left\lbrace \frac{T^2}{\sum_vT_v^2}, \sqrt{k}\right\rbrace$.

    Our algorithm conceptually runs as in two passes over the stream. First, for each vertex $v$ of the graph, we calculate an estimate $\mathcal{T}_v \in \left\lbrace 0, \dots,  \left\lceil \log\frac{\sqrt k}{\omega} \right\rceil \right\rbrace \cup \lbrace \lo \rbrace$ of $\left\lceil \log \frac{T}{\omega T_v} \right\rceil$, with $\mathcal{T}_v = \lo$ when $T_v$ is believed to be $< \frac{\wt{T}}{\sqrt k}$, and 0 when it is believed to be $> \frac{T}{\omega}$. In the second pass, we estimate $T_L$, the number of triangles $t = (u,v,w)$ s.t. $\mathcal{T}_u = \mathcal{T}_v = \mathcal{T}_w = \lo$, and $T_H$, the number of triangles using at least one vertex $v$ s.t. $\mathcal{T}_v \neq \lo$.

    We will show that it is possible to perform both conceptual passes in one pass over the data, by only calculating those $\mathcal{T}_v$ which are needed for the second pass.

    \subsection{Instance Lower Bounds}
    %

    We recall our definition of \emph{instance-optimum} for a class of
    graphs.
    \restate{dfn:solving}
    \restate{dfn:ilb}

    We demonstrate that it is sufficient to show that an algorithm
    cannot distinguish between two distributions on subgraphs of
    $G$ with triangle counts separated by $\Omega(T)$.

  \define[Distinguishing]{dfn:distinguish}{Definition}{
    We say that an algorithm $\mathcal{A}$ can \emph{distinguish} two
    random graph distributions $\mathcal{G}_1$ and $\mathcal{G}_2$ if there 
    exists $f$ such that, for a pair of draws $G_1$ and $G_2$ from
    these distributions, and any 
    relabelling of the vertices of $G_1$ and $G_2$, 
    $\Pr[f(\mathcal{A}(G_1)) = 1] \geq 3/4$
    and $\Pr[f(\mathcal{A}(G_2)) \not= 1] \geq 3/4$.
  }
  \state{dfn:distinguish}
 
  \define{lem:ibounddist}{Lemma}{%
    Let $\mathcal{A}$ be an algorithm that solves triangle counting 
    for a graph $G$ with $S$ space/samples and $(\epsilon,1/10)$ error.
    Then, for any two distributions $\mathcal{G}_1, \mathcal{G}_2$ 
    on subgraphs $G_1$ and $G_2$ of $G$, and $C$ such that $T(G_1) > 
    C + \epsilon T(G)$ with $\frac{9}{10}$ probability and $T(G_2) < C - 
    \epsilon T(G) $ with $\frac{9}{10}$ probability,
    $\mathcal{A}$ can distinguish them.
  } %
  \state{lem:ibounddist}%

  \subsubsection{Heavy Edges Graph}
  \define[Heavy Edges Graph]{dfn:hegraph}{Definition}{%
    The heavy edges graph $D_{r,d}$ consists of $r$ copies of the following graph: $d$ disjoint edges $\lbrace u_{2i}u_{2i+1}\rbrace_{i = 0}^{d-1}$, one of which has both ends connected to a further $d$ vertices $\lbrace v_i\rbrace_{i=0}^{d-1}$.
  } 
  \state{dfn:hegraph}
  \define[Heavy Edges Lower Bound]{thm:edgelb}{Theorem}{%
    For a sufficiently small constant $\epsilon$, \[\instopt(D_{r,d},\epsilon,1/10) = \Omega\left(d\right)\] bits in the insertion-only model.
  } 
  \state{thm:edgelb}
  We use a reduction shown in \cite{BOV13}.  We reduce from the indexing problem $\text{Index}_n$ to the problem of distinguishing an $rd$-triangle subset of $D_{r,d}$ and a 0-triangle subset.

  $\text{Index}_n$ is defined as follows: Alice has a binary vector $w$ of length $n$, and Bob has an index $x \in \lbrack n\rbrack$. Alice must send a message to Bob such that Bob can determine $w\lbrack x\rbrack$. By \cite{CCK10}, the randomized communication complexity of this problem is $\Omega(n)$. 

  Alice can encode $w$ in a subgraph of $D_{1,d}$ by, for each $i$ in $\lbrace0, \dots, n -1\rbrace$, including $u_{2i}u_{2i+1}$ iff $w_i = 1$. Bob then adds the $d$ vertices $\lbrace v_i\rbrace_{i=0}^{d-1}$ to the graph, connecting each of them to $u_{2x}$ and $u_{2x+1}$. The resulting graph will have $d$ triangles if $w_x = 1$ and 0 otherwise. 

  This result extends to arbitrary $r$ by letting Alice and Bob each repeat their encoding $r$ times.

\input{hedge.tex}

  \subsubsection{Hubs Graph}
  \define[Hubs Graph]{dfn:hubgraph}{Definition}{%
    The hubs graph $H_{r,d}$ consists of a single vertex $v$ with $2rd$ incident edges, and $d$ edges connecting disjoint pairs of $v$'s neighbors to form triangles. 
  }
  \state{dfn:hubgraph}
  \define[Hubs Lower Bound]{thm:hublb}{Theorem}{%
    For a sufficiently small constant $\epsilon$, \[\instopt(H_{r,d},\epsilon,\delta) = \Omega\left(r\sqrt{d}\right)\] bits in the insertion-only model.
  }
  \state{thm:hublb}
  For a single hub, we consider the following communication problem: there are $2rd$ people who might go on a trip. Alice knows who is going on the trip, and Bob knows of $d$ couples among them. Alice must send a message to Bob that lets him determine whether the couples are going on the trip or not. Intuitively, this should require $\Omega(r\sqrt{d})$ communication, as if Bob learns fewer than $r\sqrt{d}$ of the people who are going on the trip, he is likely to not learn of any of the couples. We can then consider the hub edges as the identities of of people who are going on the trip, and the other edges as the identities of the couples.

  Formally, we use a reduction from the Boolean Hidden Matching problem~\cite{BJK04,Kerenidis:2006,GKKRd07}, in particular the variant set out in~\cite{GKKRd07}.
\input{hub.tex}

  \subsubsection{Independent Triangles}
  \define[Independent Triangles Graph]{dfn:itgraph}{Definition}{%
    The independent triangles graph $I_n$ consists of $n$ vertex-disjoint triangles.
  }
  \state{dfn:itgraph}
  \define[Independent Triangles Lower Bound]{thm:itlb}{Theorem}{%
    For a sufficiently small constant $\epsilon$, \[
    \instopt(I_n,\epsilon,1/10) = \Omega\left(n^\frac{1}{3}\right)\]
    samples in the sampling model.
  }
  \state{thm:itlb}

  We prove this by showing that we need to sample edges with
  probability $\Omega\left(\frac{1}{n^{2/3}}\right)$ to achieve a
  constant chance of sampling \emph{any} triangle. 
  
  Intuitively, it seems that this should be sufficient---a sampling
  algorithm which fails to sample any triangles should not be able to 
  count the number of triangles in a graph. However, we
  have been unable to prove this in the general case, as for 
  arbitrary graphs, finding two subgraphs with a large enough 
  separation in triangle count that cannot be distinguished
  without sampling any triangles turns out to be difficult. For 
  instance, the pair would need to have the same number of edges,
  as otherwise they could be distinguished simply by edge counting,
  and the same number of wedges, as otherwise they could be distinguished
  by an algorithm that samples wedges but not triangles.
  
  In this specific case, however, we can show that there are two 
  different (distributions on) subgraphs which satisfy this
  requirement.
  
  Conditioned on failing to sample any
  triangles (and therefore, in $I_n$, any cycles), we show that, for a
  random 2-coloring $\chi$, the distributions of the following two
  graphs under the sampling scheme are identical:

  \begin{align*}
    I_1 & = (V(I_n), \lbrace uv \in E(I_n) | \chi(u) = \chi(v) \rbrace)\\
    I_2 & = (V(I_n), \lbrace uv \in E(I_n) | \chi(u) \not= \chi(v) \rbrace)
  \end{align*}

  Which, as $I_1$ has a constant fraction of the triangles, and $I_2$ has none of them, proves our theorem.

\input{indep.tex}

  \subsubsection{$G_{n,p}$}
  \define[$G_{n,p}$ Lower Bound]{thm:gnplb}{Theorem}{%
    There exists a constant $C$ such that, provided $p \geq \frac{C}{n}$, and for sufficiently small constant $\epsilon$, \[\instopt(G_{n,p}, \epsilon, 1/10) = \Omega\left(\frac{1}{p}\right)\] samples.
 }
 \state{thm:gnplb}
We make use of the same pair of colorings as in the proof of the independent triangles lower bound, but as our triangles are no longer guaranteed not to intersect, we prove that we need $\Omega\left(\frac{1}{p}\right)$ samples to sample \emph{any} cycles from the graph. Subject to our sample being acyclic, it will be identically distributed between both colorings, despite their differing by a constant fraction of the graph's triangles, and so it follows that we will not be able to count triangles.

  \subsection{Instance-Optimality for Triangle-Dependent Sampling Algorithms}
  We now present a lower bound which applies to all graph instances, but only a subclass of sampling algorithms, specifically those which depend on actually sampling at least one triangle to find a $(1 \pm \epsilon)$ estimate of the number of triangles in a graph. 

\define{dfn:triangledep}{Definition}{%
  Let $\mathcal{A}$ be a sampling algorithm for counting triangles. We say that $\mathcal{A}$ is a \emph{triangle-dependent sampling algorithm} if, for all graphs $G$, $\mathcal{A}(G)$ depends only on the set of triangles sampled by $\mathcal{A}$.
}
\state{dfn:triangledep}
Note that this definition encompasses all existing sampling
algorithms for triangle counting.

  Note that this definition encompasses all the strategies for counting triangles by edge sampling mentioned earlier, all of which depend on sampling edges by some strategy and then weighting the triangles sampled this way.

  Our method is to show, that each of the three parameters in our graph are, in any graph, necessary for sampling triangles.  This requires extending the definition of $\Delta_E$ and $\Delta_V$ to allow excluding $\epsilon T$ of the ``heaviest'' vertices or edges. This is necessary because a graph may, for instance, have a single edge with $\epsilon/2$ triangles, which contributes to the variance of the estimator, but which will not prevent the algorithm accurately estimating the triangle count if it is not sampled.

  \define{dfn:dgepsilon}{Definition}{%
    For any graph $G$, $\epsilon > 0$, let the vertices $v \in V(G)$ be ordered as $(v_i)_{i\geq 0}$ in descending order of $T_v$, and the edges $e \in E(G)$ be ordered as $(e_i)_{i\geq 0}$ in descending order of $T_e$. Then let $H_V$, $H_E$ be the maximal prefixes of $(v_i)_{i\geq 0}$,  $(e_i)_{i\geq 0}$ such that $\sum_{v\in H_V}T_v$, $\sum_{e\in H_E}T_v \leq \epsilon T$. We define $\Delta_{V,\epsilon}(G) = \max_{v \not\in H_V} T_v, \Delta_{E,\epsilon}(G) = \max_{e \not\in H_E} T_e$.

    When the graph meant is unambiguous, we will omit the parameter $G$.
  }
  \state{dfn:dgepsilon}

  \restate{thm:tdbound}

  Analyzing our algorithm in terms of the variance does not allow us to reach this bound. However,
  without altering the algorithm itself, we can refine the analysis by ``cutting off'' a small number 
  (less than $\epsilon T$ times a small constant) of the heaviest vertices and edges. This works 
  because their contribution to the estimate $\overline{T}$ is always positive, and so we may split 
  $\overline{T}$ into $\overline{T}_\epsilon$ (representing $\overline{T}$ less the contribution of
  these edges and vertices), which we bound by Chebyshev's inequality as usual, and $(\overline{T} -
  \overline{T}_\epsilon)$, which we bound by Markov's inequality. This allows ups to replace the 
  $\Delta_V, \Delta_E$ terms in our bound with $\Delta_{V,\epsilon}, \Delta_{E,\epsilon}$, so that, for constant $\epsilon$, we 
  can match the lower bound up to a log factor and by a constant\footnote{The distinction here is 
    that multiplying $\epsilon$ by a constant can cause a non-constant change in $\Delta_{E,\epsilon}$}
    factor in $\epsilon$.

    \define[Refined triangle estimation upper bound]{thm:refined}{Theorem}{%
      If $\wt{T} \leq T$, and $\epsilon > 0$, our
      algorithm obtains an $(\epsilon,\delta)$ approximation to $T$ while
      keeping \[O\left( \frac{m \log
        \frac{1}{\delta}}{\epsilon^2}\left(\frac{1}{\wt{T}^{2/3}} +
        \frac{\sqrt{\Delta_{V,\epsilon/24}}\log \wt{T}}{\wt{T}} +\frac{\Delta_{E,\epsilon/24}\log
      \wt{T}}{\wt{T}}\right)\right)\]
      edges. If $\wt{T} > T$, the algorithm either determines that
        $\wt{T} > T$ or obtains a $(1 \pm \epsilon)$ approximation
          to $T$ with probability $1 - \delta$.
}
\state{thm:refined}

\subsection{Algorithm for General Constant-Size Subgraphs}
We show that the algorithm in this paper can, with small modifications, be generalized to count the number of copies of an arbitrary constant-size subgraph $A$ in $G$. As in the triangle case, we start with the algorithm in which every vertex is sampled with probability $\frac{1}{\sqrt k}$, and edges between sampled vertices are kept. This will give an estimator with variance:

\begin{align*}
  O\left(\sum_{l = 1}^s \left(C_lk^\frac{l}{2}\right)\right)
\end{align*}
Where we define $C_i$ as follows: for any $S \subseteq V(G)$, we define $\alpha(S)$ defined as the set of copies of $A$ that use all the vertices in $S$ and $M_S$ as $|\alpha(S)|$. $C_i$ is then defined as $\sum_{|S| = i} M_S$. Note that in the case where $A$ is a triangle, $C_1 = \sum_v T_v^2$, $C_2 = \sum_e T_e^2$, and $C_3 = T$.

As in the triangle case, we will eliminate the $C_1$ term, by estimating the number of subgraphs involving $v$ at each $v$ and reducing the ``weight'' applied to $v$ accordingly. In the triangle case, this increased the $C_3$ (equivalently, $T$) term---in the general case it will increase the $C_l$ term for every $l \geq 3$. This is because, when we put less weight on a vertex, for any pair of subgraphs $a_1$, $a_2$ which intersect at $\geq 2$ edges adjacent to this vertex, the event of sampling $a_1$ will become \emph{more} correlated with the event of sampling $a_2$.

Consequently, our new variance will end up depending on $C_1$ at every term of the sum, giving us variance:

\begin{align*}
  M^2 + \sum_{l = 2}^s C_l\left(k^\frac{l}{2} + k\left(\frac{C_1^+}{M^2}k\right)^{l - 2}\right)
\end{align*}
This then gives us our general subgraphs result.
\restate{thm:general}

\section*{Acknowledgements}

The authors would like to thank David Woodruff for a helpful pointer
to the Boolean Hidden Matching problem~\cite{BJK04,Kerenidis:2006,GKKRd07}.

\bibliographystyle{alpha}
\bibliography{refs}

\newpage
\appendix

    \section{Algorithm}
    \subsection{Definitions}
    Let $G$ be a graph, where we receive $E(G)$ as a stream of edges, with $m = |E(G)|$. Edges in $E(G)$ will always be treated as undirected.

    $T(G)$ is the number of triangles in $G$. For any $v \in V(G)$, let $\tau(v)$ be the set of triangles in $G$ involving $v$, and $T_v = |\tau(v)|$. For any $e \in E(G)$, let $T_e$ be the number of triangles involving $e$. Then, let $T_V(G) = \sum_v T_v^2$ and $T_E(G) = \sum_e T_e^2$. Where the graph $G$ meant is unambiguous, we will omit the explicit parametrization by $G$. Note that $T_V \leq T\Delta_V$ and $T_E \leq T\Delta_E$, respectively.

    $k \in \lbrack 0,1\rbrack$ is our sampling parameter. We will show that the expected number of edges stored by the algorithm is $O\left(\frac{m\log k}{k}\right)$.

    Let $\wt{T}$ be a proposed lower bound on $T$, and $T_V^+$ an actual upper bound on $T_V$. If  $T_V \leq T_V^+$ and $T \geq \wt{T}$, we want to be able to accurately estimate the number of triangles in the graph. If $T_V \leq T_V^+$ and  $T < \wt{T}$, we want to be able to detect that $T < \wt{T}$. Given these two parameters, we will define $\omega \coloneqq \min \left\lbrace \frac{(\wt{T})^2}{T_V^+},\sqrt{k}\right\rbrace$. $\omega$ will be the minimum ``weight'' we put on a vertex, and as we will use weights from $\omega$ to $\sqrt{k}$, the total number of sampling schemes we use will be $\log \frac{\sqrt{k}}{\omega}$. This means that the bound on the expected number of edges stired by the algorithm can in fact be reduced to $O\left(\frac{m}{k} \log \frac{\sqrt{k}}{\omega}\right)$.

    \subsection{First pass}
    \subsubsection{Outline}
    In our first pass over the graph, we attempt to estimate the ``correct'' weight to put on each vertex $v$, $\frac{T}{T_v}$. As we do not know $T$, we use our lower bound $\wt{T}$, which will allow us to bound the variance of our final estimate in terms of $\wt{T}$ and therefore $T$. As we do not have direct access to $T_v$, we look at the number of triangles that would be counted at $v$ if it were assigned a given weight.

    We consider $\log \frac{\sqrt{k}}{\omega}$ possible weights that could be assigned to $v$. For each $h \in \left \lbrace 0,\dots, \left\lceil \log \frac{\sqrt{k}}{\omega} \right\rceil \right \rbrace$, we consider the weight $\omega2^h$ ($\omega$, therefore, being the minimum weight we will assign any vertex). Let $\mathcal{T}_v$ be our estimate of the ``correct'' choice of $h$. But what should this be? Letting $X_v^{(h)}$ be the number of triangles counted at $v$ when it is assigned weight $2^h$, the expected value of $X_v^{(h)}$ will be $\frac{2^{2h}T_v}{k^2}$. So if we want $\omega2^h$ to be $\frac{\wt{T}}{T_v}$, the expectation of $X_v^{(h)}$ should be $\frac{\omega \wt{T} 2^h}{k^2}$. We will therefore define $\mathcal{T}_v$ to be the least $h$ such that $X_v^{(h)}$ is at least this high. If $X_v^{(h)}$ is not this high for any $h \in  \left \lbrace 0,\dots, \left\lceil \log \frac{\sqrt{k}}{\omega} \right\rceil \right \rbrace$, we will conclude that $X_v \leq \frac{\wt{T}}{\sqrt{k}}$, and so our algorithm should treat it as a ``light'' vertex, and so we set $\mathcal{T}_v = \lo$.

    In order to attain the bounds we need on the distance of $\mathcal{T}_v$ from $\left\lceil \log \frac{T}{\omega T_v} \right\rceil$, we will run this procedure twice, taking the lowest value of $\mathcal{T}_v$.

    \subsubsection{Procedure}
    For $i = 1,2$, and for each $v \in V(G)$, define $\mathcal{T}_{v,i}$ as follows:

    Let $r_{D,i} : E \rightarrow \lbrack 0,1\rbrack$ be a uniformly random hash function.

    Let $d_i : V \rightarrow \lbrace 0,1\rbrace$ be a random hash function such that $\forall v \in V$:

    \begin{align*}
      d_i(v) = \begin{cases}
        1 & \mbox{With probability $\frac{1}{\sqrt{k}}$.}\\
        0 & \mbox{Otherwise.}
      \end{cases}
    \end{align*}

    For each $v \in V(G)$, $t = (u,v,w)$ a triangle in $G$, and $h \in \left\lbrace 0, \dots, \left\lceil \log \frac{\sqrt k}{\omega} \right\rceil\right\rbrace$, define $X_{v, i}^{(h,t)}$ to be 1 if:

    \begin{align*}
      d_i(u) & = 1\\
      d_i(w) & = 1\\
      r_{D,i}(vu) & < \frac{\omega 2^h}{\sqrt k}\\
      r_{D,i}(vw) & < \frac{\omega 2^h}{\sqrt k}
    \end{align*}

    And 0 otherwise. Then let $X_{v,i}^{(h)} = \sum_{t \in \tau(v)} X_{v, i}^{(h,t)}$. We then define \\$H_{v,i} = \left\lbrace h \in \left\lbrace 1, \dots, \left\lceil \log \frac{\sqrt k}{\omega} \right\rceil\right\rbrace \middle| X_{v,i}^{(h)} \geq \frac{\omega \wt{T}2^h}{k^2}\right\rbrace$. $\mathcal{T}_{v,i}$ is then defined as follows:

    \begin{align*}
      \mathcal{T}_{v,i} = \begin{cases}
        \min H_{v,i} & \mbox{If $H_{v,i} \not= \emptyset$.}\\
        \lo & \mbox{Otherwise.}
      \end{cases}
    \end{align*}

    Then, for each $v \in V(G)$, we define $\mathcal{T}_v$ to be $\lo$ if $\mathcal{T}_{v,1}$ and $\mathcal{T}_{v,2}$ are $\lo$, and otherwise to be the smallest numerical value amongst $\mathcal{T}_{v,1}, \mathcal{T}_{v,2}$.

    \subsection{Second Pass}
    \subsubsection{Outline}
    We now use the scale estimates $\mathcal{T}_v$ to determine our strategy for estimating $T$. If, for a vertex $v$, $\mathcal{T}_v = \lo$, we believe that $T_v \leq \frac{\wt{T}}{\sqrt{k}}$, so we use our na\"ive sampling method to estimate $T_L$, the number of triangles in $G$ which use only such vertices. Otherwise, we sample $v$ with probability $2^{-\mathcal{T}_v}$, and if we sample it, construct an estimate $\overline{T}_v$ of $T_v$ by assigning $v$ weight $2^{\mathcal{T}_v}$. As this could lead to a triangle being counted up to three times (if all three of its vertices $v$ had $\mathcal{T}_v \not = \lo$, it would be counted three times), we give each triangle either weight $1, \frac{2}{3}$, or $\frac{1}{3}$, depending on how many of its vertices $v$ have $\mathcal{T}_v \not = \lo$. This then lets us report our estimate $\overline{T_H}$ of $T_H$ as $\sum_v \overline{T}_v$.

    \subsubsection{Splitting the Graph}
    Let $V_{L} = \lbrace v \in V | \mathcal{T}_v = \lo\rbrace$, and let $G_L$ be the subgraph induced by $V_L$. Then we define $T_L$ as the number of triangles in $G_L$, and $T_H = T - T_L$.

    We will compute estimates $\overline{T_L},\overline{T_H}$ of $T_L,T_H$, and estimate $T$ as $\overline{T} = \overline{T_L} + \overline{T_H}$.

    \subsubsection{Estimating \texorpdfstring{$T_L$}{T\_L}}
    We will estimate $T_L$ by sampling the vertices of $V_L$ with probability $\frac{1}{\sqrt k}$ each and calculating the number of triangles in the resulting graph.

    Let $c : V \rightarrow \lbrace 0,1\rbrace$ be a random hash function such that $\forall v \in V$:

    \begin{align*}
      c(v) = \begin{cases}
        1 & \mbox{With probability $\frac{1}{\sqrt k}$.}\\
        0 & \mbox{Otherwise.}
      \end{cases}
    \end{align*}

    Let $V_L' = \lbrace v \in V_L | c(v) = 1\rbrace$, and let $G_L'$ be the subgraph of $G$ induced by $V_L'$. Then $\overline{T_L}$ is the number of triangles in $G_L'$, multiplied by $k^\frac{3}{2}$.

    \subsubsection{Estimating \texorpdfstring{$T_H$}{T\_H}}
    We will estimate $T_H$ by considering every vertex in $V\setminus V_L$ separately. We will achieve this by sampling vertices $v$ with probability $2^{-\mathcal{T}_v}$, and then sampling edges incident to $v$ with probability proportional to $2^{\mathcal{T}_v}$.

    Let $h : V \rightarrow \lbrace 0,1\rbrace$ be a random hash function such that $\forall v \in V$:

    \begin{align*}
      h(v) = \begin{cases}
        i & \mbox{With probability $\frac{1}{\omega2^i}$ for each $i \in \lbrace0,\dots,\left\lceil\log \frac{\sqrt k}{\omega}\right\rceil\rbrace$.}\\
        -\infty & \mbox{Otherwise.}
      \end{cases}
    \end{align*}

    And let $r_C : V \rightarrow \lbrack 0,1\rbrack$ be a uniformly random hash function.

    Then, for each $v\in V_H = V\setminus V_L$, we allow $v$ to contribute to $\overline{T_H}$ iff $h(v) = \mathcal{T}_v$. If it does, we calculate its contribution $\overline{T_v}$ as follows:

    We count a triangle $t = (u,v,w)$ iff:

    \begin{align*}
      c(u) & = 1\\
      c(w) & = 1\\
      r_C(vu) & < \frac{\omega 2^{\mathcal{T}_v}}{\sqrt k}\\
      r_C(vw) & < \frac{\omega 2^{\mathcal{T}_v}}{\sqrt k}
    \end{align*}

    Then, for each such triangle, we add $\frac{k^2}{\omega 2^{\mathcal{T}_v}} \times \frac{1}{|\lbrace x \in \lbrace u,v,w\rbrace | \mathcal{T}_x \not = \lo\rbrace |}$ to $\overline{T_v}$. (with the second term then compensating for the fact that a triangle could potentially be counted at multiple different vertices) 

    We then define $\overline{T_H} = \sum_{v \in V_H,\ h(v) = \mathcal{T}_v} \overline{T_v}$.

    \subsection{Final Output}
    We then output our estimate $\overline{T}$ of $T$ as $\overline{T} = \overline{T_L} + \overline{T_H}$.

    \section{Analysis}
    \subsection{First Pass}
    The following lemmas will bound the deviation of $\mathcal{T}_v$ from its desired value---$\log \frac{\wt{T}}{\omega T_v}$ if $T_v \geq \frac{\wt{T}}{\sqrt{k}}$, and $\lo$ otherwise.

    \begin{lemma} 
      \label{xvariance}
      For any $v\in V(G)$, and for $i \in \lbrack 2\rbrack$, let $X_{v,i}^{(h)}$ be as defined previously. Then, $\mathbb{E}\left\lbrack X_{v,i}^{(h)} \right\rbrack = T_v \frac{\omega^22^{2h}}{k^2}$ and $\var\left(X_{v,i}^{(h)}\right) \leq T_v \frac{\omega^22^{2h}}{k^2} + \sum_w T_{vw}^2 \frac{\omega^32^{3h}}{k^3}$.
    \end{lemma}
    \begin{proof}
      For each triangle $t = (u,v,w)$, $X_{v,i}^{(h,t)} = 1$ iff $d_i(u) = d_i(w) = 1$ and $r_{D,i}(vu), r_{D,i}(vw) < \frac{\omega2^h}{\sqrt k}$, which happens with probability $\frac{\omega^22^{2h}}{k^2}$. $X_{v,i}^{(h)} = \sum_{t \in \tau(v)} X_{v,i}^{(h,t)}$, so $\mathbb{E}\left\lbrack X_{v,i}^{(h)}\right\rbrack = T_v \frac{\omega^22^{2h}}{k^2}$.

      Then, to bound the variance, we start by bounding $\mathbb{E}\left\lbrack \left(X_{v,i}^{(h)}\right)^2\right\rbrack = \sum_{t_1,t_2 \in \tau(v)}\mathbb{E}\left\lbrack X_{v,t_1}^{(h)}X_{v,t_2}^{(h)}\right\rbrack$. We split the terms $\mathbb{E}\left\lbrack X_{v,t_1}^{(h)}X_{v,t_2}^{(h)}\right\rbrack$ by the value of $l = |V(t_1) \cap V(t_2)|$. (noting that $l \in \lbrace 1,2,3\rbrace$)

      Then, $X_{v,t_1}^{(h)}X_{v,t_2}^{(h)} = 1$ implies that, for the $5 - l$ vertices  $u \in V(t_1) \cup V(t_2) \setminus \lbrace v\rbrace$, $d(u) =1$, and for the $5 - l$ edges $e \in \left\lbrace vu \middle| u \in V(t_1) \cup V(t_2) \setminus \lbrace v\rbrace \right\rbrace$, $r_D(e) < \frac{\omega 2^h}{\sqrt k}$. So $\mathbb{E}\left\lbrack X_{v,a_1}^{(h)}X_{v,a_2}^{(h)}\right\rbrack \leq \left(\frac{\omega2^h}{k}\right)^{5 - l}$. 

      There are no more than $T_v^2$ such pairs for $l = 1$, no more than $\sum_w T_{vw}^2$ for $l = 2$, and exactly $T_v$ for $l = 3$, which gives us:

      \begin{align*}
        \var\left(X_{v,i}^{(h)}\right) & = \mathbb{E}\left\lbrack \left(X_{v,i}^{(h)}\right)^2\right\rbrack - \mathbb{E}\left\lbrack \left(X_{v,i}^{(h)}\right)\right\rbrack^2\\
        & \leq  T_v^2 \frac{\omega^42^{4h}}{k^4} + \sum_w T_{vw}^2 \frac{\omega^32^{3h}}{k^3} + T_v \frac{\omega^22^{2h}}{k^2} - \left(T_v \frac{\omega^22^{2h}}{k^2}\right)^2\\
        & =  T_v \frac{\omega^22^{2h}}{k^2} + \sum_w T_{vw}^2 \frac{\omega^32^{3h}}{k^3}
      \end{align*}

    \end{proof}

    \begin{corollary}
      \label{xdeviance}
      $\mathbb{P}\left\lbrack X_{v,i}^{(h)} \leq \frac{T_v}{2}\frac{\omega^22^{2h}}{k^2}\right\rbrack \lesssim \frac{k^2}{\omega^2 2^{2h}T_v} + \sum_w \frac{T_{vw}^2}{T_v^2} \frac{k}{\omega2^h}$.
    \end{corollary}
    \begin{proof}
      $X_{v,i}^{(h)} \leq \frac{T_v}{2}\frac{\omega^22^{2h}}{k^2}$ implies that $\left| X_{v,i}^{(h)} - \mathbb{E}\left\lbrack X_{v,i}^{(h)}\right\rbrack \right| \geq \frac{1}{2}\mathbb{E}\left\lbrack X_{v,i}^{(h)} \right\rbrack$, and so by Chebyshev's inequality it occurs with probability $\lesssim \frac{\var\left(X_{v,i}^{(h)}\right)}{\mathbb{E}\left\lbrack X_{v,i}^{(h)} \right\rbrack^2}$.
    \end{proof}

    \begin{lemma}
      \label{loerror}
      If $T_v \geq 2\frac{\wt{T}}{\sqrt{k}}$, then $\mathbb{P}\left\lbrack \mathcal{T}_v = \lo \right\rbrack \lesssim \frac{k}{T_v} + \sum_w T_{vw}^2 \frac{\sqrt k}{T_v^2}$.

    \end{lemma}
    \begin{proof}
      $\mathcal{T}_v = \lo$ implies that $X_{v,i}^{(h)} < \frac{\omega 2^h\wt{T}}{k^2} $ for all $h \in \left\lbrace 1,\dots,\left\lceil \log \frac{\sqrt k}{\omega}\right\rceil\right\rbrace$ and $i \in \lbrack 2 \rbrack$, and so in particular $X_{v,1}^{\left(\left\lceil\log \frac{\sqrt k}{\omega}\right\rceil\right)} < \frac{\wt{T}}{k^\frac{3}{2}} \leq \frac{T_v}{2k} \leq \frac{T_v}{2}\frac{\omega^2 2^{2\left\lceil\log \frac{\sqrt k}{\omega}\right\rceil}}{k^2}$.

      So the result follows by using Corollary \ref{xdeviance} with $h = \left\lceil\log \frac{\sqrt k}{\omega}\right\rceil$.
    \end{proof}

    \define{hierror}{Lemma}{%
      $\forall v \in V, \mathbb{P}\left\lbrack \mathcal{T}_v \not = \lo \right\rbrack \lesssim \frac{\sqrt{k} T_v}{\wt{T}}$.
    }
    \state{hierror}

    \begin{proof}
      $\mathcal{T}_v \not = \lo$ implies that there is at least one $h\in  \left\lbrace 1,\dots,\left\lceil \log \frac{\sqrt k}{\omega}\right\rceil\right\rbrace$ and $i \in \lbrack 2 \rbrack$ such that $X_{v,i}^{(h)} \geq \frac{\omega 2^h\wt{T}}{k^2}$. By Lemma \ref{xvariance}, $\mathbb{E}\left\lbrack X_{v,i}^{(h)} \right\rbrack = T_v \frac{\omega^22^{2h}}{k^2}$, so by Markov's inequality this occurs with probability $\leq \frac{T_v}{\wt{T}}\omega2^h$. So the probability that it holds for any $h, i$ is $\leq 2\sum_{h=\left\lceil \log \frac{\sqrt k}{\omega}\right\rceil}^1 \frac{T_v}{\wt{T}}\omega2^h \lesssim \sum_{i=0}^\infty\frac{T_v\sqrt{k}}{\wt{T}} 2^{-i}\lesssim \frac{\sqrt{k} T_v}{\wt{T}}$.
    \end{proof}

    \begin{lemma}
      \label{derror}
      $\forall v \in V, l \geq 2, \mathbb{P}\left\lbrack \mathcal{T}_v > \left\lceil\log \frac{\wt{T}}{\omega T_v}\right\rceil + l\right\rbrack \lesssim \left(\frac{k^2T_v}{2^{2l} \left(\wt{T}\right)^2} + \sum_w T_{vw}^2 \frac{k}{2^{l}\wt{T}T_v}\right)^2$.
    \end{lemma}

    \begin{proof}
      $\mathcal{T}_v > \left\lceil\log \frac{\wt{T}}{\omega T_v}\right\rceil + l$ implies that $\forall i\in\lbrack 2\rbrack, X_{v,i}^{\left(\left\lceil\log \frac{\wt{T}}{\omega T_v}\right\rceil + l\right)} < \frac{2^l \left(\wt{T}\right)^2}{k^2T_v}$. By applying Corollary \ref{xdeviance} with $h = \left\lceil\log \frac{\wt{T}}{\omega T_v}\right\rceil + l$, and using the independence of $X_{v,1}^{(h)}, X_{v,2}^{(h)}$, this happens with probability $\lesssim \left(\frac{k^2T_v}{2^{2l} \left(\wt{T}\right)^2} + \sum_w T_{vw}^2 \frac{k}{2^{l}\wt{T}T_v}\right)^2$.
    \end{proof}

    \subsection{Second Pass}
    Making use of the lemmas from the previous section, we bound the error of our estimators $\overline{T_L}$ and $\overline{T_H}$, and thus of our final estimator $\overline{T}$.
    \subsubsection{\texorpdfstring{$\overline{T_L}$}{T\_L}}
    \begin{lemma}
      \label{tlexpec}
      $\mathbb{E}\left\lbrack \overline{T_L} | T_L \right\rbrack = T_L$.
    \end{lemma}
    \begin{proof}
      $T_L$ is the number of triangles in the random graph $G_L \subseteq G$, and $\overline{T_L}$ is $k^\frac{3}{2}$ times the number of triangles in the random graph $G_L' \subseteq G_L$. Any triangle $t = (u,v,w)$ in $G_L$ is in $G_L'$ iff $c(u) = c(v) = c(w) = 1$, which occurs with probability $\frac{1}{k^\frac{3}{2}}$. So the expected number of triangles in $G_L'$ is $\frac{T_L}{k^\frac{3}{2}}$, and so $\mathbb{E}\left\lbrack \overline{T_L} | T_L \right\rbrack = T_L$.
    \end{proof}
    \begin{lemma}
      \label{tlvar}
      $\var(\overline{T_L}) \lesssim Tk^\frac{3}{2} + kT_E + T\wt{T}$.
    \end{lemma}
    \begin{proof}
      Let $T_L'$ be the number of triangles in $G_L'$ (so $\overline{T_L} = k^\frac{3}{2}T_L'$). A triangle $t=(u,v,w)$ in $G$ is in $G_L'$ iff $c(u) = c(v) = c(w)$ and $\mathcal{T}_u=\mathcal{T}_v=\mathcal{T}_w = \lo$. We proceed by bounding $\mathbb{E}\left\lbrack T_L'^2\right\rbrack$. This will be equal to the sum over every ordered pair $(t_1,t_2)$ of the probability that both $t_1$ and $t_2$ are in $G_L'$. There are four scenarios to consider:
      \begin{description}
        \item[$t_1 = t_2 = (u,v,w)$:] Then both will be counted iff $t_1$ is, which requires $c(u) = c(v) = c(w)$, so this happens with probability $\leq \frac{1}{k^\frac{3}{2}}$. There are no more than $T$ such ``pairs'' in $G_L$, and so they contribute at most $\frac{T}{k^\frac{3}{2}}$ to $\mathbb{E}\left\lbrack T_L'^2\right\rbrack$.
        \item[$t_1$ and $t_2$ share one edge $vw$:] Let $u_1, u_2$ be the remaining two vertices. Then for both $t_1$ and $t_2$ to be in $G_L'$, it is necessary (but not sufficient) that:

          \begin{align*}
            c(u_1) & = 1\\
            c(u_2) & = 1\\
            c(v) & = 1\\
            c(w) & = 1
          \end{align*}

          So this holds with probability at most $\frac{1}{k^2}$. There are no more than $T_{e}^2$ such pairs in for each edge $e$ in $G_L$, so they contribute at most $\sum_e \frac{T_{e}^2}{k^2} = \frac{T_E}{k^2}$ to $\mathbb{E}\left\lbrack T_L'^2\right\rbrack$ in total.
        \item[$t_1$ and $t_2$ share only $v$:] 
          Let $u_1,w_1, u_2,w_2$ be the remaining four vertices. Then for both $t_1$ and $t_2$ to be in $G_L'$, it is necessary that:

          \begin{align*}
            c(v) & = 1\\
            c(u_1) & = 1\\
            c(w_1) & = 1\\
            c(u_2) & = 1\\
            c(w_2) & = 1
          \end{align*}

          So this holds with probability at most $\frac{1}{k^\frac{5}{2}}$. We now split such pairs into two categories. 
          \begin{description}
            \item[$T_v \leq \frac{2\wt{T}}{\sqrt{k}}$:] There are at most $\sum_{v, T_v \leq \frac{2\wt{T}}{\sqrt k}}T_v^2 \lesssim \frac{T\wt{T}}{\sqrt k}$ such pairs in $G_L$, and so this category contributes at most $\frac{T\wt{T}}{k^3}$ to $\mathbb{E}\left\lbrack T_L'^2\right\rbrack$.

            \item[$T_v > \frac{2\wt{T}}{\sqrt{k}}$:] In this case, we also use the bound on the probability that $\mathcal{T}_v = \lo$ from Lemma $\ref{loerror}$. This is $ \lesssim \frac{k}{T_v} + \sum_w T_{vw}^2 \frac{\sqrt k}{T_v^2}$, so the contribution to $\mathbb{E}\left\lbrack T_L'^2\right\rbrack$ from this category is $\lesssim \sum_v\frac{T_v^2}{k^\frac{5}{2}} \left(\frac{k}{T_v} + \sum_w T_{vw}^2 \frac{\sqrt k}{T_v^2}\right) = \frac{T}{k^\frac{3}{2}} + \frac{T_E}{k^2}$.
          \end{description}
        \item[$t_1$ and $t_2$ are disjoint:] Then for $t_1$ and $t_2$ to be in $G_L'$, it is necessary that all 6 vertices $v$ of $t_1,t_2$ have $c(v) = 1$, which happens with probability $\frac{1}{k^3}$. So then, letting $p_t$ be the probability that a triangle $t \in G$ is in $G_L$, the total contribution to $\mathbb{E}\left\lbrack T_L'^2\right\rbrack$ from this case will be at most $\sum_{t_1 \cap t_2 = \emptyset} \frac{p_{t_1} p_{t_2}}{k^3} \leq \frac{1}{k^3} \mathbb{E}\left\lbrack T_L \right \rbrack^2 = \mathbb{E}\left\lbrack T_L' \right \rbrack^2$.
      \end{description}

      By summing these together, as $\overline{T_L} = k^\frac{3}{2} T_L'$, this gives us:
      \begin{align*}
        \var(\overline{T_L}) & = k^3 \var(T_L')\\
        & = k^3(\mathbb{E}\left\lbrack T_L'^2\right\rbrack - \mathbb{E}\left\lbrack T_L'\right\rbrack^2)\\
        & \lesssim k^3\left(\frac{T}{k^\frac{3}{2}} +  \frac{T_E}{k^2} +  \frac{T\wt{T}}{k^3} + \mathbb{E}\left\lbrack T_L' \right \rbrack^2 - \mathbb{E}\left\lbrack T_L' \right \rbrack^2\right)\\
        & \lesssim Tk^\frac{3}{2} + kT_E + T\wt{T}
      \end{align*}
    \end{proof}

    \subsubsection{\texorpdfstring{$\overline{T_H}$}{T\_H}}
    \begin{lemma}
      \label{thexpec}
      $\mathbb{E}\left\lbrack \overline{T_H} | T_H \right\rbrack = T_H$.
    \end{lemma}
    \begin{proof}
      $\overline{T_H} = \sum_{v\in V_H} \overline{T_v}$, and $T_H$ is the number of triangles in $G$ that are not in $G_L$. If a triangle in $G$ is in $G_L$, it can never contribute to $\overline{T_v}$ for any $v$, as every vertex $v$ it uses will have $\mathcal{T}_v = \lo$. If a triangle in $G$ is not in $G_L$, it has $s \in \lbrace 1,2,3\rbrace$ vertices $u$ s.t. $\mathcal{T}_u \not = \lo$. It will therefore have $s$ vertices $v$ such that it can contribute to $T_v$. 

      At each of those vertices $v$, if the triangle is $t = (u,v,w)$, it will contribute $\frac{k^2}{\omega 2^{\mathcal{T}_v}s}$ to $\overline{T_v}$ iff:

      \begin{align*}
        h(v) & = \mathcal{T}_v\\
        c(u) & = 1\\
        c(w) & = 1\\
        r_C(vu) & < \frac{\omega2^{\mathcal{T}_v}}{\sqrt k}\\
        r_C(vw) & < \frac{\omega2^{\mathcal{T}_v}}{\sqrt k}
      \end{align*}

      This happens with probability $\frac{1}{\omega2^{\mathcal{T}_v}} \times \frac{1}{\sqrt k} \times \frac{1}{\sqrt k} \times \frac{\omega2^{\mathcal{T}_v}}{\sqrt k} \times \frac{\omega2^{\mathcal{T}_v}}{\sqrt k} = \frac{\omega2^{\mathcal{T}_v}}{k^2}$, so the expected contribution of $t$ to $\overline{T_v}$ is $\frac{1}{s}$. Therefore, as there are $s$ vertices where $t$ can contribute, its expected contribution to $\overline{T_H}$ is 1.

      Therefore,  $\mathbb{E}\left\lbrack \overline{T_H} | T_H\right\rbrack$ is precisely the number of triangles in $G$ that are not in $G_L$, which is $T_H$.
    \end{proof}

    \begin{lemma}
      \label{thvar}
      $\var(\overline{T_H}) \lesssim T\frac{k^2}{\omega} +  T_Ek + T^2$.
    \end{lemma}
    \begin{proof}
      We now care, for any triangle $t$, which vertex we are counting it at (and so which $T_v$ it may contribute to). We will therefore use $t^v$ to denote the triangle $t$, counted at the vertex $v$. We then define $Y_{t^v}$ as $\frac{k^2}{\omega 2^{\mathcal{T}_v}}$ if $t$ is counted at $v$ and 0 otherwise.

      Then, $\overline{T_v} \leq \sum_{t \in \tau(v)} Y_{t^v}$, as $Y_{t^v}$ is equal or greater to the contribution to $T_v$ from $t$. So with $Y = \sum_{v}\sum_{t \in \tau(v)} Y_{t^v}$, $\mathbb{E}\left\lbrack \overline{T_H}^2 \right\rbrack \leq \mathbb{E}\left\lbrack Y^2 \right\rbrack$. We will bound  $ \mathbb{E}\left\lbrack Y^2 \right\rbrack$ by bounding $\mathbb{E}\left\lbrack Y_{t_1^u}Y_{t_2^v} \right\rbrack$ for each pair $t_1^u, t_2^v$.

      We will bound the total contribution to $Y^2$ from each possible value of $l = |V(t_1) \cap V(t_2)|$, treating $l = 1$ as a special case, and treating $u = v$ and $u \not = v$ separately. 

      \begin{description}
        \item[$l \not = 1, u \not = v$:] $ Y_{t^v} > 0$ implies that $h(u) = \mathcal{T}_u$, $h(v) = \mathcal{T}_v$, which happens with probability $\omega^{-2}2^{-\mathcal{T}_u - \mathcal{T}_v}$. It also implies that $\forall w \in (V(t_1) \setminus \lbrace u\rbrace) \cup (V(t_2) \setminus \lbrace v\rbrace ), c(w) = 1$. There are $4 - l$ vertices in this set, so this happens with probability $\leq k^{-2 +\frac{l}{2}}$.  It also implies that $\forall e \in \left\lbrace uw \middle| uw \in E(t_1) \right\rbrace r_C(e) < \frac{\omega 2^{\mathcal{T}_u}}{\sqrt k}$ and $\forall e \in \left\lbrace vw \middle| vw \in E(t_2) \right\rbrace, r_C(e) < \frac{\omega 2^{\mathcal{T}_v}}{\sqrt k}$. These sets can overlap in at most one edge ($uv$), and each have size 2, so this happens with probability $\leq \left(\frac{\omega}{\sqrt k}\right)^{3} 2^{2\mathcal{T}_u + 2\mathcal{T}_v - \max\lbrace \mathcal{T}_u, \mathcal{T}_v\rbrace}$. Furthermore, the overlap in $uv$ can only occur if $u \in V(t_2)$ and $v\in V(t_1)$, and in this case we will also need $c(u) = 1$ and $c(v) = 1$. So this either reduces the probability by a factor of $\frac{\omega2^{\max \lbrace \mathcal{T}_u, \mathcal{T}_v\rbrace}}{\sqrt k}$ or $\frac{1}{k}$, and so in either case by at least a factor of $\frac{\omega2^{\max \lbrace \mathcal{T}_u, \mathcal{T}_v\rbrace}}{\sqrt k}$.

          So as these three conditions are independent, and multiplying the probability that $Y_{t_1^u}, Y_{t_2^v} > 0$ by the values they take when they are, $\mathbb{E}\left\lbrack Y_{t_1^u}Y_{t_2^v}\right\rbrack \leq k^\frac{l}{2}$. So the contribution to $\mathbb{E}\left\lbrack Y^2\right\rbrack$ from such pairs is $\lesssim T^2$ for $l = 0$ (as there are at most $T^2$ such pairs in $G$), $T_Ek$ for $l = 2$ (as there are at most $\sum_e T_e^2 = T_E$ such pairs in $G$), and $Tk^\frac{3}{2}$ for $l = 3$ (as any such ``pair'' has $t_1 = t_2$, so there are at most $T$ of them).

        \item[$l \not = 1, u = v$:] Note that as $u = v$, $l \not = 0$.  $Y_{t^v} > 0$ implies that  $h(v) = \mathcal{T}_v$, which happens with probability $\omega^{-1}2^{- \mathcal{T}_v}$. It also implies that  $\forall w \in V(t_1) \setminus \lbrace u\rbrace \cup V(t_2) \setminus \lbrace v\rbrace, c(w) = 1$. There are $5 - l$ vertices in this set, so this happens with probability $\leq k^{-\frac{5 - l}{2}}$. It also implies that $\forall e \in \left\lbrace uw \middle| uw \in E(t_1) \right\rbrace \cup \left\lbrace vw \middle| vw \in E(t_2) \right\rbrace, r_C(e) < \frac{\omega 2^{\mathcal{T}_v}}{\sqrt k}$. As this set has $5 - l$ elements, this happens with probability $\leq \left(\frac{\omega 2^{\mathcal{T}_v}}{\sqrt k}\right)^{5 - l}$.

          So as these three conditions are independent, and multiplying the probability that $Y_{t_1^u}, Y_{t_2^v} > 0$ by the values they take when they are, $\mathbb{E}\left\lbrack Y_{t_1^u}Y_{t_2^v}\right\rbrack \leq \left(\omega 2^{\mathcal{T}_v}\right)^{2-l}k^{l-1}$. So as $\omega 2^{\mathcal{T}_v} \leq \sqrt{k}$, this is $\leq k^\frac{l}{2}$ for $l < 3$, and $\leq k^{l - 1}$ otherwise. So the contribution to $\mathbb{E}\left\lbrack Y^2\right\rbrack$ from such pairs is $T_E k$ from $l = 2$ and $T k^\frac{3}{2}$ from $l = 3$.

        \item[$l = 1, u \not = v$:] As in the $l \not = 1$ case, $\mathbb{E}\left\lbrack Y_{t_1^u}Y_{t_2^v}\right\rbrack \leq k^\frac{l}{2} = \sqrt{k}$. So we seek to bound the number of such pairs.

          For any $w \in V(G)$, the number of such pairs intersecting at $w$ is $\leq \left(\sum_{v \in V(G), \mathcal{T}_v \not = \lo} T_{wv}\right)^2$. Now let $L = |\lbrace v \in V(G) | \mathcal{T}_v \not= \lo\rbrace |$. Suppose $L = r$. By Cauchy-Schwartz, \[\left(\sum_{v \in V(G), \mathcal{T}_v \not = \lo} T_{wv}\right)^2 \leq r \sum_{v \in V(G)} T_{wv}^2\]. So by summing across all $w$, the total number of such pairs is $\leq rT_E$. So the contribution to the expectation conditioned on $L = r$ is $\leq \sqrt{k} rT_E$. As our bound on $\mathbb{E}\left\lbrack Y_{t_1^u}Y_{t_2^v}\right\rbrack$ holds for any values of $\mathcal{T}_u, \mathcal{T}_v$, we can then bound the contribution to $\mathbb{E}\left\lbrack Y^2\right\rbrack$ from such pairs by:

          \begin{align*}
            \sum_r \sqrt{k}rT_E \mathbb{P}\left\lbrack L = r \right\rbrack & = \mathbb{E}\left\lbrack L \right\rbrack \sqrt{k}T_E\\
            & = \sqrt{k}T_E\sum_{v \in V(G)} \mathbb{P}\left\lbrack \mathcal{T}_v \not = \lo\right\rbrack\\
            & \lesssim \sqrt{k}T_E\sum_{v \in V(G)}\frac{\sqrt{k}T_v}{\wt{T}} & \mbox{By Lemma \ref{hierror}.}\\
            & = T_Ek\frac{T}{\wt{T}}
          \end{align*}

        \item[$l = 1, u = v$:]  As in the $l \not = 1$ case, $\mathbb{E}\left\lbrack Y_{t_1^u}Y_{t_2^v}\right\rbrack \leq \frac{k^{l -1}}{\omega^{l-2}2^{\mathcal{T}_v(l -2 )}} = \omega 2^{\mathcal{T}_v}$. So at any vertex $v$, the contribution to the expectation, conditioned on $\mathcal{T}_v$, is $\leq T_v^2 \omega 2^{\mathcal{T}_v}$. 

          We will consider two cases: $\mathcal{T}_v \leq \left\lceil\log\frac{\wt{T}}{\omega T_v}\right\rceil$ and $\mathcal{T}_v = \left\lceil\log \frac{\wt{T}}{\omega T_v}\right\rceil + i$ for some $i \geq 1$.

          In the first case, $\mathbb{E}\left\lbrack Y_{t_1^u}Y_{t_2^v}\right\rbrack \leq  \frac{\wt{T}}{T_v}$, and there are no more than $T_v^2$ such pairs for each  vertex $v$, so the total contribution to the expectation from such vertices is $\leq \sum_{v \in V(G)} T_v\wt{T} \lesssim T\wt{T}$.

          In the second case, $\mathbb{E}\left\lbrack Y_{t_1^u}Y_{t_2^v}\right\rbrack \leq  T_v\wt{T}2^i$, and the probability of $\mathcal{T}_v$ being at least this high is:

          \begin{align*}
            & \lesssim \min\left\lbrace 1, \left( \frac{k^2T_v}{2^{2i} \left(\wt{T}\right)^2} + \sum_w T_{vw}^2 \frac{k}{2^{i}\wt{T}T_v} \right)^2 \right\rbrace & \mbox{By Lemma \ref{derror}}
          \end{align*}

          Now let $x \in \mathbb{R}$ be the unique solution to $\frac{k^2T_v}{2^{2x} \left(\wt{T}\right)^2} + \sum_w T_{vw}^2 \frac{k}{2^{x}\wt{T}T_v} = 1$. For $i \leq x$, \[1 \leq 2^{-|i - x|}\left( \frac{k^2T_v}{2^{2i} \left(\wt{T}\right)^2} + \sum_w T_{vw}^2 \frac{k}{2^{i}\wt{T}T_v} \right)
          \]. 
          For $i \geq x$, \[\left( \frac{k^2T_v}{2^{2i} \left(\wt{T}\right)^2} + \sum_w T_{vw}^2 \frac{k}{2^{i}\wt{T}T_v} \right)^2 \leq 2^{-|i - x|}\left( \frac{k^2T_v}{2^{2i} \left(\wt{T}\right)^2} + \sum_w T_{vw}^2 \frac{k}{2^{i}\wt{T}T_v} \right)\]. 

          So in either case, the probability of $\mathcal{T}_v$ being at least this high is \[\lesssim 2^{-|i - x|}\left( \frac{k^2T_v}{2^{2i} \left(\wt{T}\right)^2} + \sum_w T_{vw}^2 \frac{k}{2^{i}\wt{T}T_v} \right)
          \]. By summing across all possible values of $i$, it follows that the contribution to $\mathbb{E}\left\lbrack Y^2 \right\rbrack$ is:

          \begin{align*}
            &\lesssim T_v\wt{T}\sum_{i = \max\left\lbrace 0, -\left\lceil\log \frac{\wt{T}}{\omega T_v}\right\rceil \right\rbrace}^{\left\lceil \log \frac{\sqrt k}{\omega}\right\rceil} 2^{i-|i - x|}\left( \frac{k^2T_v}{2^{2i} \left(\wt{T}\right)^2} + \sum_w T_{vw}^2 \frac{k}{2^{i}\wt{T}T_v} \right)\\
            & \leq \sum_{i = \max\left\lbrace 0, -\left\lceil\log \frac{\wt{T}}{\omega T_v}\right\rceil \right\rbrace}^{\left\lceil \log \frac{\sqrt k}{\omega}\right\rceil} 2^{-|i - x|}\left( \frac{k^2T_v^2}{\wt{T}} + \sum_w T_{vw}^2 k \right)\\
            & \leq \sum_{i = 0}^\infty 2^{-i}\left( \frac{k^2T_v^2}{\wt{T}} + \sum_w T_{vw}^2 k \right)\\
            & \lesssim \frac{k^2T_v^2}{\wt{T}} + \sum_w T_{vw}^2 k
          \end{align*}

          And so, by summing across all $v$, the contribution is:

          \begin{align*}
            \lesssim  \frac{k^2T_V}{\wt{T}} + T_E k
          \end{align*}

      \end{description}

      Now, by summing the bounds from all of the above cases together, we can bound $\mathbb{E}\left\lbrack Y^2 \right\rbrack$ and therefore $\var(\overline{T_H})$.

      \begin{align*}
        \var(\overline{T_H}) & \leq  \mathbb{E}\left\lbrack \overline{T_H}^2 \right\rbrack\\
        & \leq \mathbb{E}\left\lbrack Y^2 \right\rbrack\\
        &\lesssim T^2 + T\wt{T} + \frac{k^2T_V}{\wt{T}} + T_E k + Tk^\frac{3}{2}
      \end{align*}

    \end{proof}
    \subsubsection{\texorpdfstring{$\overline{T}$}{T}}
    \begin{lemma}
      \label{texpec}
      $\mathbb{E}\left\lbrack\overline{T}\right\rbrack = T$.
    \end{lemma}
    \begin{proof}
      For any $T_L, T_H$, $T_L + T_H = T$. So as $\overline{T} =  \overline{T_L} + \overline{T_H}$, by Lemmas \ref{tlexpec}, \ref{thexpec}, $\mathbb{E}\left\lbrack\overline{T}| T_L,T_H\right\rbrack = T_L + T_H = T$.
    \end{proof}

    \begin{lemma}
      \label{tvar}
      $\var\left(\overline{T}\right) \lesssim Tk^\frac{3}{2} + \frac{k^2 T_V}{\wt{T}} +  T_Ek  + T^2 + T\wt{T}$.
    \end{lemma}

    \begin{proof}
      $\overline{T} =  \overline{T_L} + \overline{T_H}$, so $\var(\overline{T}) \lesssim \var(\overline{T_L}) + \var(\overline{T_H})$. Our bound then follows by using the bounds in Lemmas \ref{tlvar} and \ref{thvar}.
    \end{proof}

    \section{Single-Pass Algorithm} \label{spalgo}
    \subsection{Outline}
    The algorithm presented earlier calculates an estimate of $T$ in two conceptual ``passes''. We show how, with only one pass, we can calculate the output of this algorithm while storing only $O\left(\frac{m\log \frac{T_V\sqrt{k}}{T^2}}{k}\right)$ edges. Our main theorem then follows as a corollary.
    
    We can do this because, in order to calculate the output of the second pass, we only need to know $\mathcal{T}_v$ when it is equal to $h(v)$. In order to calculate $\mathcal{T}_v$, we need all the edges $uv$ such that $r_{D,i}(uv) < \frac{2^{h(u)}\omega}{\sqrt{k}}$ (for $i = 1,2$) and $d(v) = 1$. So as there are $\sim \frac{\log \frac{T_V\sqrt{k}}{T^2}}{k}$ possible values $h$ of $h(u)$, each of which is taken with probability $\omega2^{-h}$, this means we need to store $O\left(\frac{m\log \frac{T_V\sqrt{k}}{T^2}}{k}\right)$ edges in expectation.

    We also show that the post-processing time required is $O\left(\overline{m}^\frac{3}{2}\right)$, where $\overline{m}$ is the number of edges we sample. 

    \subsection{Space Complexity}
    \begin{lemma}
      \label{spass}
      The first and second pass calculations can be performed while storing no more than $O\left(\frac{m\log \frac{T_V\sqrt{k}}{T^2}}{k}\right)$ edges in expectation.
    \end{lemma}
    \begin{proof}
      We define 5 sets of edges $\left({E_i}\right)_{i = 0}^4$ that our algorithm will store. These are expressed as ordered pairs $uv$, but as our input is undirected edges, an edge $e = uv$ will be kept if either $uv$ or $vu$ is in one of the sets $E_i$.

      \begin{align*}
        E_0 = \lbrace uv \in E & | \exists i, d_i(u) = 1, d_i(v) = 1\rbrace\\
        E_1 = \lbrace uv \in E & | \exists i, c(u) = 1, d_i(v) = 1\rbrace\\
        E_2 = \lbrace uv \in E & | c(u) = 1, c(v) = 1\rbrace\\
        E_3 = \lbrace uv \in E & | \exists i, r_{D,i}(uv) < \frac{2^{h(u)}\omega}{\sqrt{k}}, d_i(v) = 1\rbrace\\
        E_4 = \lbrace uv \in E & | r_C(uv) < \frac{2^{h(u)}\omega}{\sqrt{k}}, c(v) = 1\rbrace
      \end{align*}

      Then, $\forall uv \in E$:

      \begin{align*}
        \mathbb{P}\left\lbrack uv \in E_0 \right\rbrack & \lesssim \frac{1}{k}\\
        \mathbb{P}\left\lbrack uv \in E_1 \right\rbrack & \lesssim \frac{1}{k}\\
        \mathbb{P}\left\lbrack uv \in E_2 \right\rbrack & = \frac{1}{k}\\
        \mathbb{P}\left\lbrack uv \in E_3 \right\rbrack & = \frac{1}{\sqrt{k}} \sum_{i = 1}^2\sum_{h = 0}^{\left\lceil \log \frac{\sqrt{k}}{\omega} \right\rceil} \mathbb{P}\left\lbrack r_{D,i}(uv) < \frac{2^{h(u)}\omega}{\sqrt{k}} \middle| h(u) = h \right\rbrack \mathbb{P}\left\lbrack h(u) = h \right\rbrack\\
        & \lesssim \frac{1}{\sqrt{k}} \sum_{h = 0}^{\left\lceil \log \frac{\sqrt{k}}{\omega} \right\rceil} \frac{2^{h}\omega}{\sqrt{k}} \frac{1}{\omega 2^h}\\
        & \lesssim \frac{\log \frac{\sqrt{k}}{\omega}}{k}\\
        \mathbb{P}\left\lbrack uv \in E_4 \right\rbrack & \lesssim\frac{\log \frac{\sqrt{k}}{\omega}}{k}\\ 
      \end{align*}

      So these sets will contain $O\left(\frac{m \log  \frac{\sqrt{k}}{\omega}}{k}\right)$ edges in expectation. Then, as $\omega = \min \left\lbrace \frac{(\wt{T})^2}{T_V^+},\sqrt{k}\right\rbrace$, and $T_V^+$ can be any upper bound on $T_V$, maintaining these sets requires storing $O\left(\frac{m \log  \frac{T_V\sqrt{k}}{T^2}}{k}\right)$ edges.

      We will then use these to calculate $\overline{T_L}$ and $\overline{T_H}$ as follows:

      \begin{description}
        \item[$\overline{T_L}$:] Recall that $\overline{T_L}$ is the number of triangles in $G_L'$, the subgraph of $G$ induced by $V_L' = \lbrace u \in V(G) | c(u) = 1, \mathcal{T}_u = \lo \rbrace$. $E_2$ will contain every edge in $E(G_L')$, so we can calculate $\overline{T_L}$ provided we can determine which of the edges in $E_2$ are in $G_L'$. It is therefore sufficient to calculate $\mathcal{T}_v$ for all $v$ s.t. $c(v) = 1$.

          To do this, we will need to calculate $\sum_{t \in \tau(v)} X_{v, i}^{(h,t)}$ for $h = \left\lceil \log \frac{\sqrt k}{\omega} \right\rceil$, and each $i \in \lbrack 2\rbrack$. For each triangle $(u,v,w) \in \tau(v)$, $X_{v, i}^{(h,t)} = 1$ iff 

          \begin{align*}
            d_i(u) & = 1\\
            d_i(w) & = 1\\
            r_{D,i}(vu) & < \frac{\omega 2^h}{\sqrt k}\\
            r_{D,i}(vw) & < \frac{\omega 2^h}{\sqrt k}
          \end{align*}

          And 0 otherwise. The third and fourth conditions always hold, as we are dealing with the case when $h = \left\lceil \log \frac{\sqrt k}{\omega} \right\rceil$, so we need to know how many such triangles exist with $d_i(u) = d_i(w) = 1$. When $d_i(u) = d_i(w) = 1$, $uw \in E_0$ and $vu, vw \in E_1$ (as $c(v) = 1$), so $\sum_{t \in \tau(v)} X_{v, i}^{(h,t)}$ will be equal to the number of triangles $(u,v,w)$ with $d_i(u) = d_i(w) = 1$ in the edges we have sampled. So then $\mathcal{T}_{v,i} = \lo$ iff this number is $< \frac{\wt{T}}{k^\frac{3}{2}}$. 

          So we can compute $\mathcal{T}_v$ for each $v$ s.t. $c(v) = 1$ using our sampled edges, and therefore we can compute $\overline{T_L}$.

        \item[$\overline{T_H}$:] For any $v \in V(G)$, $\overline{T_v} = 0$ if $h(v) \not =  \mathcal{T}_v$. So it is sufficient to calculate $\mathcal{T}_v$ when $h(v) = \mathcal{T}_v$, and to know that $h(v) \not = \mathcal{T}_v$ otherwise. We can calculate  $\sum_{t \in \tau(v)} X_{v, i}^{(h,t)}$ for each $h \leq h(v)$, $i \in \lbrack 2 \rbrack$, as for any any triangle $(u,v,w) \in G$ s.t. $ d_i(u) = d_i(w) = 1$, $uw$ will be in $E_0$, and if $r_{D,i}(vu), r_{D,i}(vw)  < \frac{\omega 2^h}{\sqrt k}$ for some $h \leq h(v)$, $vu$ and $vw$ will be in $E_3$. 

          So then, as $H_{v,i} = \left\lbrace h \in \left\lbrace 1, \dots, \left\lceil \log \frac{\sqrt k}{\omega} \right\rceil\right\rbrace \middle| X_{v,i}^{(h)} \geq \frac{\omega \wt{T}2^h}{k^2}\right\rbrace$ we can compute $H_{v,i} \cap \lbrace 0, \dots, h(v) \rbrace$. As $\mathcal{T}_{v,i} = \min H_{v,i}$, we can compute each $\mathcal{T}_{v,i}$ if it is $\leq h$, and determine that it is $> h(v)$ or $\lo$ otherwise. Then, if $\mathcal{T}_v \leq h(v)$, at least one of $\mathcal{T}_{v,1},\mathcal{T}_{v,2}$ is $\leq h(v)$, and so we can calculate it and therefore calculate $\mathcal{T}_v$, and if not we can determine that $\mathcal{T}_v$ is either $> h(v)$ or $\lo$. (although not necessarily which one)

          So if $\mathcal{T}_v \not = h(v)$, we know $\overline{T_v} = 0$, and then if $\mathcal{T}_v = h(v)$, we need to know how many triangles $(u,v,w)$ in $G$ there are such that:

          \begin{align*}
            c(u) & = 1\\
            c(w) & = 1\\
            r_C(vu) & < \frac{\omega 2^{\mathcal{T}_v}}{\sqrt k}\\
            r_C(vw) & < \frac{\omega 2^{\mathcal{T}_v}}{\sqrt k}
          \end{align*}

          If these criteria hold, then $uw \in E_2$, and $vu,vw \in E_4$. So we can calculate $\overline{T_v}$ by counting the number of triangles amongst our sampled edges such that these criteria hold.

          We then calculate $\overline{T_H}$ by summing $\overline{T_v}$ for each $v \in V(G)$.

      \end{description}

    \end{proof}

    \restate{thm:main}
    \begin{proof}
    We will assume that $\epsilon \leq 1/2$, as if $\epsilon > 1/2$, we may follow the procedure for an $(1/2,\delta)$ approximation while losing at most a constant factor in the number of edges stored.

      By Lemma \ref{texpec}, the expectation of our estimate $\overline{T}$ of $T$ will be $T$. We now choose $k$ such that the variance of $\overline{T}$ is $O(T\wt{T} + T^2)$. By Lemma \ref{tvar}, $\var(\overline{T}) \lesssim Tk^\frac{3}{2} + \frac{k^2 T_V}{\wt{T}} +  T_Ek  + T\wt{T} + T^2$. Then, as we may choose $T_V^+$ to be any upper bound on $T_V$, we can choose it to be within a constant factor of the true value, and so the variance will be $O\left(T\wt{T} + T^2\right)$ if $k = \min\left\lbrace \wt{T}^{2/3}, \frac{\wt{T}^\frac{3}{2}}{\sqrt{T_V}}, \frac{\wt{T}^2}{T_E}\right\rbrace$.

      We can repeat the algorithm $O\left(\frac{1}{\epsilon^2}\log\frac{1}{\delta}\right)$ times. We split these results into blocks of size $O\left(\frac{1}{\epsilon^2}\right)$, take the mean of each block, and then take the median of all these means, which gives us an estimate within $\frac{1}{4}\epsilon\left(\sqrt{T\wt{T}} + T\right)$ of $T$ with probability $1 - \delta$, provided we choose the constant factors in our block sizes appropriately. We will then return this result if it is at least $\frac{3}{4}\wt{T}$, and report that $\wt{T} > T$ otherwise.
      
      To show that this gives a correct result (that is, giving an $1 \pm \epsilon$ approximation if $\wt{T} < T$, and either giving a $1 \pm \epsilon$ approximation or reporting that $\wt{T} > T$ otherwise) with probability $1 - \delta$, we consider three cases.
      
      \begin{description}
        \item[$\wt{T} \leq T$:] With probability $1 - \delta$, our result will be within $\frac{\epsilon}{2}T$ of $T$. As $\epsilon < 1/2$, it will therefore be greater than $\frac{3}{4}\wt{T}$, so we will return a $1 \pm \epsilon/2$, and therefore a $1 \pm \epsilon$, approximation of $T$.

        \item[$\wt{T} \in ( T, 2T )$:] With probability $1 - \delta$, our result will be within $\epsilon T$ of $T$. So we will return a $1 \pm \epsilon$ approximation of $T$, or report that $\wt{T} > T$, both of which are valid results for $\wt{T}$ in this range.

        \item[$\wt{T} > 2T$] With probability $1 - \delta$, our result will be within $\frac{\epsilon}{2} \wt{T}$ of $T$. So as $\epsilon < 1/2$, it will be $< \frac{3}{4} \wt{T}$, and so we will report that $\wt{T} > T$.
      \end{description}

      By Lemma \ref{spass}, we can obtain this approximation while keeping at most
      $
      \lesssim \frac{m\log \frac{T_V\sqrt{k}}{\wt{T}^2}}{k}
      $
      edges. Then, if the dominating term in our expression for $k$, $\min\left\lbrace \wt{T}^{2/3}, \frac{\wt{T}^\frac{3}{2}}{\sqrt{T_V}}, \frac{\wt{T}^2}{T_E}\right\rbrace$, is $\wt{T}^{2/3}$, $\wt{T}^{2/3} \leq \frac{\wt{T}^\frac{3}{2}}{\sqrt{T_V}}$, so $\frac{T_V}{\wt{T}^2} \leq \wt{T}^{-\frac{1}{3}}$ and $\sqrt{k} \leq \wt{T}^\frac{1}{3}$, so $\frac{T_V\sqrt{k}}{\wt{T}^2} \leq 1$, and so we need $\lesssim \frac{m}{k}$ edges. Otherwise, as $\log \frac{T_V\sqrt{k}}{\wt{T}^2} \leq \log k$, we need $\lesssim \frac{m\log k}{k}$ edges.

      The result then follows from the fact that $T_V \leq \Delta_VT$ and $T_E \leq \Delta_ET$.
    \end{proof}
    \begin{corollary}
      We can obtain a constant-error approximation to $T$ while sampling $\lesssim 
      \frac{1}{\wt{T}^{2/3}}
      + \frac{\sqrt{\Delta_V}\log \wt{T}}{\wt{T}}
      +\frac{\Delta_E\log \wt{T}}{\wt{T}} $ edges.
    \end{corollary}

    \subsection{Time Complexity}
    \begin{lemma}
      If our algorithm samples $\overline{m}$ edges, we can compute $\overline{T}$ in $O(\overline{m}^\frac{3}{2})$ time.
    \end{lemma}
    \begin{proof}
      Let $\left(E_i\right)_{i = 1}^5$ be as described in the previous section. After executing the sampling phase of our single-pass algorithm, by Lemma \ref{spass}, $E' = \bigcup_{i=1}^5 E_i$ contains every triangle we need to calculate $\overline{T}$. We can list the triangles in $E'$ in $O\left(\overline{m}^\frac{3}{2}\right)$ time by using the algorithm in \cite{IR78}. We can then calculate $\overline{T}$ by passing over this list (of length no more than $\overline{m}^\frac{3}{2}$) and, for each triangle in the list, checking whether it contributes to $\overline{T_L}$ or $\overline{T_H}$ and if so, what that contribution is.

      Once the $\mathcal{T}_v$ are known, this only requires checking a constant number of criteria per triangle. However, we must first calculate the $\mathcal{T}_v$. Doing this na\"ively, by iterating over every triangle and checking all the $X_{v,i}^h$ it could contribute to, could take as much as $O\left(\overline{m}^\frac{3}{2} \log k\right)$ time. 

      Recall that a triangle $(u,v,w)$ contributes to $X_{v,i}^{(h)}$ iff $d(u) = d(w) = 1$ and $r_{D,i}(vu), r_{D,i}(vw) < \frac{2^h\omega}{\sqrt{k}}$. $X_{v,i}^{(h)}$ is therefore monotone increasing in $h$, so we can calculate all the $X_{v,i}^{(h)}$ in $O(\overline{m}^\frac{3}{2})$ time as follows:

      \begin{enumerate}
        \item Sort our edges $e$ by the value of $r_{D,i}(e)$.
        \item Copy each triangle $(u,v,w)$ three times, choosing a different vertex as the ``main vertex''$v$ each time.
        \item Sort these triangles in ascending order of $\max\lbrace r_{D,i}(vu), r_{D,i}(vw)\rbrace$. (and therefore by the least value of $h$ needed for the triangle to contribute to $X_{v,i}^{(h)}$)
        \item Run through the list in order, for each triangle adding 1 to $X_{v,i}^{(h)}$ for each $h$ such that \[
            h \geq \log\left(\frac{\sqrt k}{\omega}\max\lbrace r_{D,i}(vu), r_{D,i}(vw)\rbrace\right)
          \] By storing the running total as we run through the list, we can do this in a constant number of updates. For any $v,i$, the first time we set $X_{v,i}^{(h)}$ to be greater than $\frac{\omega \wt{T} 2^h}{k^2}$ for some $h$, set $\mathcal{T}_{v,i}$ to $h$.
      \end{enumerate}

      For any $v,i$ where $\mathcal{T}_{v,i}$ has not been set at the end of this process, it is then $\lo$. We now have $\mathcal{T}_{v,i}$ for $i \in \lbrack 2 \rbrack$ and for every $v$ involved in one of our triangles, so we may calculate $\overline{T}$ in time $O\left(\overline{m}i^\frac{3}{2}\right)$ by iterating over our list of triangles and checking the contribution of each one to $\overline{T_L}$ and $\overline{T_H}$.

    \end{proof}

    \section{Lower Bounds}
    We recall our definition of an \emph{instance lower bound}.
    \restate{dfn:solving}
    \restate{dfn:ilb}

    For each bound, we will show two distributions on subgraphs that cannot be distinguished with less space/samples than our lower bound.

    \restate{dfn:distinguish}
    \restate{lem:ibounddist}
    \begin{proof}
    By our definition of ``solving'' a graph, any algorithm which $(\epsilon,1/10)$ solves $G$ can also distinguish between a draw from $\mathcal{G}_1$ and $\mathcal{G}_2$ by returning 1 for any graph with at least $C$ triangles and 0 for any graph with fewer than $C$ triangles, so that $\Pr[f(\mathcal{A}(G_1)) = 1] \geq 4/5$
        and $\Pr[f(\mathcal{A}(G_2)) \not= 1] \geq 4/5$.

    \end{proof}

    \subsection{Multiple Heavy Edges}
    We will express the bound shown in \cite{BOV13} as an instance bound.

    \restate{dfn:hegraph}
    \restate{thm:edgelb}
    \begin{proof}
      As in, \cite{BOV13} we reduce to the indexing problem $\text{Index}_n$: Alice has a binary vector $w$ of length $n$, and Bob has an index $x \in \lbrack n\rbrack$. Alice must send a message to Bob such that Bob can communicate $w\lbrack x\rbrack$. By \cite{CCK10}, the randomized communication complexity of this problem is $\Omega(n)$. So for any instance of the problem:

      \begin{description}
        \item[If $r = 1$:] Alice can encode $w$ in a subgraph of $D_{r,d}$ by, for each $i$ in $\lbrace0, \dots, n -1\rbrace$, including $u_{2i}u_{2i+1}$ iff $w_i = 1$. She can then run a distinguishing algorithm on these edges and send them to Bob, who proceeds to add the $d$ vertices $\lbrace v_i\rbrace_{i=0}^{d-1}$ to the graph, connecting each of them to $u_{2x}$ and $u_{2x+1}$. If the algorithm reports the graph has 0 triangles, Bob reports that $w_x = 0$, and if it has $d$ triangles, he reports that $w_x = 1$.
        \item[If $r > 1$:] Alice and Bob can perform the same procedure but each copying their graph $r$ times, with Bob now checking if he has $rd$ triangles.
      \end{description}

      By \cite{BOV13}, there is a distribution on inputs such that this algorithm requires at least $\Omega\left(d\right)$ bits of storage. Let $G'_1$, $G'_2$ correspond to the distributions on the graph Alice and Bob create conditioned on $w_x = 1$ and $w_x = 0$ respectively. Note that $T(G_1) = T(D_{r,d})$ and $T(G_2) = 0$ with probability 1.

      Then, letting $H$ be an instance of $D_{r,d}$, each draw from a $G'_i$ is isomorphic to a subgraph of $H$. We can therefore define subgraph distributions $G_1, G_2$ on $H$ and a permutation distribution $\sigma$ such that $\sigma^{-1}(G'_i) = G_i$.

      So we now have distributions on subgraphs $G_1, G_2$ of $H$ such that it is hard to distinguish between $\sigma(G_1),\sigma(G_2)$ for some non-uniform permutation distribution $\sigma$. However, if it is hard to distinguish these, they must also be hard to distinguish for a uniformly random $\sigma$, as otherwise any distinguishing algorithm could distinguish them by randomly permuting its input.

      So counting triangles for $D_{r,d}$ requires $\Omega\left(d\right)$ bits of storage.

    \end{proof}

    \subsection{Hubs Graph}
    \restate{dfn:hubgraph}
    \restate{thm:hublb}
    \begin{proof}
      We use a reduction from the Boolean Hidden Matching problem~\cite{BJK04,Kerenidis:2006,GKKRd07}, in particular the variant set out in \cite{GKKRd07}, which we will refer to as $\text{BHM}_{n,\alpha}$. In this problem, Alice gets a string $x \in \lbrace 0, 1\rbrace^{2n}$, Bob an $\alpha$-partial matching $M$ on $\lbrack 2n \rbrack$ (interpreted as an $\alpha n \times 2n$ matrix, where the $k^\text{th}$ row corresponds to the $k^\text{th}$ edge of the matching, so that if the edge is $ij$, the $i^\text{th}$ and $j^\text{th}$ position of the row are 1, and all others are 0) and a string $w \in \lbrace 0, 1\rbrace^{\alpha n}$. 

      Then, with matrix multiplication interpreted mod 2, Bob must determine whether $Mx = w$ or $Mx = \overline{w}$.

      In \cite{GKKRd07}, it was shown that this requires $\Omega(\sqrt{n/\alpha})$ bits of communication in the randomized one-way communication model.

      We provide a protocol for $\text{BHM}_{n,\alpha}$, given an algorithm capable of distinguishing between graphs with $\alpha n$ and $0$ triangles, as follows:

      \begin{enumerate}
        \item Alice encodes her string $x \in \lbrace 0, 1\rbrace^{2n}$ as a graph as follows: a single vertex $a$, $2n$ vertices $\lbrace b_i \rbrace_{i \in \lbrack 2n\rbrack}$, and $2n$ vertices $\lbrace c_i \rbrace_{i \in \lbrack 2n\rbrack}$. Then, for each $i \in \lbrack 2n\rbrack$, she adds the edge $ab_i$ if $x_i = 0$, and $ac_i$ if $x_i = 1$. She then runs the distinguishing algorithm on the resultant graph and sends the internal state of the algorithm to Bob.

        \item Bob then, using the same vertex IDs, encodes $M$ and $w$ as follows: For the $k^\text{th}$ edge $ij$ in the matching, he adds $b_ib_j,c_ic_j$ to his graph if $w_k = 0$, and $b_ic_j,b_jc_i$ if $w_k = 1$. He then starts the triangle counting algorithm with the internal state sent to him by Alice, inserts the edges he has created, and reads the output of the algorithm. If the algorithm reports that the graph has $\alpha n$ triangles, he reports that $Mx = w$, and if it reports $0$ triangles, he reports that $Mx = \overline{w}$.
      \end{enumerate}

For each $l \in \lbrack \alpha n\rbrack$ (with the $l^\text{th}$ edge of the matching being $ij$), the graph will contain the triangle $(a,b_i,b_j)$ if $w_l = 0$ and $x_i = x_j = 0$,  $(a,c_i,c_j)$ if $w_l = 0$ and $x_i = x_j = 1$, $(a,b_i,c_j)$ if $w_l = 1$ and $x_i = 0, x_j = 1$, and $(a,b_j,c_i)$ if $w_l = 1$ and $x_i = 1, x_j = 0$. So in each copy, there will be one triangle for the $l^\text{th}$ edge $ij$ iff $w_l = x_i + x_j \mod 2$. So the number of triangles $T$ in the graph is $\alpha n$ if $Mx = w$ and $0$ if $Mx = \overline{w}$, and so this protocol correctly solves $\text{BHM}_{n, \alpha}$ provided the triangle counting algorithm approximates the number of triangles up to an $\alpha n/2$ additive error.

      By \cite{GKKRd07}, there is a distribution on inputs such that this algorithm requires at least $\Omega\left(\sqrt{n/\alpha}\right)$ bits of storage. Let $G'_1$, $G'_2$ correspond to the distributions on the graph Alice and Bob create conditioned on $Mx = w$ and $Mx = \overline{w}$ respectively. Note that $T(G_1) = \alpha n$ and $T(G_2) = 0$ with probability 1.

      Then, each of the $G'_i$ will always consist of a single hub vertex, $2n$ other vertices, $n$ edges from the hub to the other vertices, and a partial matching on the other vertices. Therefore, letting $H$ be an instance of $H_{1/\alpha,\alpha n}$, each draw from a $G'_i$ is isomorphic to a subgraph of $H$. We can therefore define subgraph distributions $G_1, G_2$ on $H$ and a permutation distribution $\sigma$ such that $\sigma^{-1}(G'_i) = G_i$.

      So we now have distributions on subgraphs $G_1, G_2$ of $H$ such that it is hard to distinguish between $\sigma(G_1),\sigma(G_2)$ for some non-uniform permutation distribution $\sigma$. However, if it is hard to distinguish these, they must also be hard to distinguish for a uniformly random $\sigma$, as otherwise any distinguishing algorithm could distinguish them by randomly permuting its input.

      So counting triangles for $H_{1/\alpha, \alpha n}$ requires $\Omega\left(\sqrt{n/\alpha}\right)$ bits of storage, and the result follows from setting $r = 1/\alpha$, $d = \alpha n$.
    \end{proof}

    \subsection{Independent Triangles}
    \restate{dfn:itgraph}
    \restate{thm:itlb}
    If $n$ is constantly bounded, the result holds automatically, so we will assume $n \geq 80$.

    We start by proving a lemma that applies to all graphs.

    \define{lem:ttwothirds}{Lemma} {
      Let $G$ be a graph with $m$ edges. Let $\mathcal{A}$ be an algorithm which, when given any graph isomorphic to $G$ as input, samples at least one triangle with constant probability. $\mathcal{A}$ must sample $\Omega\left(\frac{m}{T^{2/3}}\right)$ edges in expectation.
    }
    \state{lem:ttwothirds}
    \begin{proof}
      As such algorithm must work for an arbitrary permutation of the vertices of $G$, we may assume without loss of generality that for any sampling algorithm we consider, for all potential triangles $t$ in $G$, the probability $p_t$ that $t$ is sampled by the algorithm if it exists is the same, and likewise for all potential edges $e$ in $G$. Then, we can bound $p_t$ as follows: For any set of $l$ vertices, let $R$ be the number of edges amongst these $l$ vertices that are sampled. Clearly $\mathbb{E}\left\lbrack R\right\rbrack \leq l^2p_e$. Conditioned on $R = r$, it can be shown (\cite{stackex:823650}) that there are $\lesssim r^\frac{3}{2}$ triangles amongst the edges sampled. So as there are $\Omega(l^3)$ triangles amongst the $l$ vertices, $p_t$ conditioned on $R = r$ is $\lesssim \left(\frac{\sqrt{r}}{l}\right)^3$. By Markov's inequality, $R \lesssim l^2p_e$ with constant probability (and we may choose any constant). Letting this event be $A$, we therefore have $p_t \lesssim p_e^\frac{3}{2}$ conditioned on $A$. So $p_e$ must be $\Omega\left(\frac{1}{T^{2/3}}\right)$ if we want to sample at least one triangle with constant probability.
    \end{proof}

    Let $I$ be an instance of $I_n$. As $T(I) = n$, and $|E(I)| = 3n$, we therefore need to sample $\Omega(n^\frac{1}{3})$ edges if we are to sample at least one triangle with constant probability. We now need to demonstrate that a sampling algorithm that does not sample any triangles from $I$ will also fail to distinguish between certain graphs with very different numbers of triangles. We will need the following lemma to proceed:

    \begin{lemma}
      \label{cyclecoloring}
      Let $H$ be an acyclic graph, and let $\chi$ be a random 2-coloring of $V(H)$. The following graphs are identically distributed:

      \begin{align*}
        H_1 & = (V(H), \lbrace uv \in E(H) | \chi(u) = \chi(v) \rbrace)\\
        H_2 & = (V(H), \lbrace uv \in E(H) | \chi(u) \not= \chi(v) \rbrace)
      \end{align*}
    \end{lemma}
    \begin{proof}
      As $H$ is acyclic, it is a forest. Clearly the behavior of two distinct trees in the forest under any given coloring are independent of one another, so it suffices to prove that the result holds for $H$ a tree. Choose a vertex $v$ to act as the root of $H$. For all subtrees $H'$ of $H$, with roots $r'$, let $H'_1 = H' \cap H_1, H'_2 = H' \cap H_2$. We will show that for $i = 1,2$, $H'_i$ includes each edge in $H'$ independently with probability $\frac{1}{2}$, and independently of value of $\chi(r')$.

      We proceed by structural induction. Suppose the result holds for all proper subtrees of $H'$. Then let $C$ be the set of children of $r'$. As the proper subtrees of $H'$ do not share any vertices, the inclusion of any edge in $H'_i$ is independent between any two distinct proper subtrees, and independent of the values $\chi$ takes on $C$.  The only remaining edges are the edges $\lbrace r'c | c \in C\rbrace$. Then, for any value of $\chi(r')$, the inclusion of each $r'c$ happens with probability $\frac{1}{2}$, independently of each other and of the value of $\chi(r')$. So as the proper subtree edges are included independently of the value of $\chi$ on $C$, they are included independently of the $r'c$. So every edge in $E(H')$ is included with probability $\frac{1}{2}$, independently of one another and the value of $\chi(r')$.

      So $H_1$, $H_2$ are identically distributed.
    \end{proof}

    Now, with $\chi$ a random 2-coloring of the vertices of $I$, we define two subgraphs $I_1, I_2$ of $I$ as follows:

    \begin{align*}
      I_1 & = (V(I), \lbrace uv \in E(I) | \chi(u) = \chi(v) \rbrace)\\
      I_2 & = (V(I), \lbrace uv \in E(I) | \chi(u) \not= \chi(v) \rbrace)
    \end{align*}

    Note that $I_1,I_2$ both have $\frac{n}{2}$ edges in expectation. $T(I_2)$ is always 0, and $I_1$ will include every triangle independently with probability $\frac{1}{4}$. So as $n \geq 80$, $T(I_1) \geq \frac{n}{8} > 0 \geq T(I_2)$ with probability $\geq \frac{9}{10}$.

    So for any sampling algorithm to sample a triangle from either $I_1$ or $I_2$ with constant probability, it must sample $\Omega\left(m^\frac{1}{3}\right)$ edges in expectation. Suppose it does not. Then let $\overline{I}$ be the set of edges from $I$ it samples. (and so $\overline{I_i} = I_i \cap \overline{I}$ is the set of edges of edges from $I_i$ it samples) As $I$ contains no cycle longer than a triangle, $\overline{I}$ is acyclic. Therefore, by Lemma \ref{cyclecoloring},  $\overline{I_1},  \overline{I_2}$ are identically distributed. So conditioned on failing to sample a triangle, the algorithm is incapable of distinguishing between $I_1$ and $I_2$.

    So to distinguish between $I_1$ and $I_2$ a sampling algorithm must sample $\Omega\left(n^\frac{1}{3}\right)$ edges in expectation.

    \subsection{\texorpdfstring{$G_{n,p}$}{G(n, p)}}
    \restate{thm:gnplb}
    Let $G$ be a random draw from $G_{n,p}$. We will use the same pair of colorings as in our proof of the $\frac{m}{T^{2/3}}$ bound. However, as a graph drawn from $G_{n,p}$ can contain cycles longer than triangles, we need some extra work to prove the subsampled graph is acyclic. We start by extending a result from \cite{stackex:823650}.

    \begin{lemma}
      \label{cyclecount}
      Let $G$ be a graph with $m$ edges. There are at most $(2m)^\frac{l}{2}$ length-$l$ cycles in $G$.
    \end{lemma}
    \begin{proof}
      As in \cite{stackex:823650}, we define $\forall v \in V(G), A(v) = \lbrace u \in N(v), d(u) \geq d(v)\rbrace$. $\forall w \in V(G)$, $|A(w)| \leq \sqrt{m}$, if $|A(w)| > \sqrt{2m}$, then there are $>\sqrt{2m}$ vertices of degree $>\sqrt{2m}$ in the graph, and thus $> m$ edges. So $|A(w)| \leq \sqrt{2m}$.

      We order each $l$-cycle as $(v_0, \dots,v_{l-1})$ so that the $v_{l-1}$ has the highest degree in the cycle. (there are two such orderings, we pick between them arbitrarily) So each cycle now corresponds to $l$ \emph{directed} edges $(v_0v_1,v_1v_2,\dots,v_{l-2}v_{l-1})$. 

      Suppose $l$ is even. The cycle is then uniquely determined by the $\frac{l}{2}$ edges $\lbrace v_{2i}v_{2i+1} \rbrace_{i \in \left\lbrack\frac{l}{2}\right\rbrack}$ and their directions. So there are at most $(2m)^\frac{l}{2}$ $l$-cycles in $G$.

      Now suppose $l$ is odd. The cycle is then uniquely determined by the $\frac{l - 1}{2}$ edges $\lbrace v_{2i}v_{2i+1} \rbrace_{i \in \left\lbrack\frac{l - 1}{2}\right\rbrack}$ and their directions, plus the final vertex $v_{l-1}$. There are at most $(2m)^\frac{l-1}{2}$ possible choices for the edges and directions. Furthermore, as $v_{l-1}$ has the highest degree of any vertex in the cycle, $v_{l-1} \in A(v_{l-2}) \cap A(v_0)$. So as this set has size $\leq \sqrt{2m}$, there are at most $\sqrt{2m}$ possibilities for $v_{l-1}$, and so there are at most $(2m)^\frac{l}{2}$ $l$-cycles in $G$.
    \end{proof}

    We can now show that a sampling algorithm that does not use $\Omega\left(\frac{1}{p}\right)$ samples will, with probability bounded below by an arbitrary constant, fail to sample any cycle in the graph. Let $\overline{G}$ be the subgraph of $K_n$ that the algorithm would sample, and let $q$ be the probability of sampling any given edge. (as in the independent triangles case, we will choose our graph permutation uniformly at random, so without loss of generality we can assume this is the same for all edges) Then $\mathbb{E}\left\lbrack |E(\overline{G})|\right\rbrack \leq qn^2$. So, by Lemma \ref{cyclecount} and Markov's inequality, the number of $l$-cycles in $\overline{G}$ is $\leq q^\frac{l}{2}(2n)^l$ with constant probability. (for any arbitrarily small constant, so let us choose it to be $\frac{1}{2}$) So, conditioned on this constant probability event $A$, the probability of any $l$-cycle being in $\overline{G} \cap G$ is $\leq p^lq^\frac{l}{2}(2n)^l$. So, taking the union bound across all $l$, this probability is $\leq \sum_{l = 3}^\infty p^lq^\frac{l}{2}(2n)^l$. So $q$ must be $\Omega\left(\frac{1}{n^2p^2}\right)$ for at least one $l$-cycle to be found with probability $\geq \frac{1}{2}$.

    Then, for $\chi$ a random coloring of $\lbrack n\rbrack$, let $G_1$, $G_2$ be defined as follows:

    \begin{align*}
      G_1 & = (V(G), \lbrace uv \in E(G) | \chi(u) = \chi(v) \rbrace)\\
      G_2 & = (V(G), \lbrace uv \in E(G) | \chi(u) \not= \chi(v) \rbrace)
    \end{align*}

    Note that $\mathbb{E}\left\lbrack T(G) \right\rbrack = \binom{n}{3}p^3$. $T(G_2)$ will always be 0. Then, there are at least $\left\lceil \frac{n}{2} \right\rceil$ vertices with the same coloring under $\chi$. So we can take $H \subseteq G_1$ such that all vertices of $H$ have the same color and $H$ is distributed as $G_{\left\lceil\frac{n}{2}\right\rceil,p}$. Then, let $B$ be the number of triangles in $H$. $\mathbb{E}\left\lbrack B\right\rbrack = \binom{\left\lceil n/2\right\rceil}{3}p^3$, and by \cite{matvon08}, $\frac{\var(B)}{\mathbb{E}\left\lbrack B\right\rbrack^2} = O\left(\frac{1}{n^3p^3} + \frac{1}{n^2p}\right)$, so for $p = \Omega\left(\frac{1}{n}\right)$, $B$ will be a constant fraction of $\mathbb{E}\left\lbrack T(G)\right\rbrack = O\left(n^3p^3\right)$ with probability $\frac{9}{10}$.

    Now, by Lemma \ref{cyclecoloring}, conditioned on a sampling algorithm failing to sample any cycles, the samples taken from $G_1$ and $G_2$ are identically distributed, in which case the algorithm can distinguish them with probability at most $\frac{1}{2}$. So $q$ must be $\Omega\left(\frac{1}{n^2p^2}\right)$ to distinguish them with probability $\geq \frac{3}{4}$, and so as $G_1$, $G_2$ each contain $\Omega(n^2p)$ edges in expectation, the algorithm must sample $\Omega\left(\frac{1}{p}\right)$ edges in expectation.

    \section{Triangle-Dependent Sampling Lower Bound}

    \restate{dfn:triangledep}

    \restate{dfn:dgepsilon}

    \begin{lemma}
      \label{epsneeded}
      Let $f$ be a triangle-dependent sampling algorithm that solves triangle counting for a graph $G$ with $(\epsilon,\delta)$ error. For any set $A$ such that $A \subseteq V(G) \cup E(G)$, if the number of triangles in $G$ involving edges or vertices from $A$ is $\geq 2\epsilon T$, the probability that $f$ samples no triangles involving edges or vertices from $A$ must be $\leq 2\delta$.
    \end{lemma}

    \begin{proof} 
      Suppose this did not hold for some $A$. Then we can give the algorithm as input $G$ with probability $\frac{1}{2}$ and $G \setminus A$ with probability $\frac{1}{2}$. Suppose $f$ samples no triangles involving $A$. Then, as $f$ is triangle-dependent, the output of the algorithm is independent of which of the two inputs the algorithm received. So there is at most a $\frac{1}{2}$ chance the algorithm will return a value within $\epsilon T(G)$ of its input. So this occurs with probability at most $2\delta$.
    \end{proof}

    We can now prove several lower bounds for such algorithms. Throughout, we will use the fact that for any sampling algorithm that is resilient to permutations of the input, the probability of sampling an edge can be taken to the same for all edges, or copies of any other fixed subgraph, without loss of generality. We will use $m$ to refer to $|E(G)|$ throughout.

    \begin{lemma}
      \label{tdindep}
      Let $f$ be a triangle-dependent sampling algorithm that solves triangle-counting for a graph $G$ with probability $\frac{9}{10}$ and error $\epsilon$. Then $f$ samples $\Omega\left(\frac{m}{T^{2/3}}\right)$ edges in expectation.
    \end{lemma}
    \begin{proof}
      This is a direct consequence of Lemma~\ref{lem:ttwothirds}.
    \end{proof}

    \begin{lemma}
      \label{tdhubs}
      Let $f$ be a triangle-dependent sampling algorithm that solves triangle-counting a graph $G$ with probability $\frac{9}{10}$ and error $\epsilon$. Then $f$ samples $\Omega\left( m\frac{ \sqrt{\Delta_{V,2\epsilon}}}{T}\right)$ edges in expectation.
    \end{lemma}
    \begin{proof}
    Our proof will make use of the fact that $f$ must sample at least one triangle from the $2\epsilon T$ triangles at the heaviest vertices. We will also use the fact that, as $f$ needs to work for an arbitrary relabelling of the vertices, and because our sampling strategy is not allowed to vary based on which edges we've already seen, our sampling strategy will be the same for each vertex.

      Let $n$ be the number of vertices in $G$. Let $X_v$ denote the distribution on how many of the $n - 1$ edges that could be incident to $v$ will be sampled by $f$ if they are in $G$. As we require our algorithm to work on all permutations of the vertices, $X_v$ can be taken to be identically distributed for each $v \in V(G)$. Then, conditioned on $X_v = i$, $\Theta( i^2)$ of the $\Theta(n^2)$ potential wedges next to $v$ will be sampled if present. So the probability of sampling any given triangle that uses $v$ is less than the chance of sampling its two edges incident to $v$, which is in turn $\Theta(i^2/n^2)$.
      Hence the chance of sampling a triangle at $v$ is at most $\Theta(\frac{i^2 T_v}{n^2})$.

      Let the vertices $v$ of $G$ be ordered as $(v_i)_{i \geq 0}$ in descending order of $T_v$. Let $A$ be the minimal prefix of $(v_i)_{i \geq 0}$ such that $\sum_{v \in A} T_v \geq 2\epsilon T$. Then as $f$ is triangle-dependent, it must sample at least one triangle involving a vertex in $A$ with probability at least $\frac{4}{5}$.

      Conditioned on $X_v = i$, the chance of sampling a triangle at $v$ is at most $\Theta(\frac{i^2 T_v}{n^2})$, reaching $\Theta(1)$ when $\frac{i}{n} = \Theta(\sqrt{T_v}^{-1})$. So, letting $Y$ be the number of vertices $v \in A$ such that at least one triangle is sampled at $v$, $\E\lbrack Y| X_v = i_v \rbrack \lesssim \sum_{v \in A} \min \left\lbrace 1, \frac{i_v^2 T_v}{n^2} \right\rbrace$, and so \[ %
      \E \lbrack Y \rbrack  \lesssim \sum_{v \in A} \E \left \lbrack \min \left\lbrace 1, \frac{X_v^2 T_v}{n^2} \right\rbrace \right\rbrack
      \]
      As the $X_v$ are identically distributed, $\E\lbrack X_v^2\rbrack$ is the same for all $v$. Given the constraints that $\sum_{v\in A} T_v = O(T)$ and $\forall v \in A, T_v \geq  \Delta_{V,\epsilon}$, the sum above will be maximised when there are $O(T/ \Delta_{V,\epsilon})$ vertices in $A$, each with $T_v  = \Delta_{V,\epsilon}$. (as this is the configuration that minimizes the truncation of the tails.) So, \[ %
       \E \lbrack Y \rbrack  \lesssim \frac{T}{\Delta_{V,\epsilon}} \E \left \lbrack \min \left\lbrace 1, \frac{X_v^2 \Delta_{V,\epsilon}}{n^2} \right\rbrace \right\rbrack
       \]
       Now, for $\E \lbrack Y \rbrack$ to be $\Omega(1)$, we need \[
       \E \left \lbrack \min \left\lbrace 1, \frac{X_v^2 \Delta_{V,\epsilon}}{n^2} \right\rbrace \right\rbrack \gtrsim \frac{\Delta_{V,\epsilon}}{T}
       \]
     This, in turn, requires that $\E \lbrack X_v \rbrack$ be at least (up to constants) $\frac{\Delta_{V,\epsilon}}{T}\frac{n}{\sqrt{\Delta_{V,\epsilon}}} = \frac{n\sqrt{\Delta_{V,\epsilon}}}{T}$. (as we minimize $\E \lbrack X_v \rbrack$ by setting $X_v$ to $\frac{n}{\sqrt{\Delta_{V,\epsilon}}}$ with probability $\frac{\Delta_{V,\epsilon}}{T}$, and 0 otherwise.) 
      
      As, at each vertex, there are $n$ possible edges, this implies sampling edges with $\Omega \left(\frac{\sqrt{\Delta_{V,\epsilon}}}{T}\right)$ probability, and therefore sampling \[
      \Omega\left( m\frac{ \sqrt{\Delta_{V,2\epsilon}}}{T}\right)
      \] edges in expectation. As we need $\E \lbrack Y \rbrack = \Omega(1)$ to sample at least one triangle with $4/5$ probability, this completes our proof.

    \end{proof}

    \begin{lemma}
      \label{tdedges}
      Let $f$ be a triangle-dependent sampling algorithm that solves triangle-counting for a graph $G$ with probability $\frac{9}{10}$ and error $\epsilon$. Then $f$ samples $\Omega\left( m\frac{ \Delta_{E,2\epsilon}}{T}\right)$ edges in expectation.
    \end{lemma}
    \begin{proof}
      Let the edges $e$ of $G$ be ordered as $(e_i)_{i \geq 0}$ in descending order of $T_e$. Let $A$ be the minimal prefix of $(e_i)_{i \geq 0}$ such that $\sum_{e \in A} T_e \geq 2\epsilon T$. Then as $f$ is triangle-dependent, it must sample at least one triangle involving an edge in $A$ with probability $\frac{4}{5}$.

      So as $\Delta_{E,\epsilon} \leq T_e$ for all $e \in A$, $|A| \lesssim \frac{T}{\Delta_{E, \epsilon}}$. So the algorithm must sample at a rate $\gtrsim \frac{\Delta_{E,\epsilon}}{T}$. 

    \end{proof}

    \restate{thm:tdbound}
    \begin{proof}
      By combining the above three bounds.
    \end{proof}

    \section{Refining the Upper Bound}

    \restate{dfn:dgepsilon}

    \restate{thm:refined}
    \begin{proof}
      As in the proof of Theorem~\ref{thm:main}, the variance of the
      output $\overline{T}$ of a single iteration of the algorithm will be
      $O(T\wt{T})$ if $k = \min\left\lbrace \wt{T}^{2/3},
      \frac{\wt{T}^\frac{3}{2}}{\sqrt{\Delta_V}},
      \frac{\wt{T}^2}{\Delta_E}\right\rbrace$.

      Now define the graph $G'$ as follows: order the vertices $v \in V(G)$ as $(v_i)_{i \geq 0}$ in descending order of $T_v$, and let $V'$ be the minimal postfix of $(v_i)_{i \geq 0}$ such that $\sum_{v \in V(G) \setminus V'} T_v \leq \frac{\epsilon}{24} T$. Order the edges $e \in E(G)$ as $(e_i)_{i \geq 0}$ in descending order of $T_e$, and let $E'$ be the minimal postfix of $(e_i)_{i \geq 0}$ such that $\sum_{e \in E(G) \setminus E'} T_e \leq \frac{\epsilon}{24} T$. $G'$ is $(V', E' \cap (V' \times V'))$.

      Let $\overline{T_\epsilon}$ be $\overline{T}$, less any contributions \emph{in the second pass} from triangles not contained in $G'$. We can then compare this to $\overline{T}(G')$, the result the algorithm would have given if run with $G'$ as input. The only difference between $\overline{T_\epsilon}$ and $\overline{T}(G')$ will come from triangles counted in the \emph{first} pass. As the expectation of $\overline{T}(G')$ is independent of the output of the first pass, $\mathbb{E}\left\lbrack \overline{T_\epsilon} \right\rbrack = \mathbb{E}\left\lbrack \overline{T}(G') \right\rbrack = T(G') = (1 - \epsilon')T$ for some $\epsilon' \in \lbrack 0, \epsilon / 12 \rbrack$. Furthermore, as counting additional triangles in the second pass either has no effect or increases the estimate of $T$, $\overline{T} - \overline{T}_\epsilon$ will always be $\geq 0$.

      We now seek to bound the variance of $\overline{T}_\epsilon$. We do so by comparison to the variance of $\overline{T}(G')$. As the execution of the two algorithms differs only in the first pass, and therefore in the values of the ``triangle degree estimates'' for each vertex, we need consider only how this impacts lemmas \ref{loerror}, \ref{hierror}, and \ref{derror}, the three results about the distributions of these estimates on which the variance depends. This will allow us to show that the proof of our bound on $\var(\overline{T}(G'))$ is, within a constant factor, a bound on $\var(\overline{T}_\epsilon)$.

      For clarity, we will use $\mathcal{T}_v$ to refer to the estimate for the vertex $v$ in the calculation of $\overline{T}_\epsilon$ (as it is identical to the estimate calculated for $\overline{T}$), and $\mathcal{T}_v'$ for the estimate used in the calculation of $\overline{T}$. As $\mathcal{T}_v$ is monotonic decreasing in the number of first-pass triangles counted at $v$, $\mathcal{T}_v \leq \mathcal{T}_v'$ for all v. 

      We note that $\mathcal{T}_v \leq \mathcal{T}_v'$ for all $v$, as $\mathcal{T}_v$ is monotonic decreasing in the number of first-pass triangles counted at $v$. Furthermore, if $\mathcal{T}_v$ is $\lo$, so is $\mathcal{T}_v'$. Therefore, Lemmas \ref{loerror} and \ref{derror} will hold without further work. Lemma \ref{hierror}, however, gives us a bound on $\mathcal{T}_v'$ not being $\lo$, in terms of the parameters of $G'$, which does not directly transfer to a bound on $\mathcal{T}_v$.

      \restate{hierror}

      The lemma does, however, give us a bound on $\mathcal{T}_v$ in terms of the parameters $T_v(G)$ and $\wt{T}$, rather than $T_v(G')$. The only place this lemma is used in bounding the $l = 1, u\not = v$ case in Lemma \ref{thvar}, and so substituting in this bound instead of the one we would normally benefit from, we replace the $T_E(G')k\frac{T(G')}{\wt{T}}$ term in the variance with 
      \begin{align*}
        \sqrt{k}T_E(G')\sum_{v \in V(G')} \mathbb{P}\left\lbrack \mathcal{T}_v \not= \lo \right\rbrack & \lesssim \sqrt{k}T_E(G')\sum_{v \in V(G')} \frac{\sqrt{k} T_v(G)}{\wt{T}} \\
        & = T_E(G')k\frac{T(G)}{\wt{T}}\\
        & \lesssim T_E(G')k\frac{T(G')}{\wt{T}}
      \end{align*}

      As $T(G') > (1 - \epsilon/12T(G))$. Therefore, we lose at most a constant factor in going from our bound on the variance of $\overline{T(G')}$ to bounding the variance of $\overline{T_\epsilon}$. So $\var(\overline{T}_\epsilon) \lesssim T\wt{T}$ if $k = \min\left\lbrace \wt{T}^{2/3}, \frac{\wt{T}^\frac{3}{2}}{\sqrt{\Delta_{V,\epsilon/24}}}, \frac{\wt{T}^2}{\Delta_{E, \epsilon/24}}\right\rbrace$. (as $\Delta_V(G') \leq \Delta_{V, \epsilon}(G)$ and $\Delta_E(G') \leq \Delta_{E, \epsilon}(G)$)

      We can therefore split $\overline{T}$ into two (not independent) random variables, $\overline{T}_\epsilon$ and $T - \overline{T}_\epsilon$, with $\mathbb{E}\left\lbrack \overline{T}_\epsilon\right\rbrack = (1 - \epsilon')T$, $\mathbb{E}\left\lbrack T - \overline{T}_\epsilon \right\rbrack = \epsilon' T$, and $\var(\overline{T}_\epsilon) \lesssim T\wt{T}$. Now, we can repeat the algorithm $c$ times, averaging the results as $\overline{T}^{(c)}$. We then have $\overline{T}^{(c)} = \overline{T}_\epsilon^{(c)} + (\overline{T}^{(c)} - \overline{T}_\epsilon^{(c)})$, where $\overline{T}_\epsilon^{(c)}$ is the average of the $c$ instances of $\overline{T}_\epsilon$. So $\mathbb{E}\left\lbrack \overline{T}_\epsilon^{(c)}\right\rbrack = (1 - \epsilon')T$, $\mathbb{E}\left\lbrack \overline{T}^{(c)} - \overline{T}_\epsilon^{(c)}\right\rbrack = \epsilon' T$, and $\var(\overline{T}_\epsilon^{(c)}) \lesssim \frac{1}{\sqrt{c}} T\wt{T}$. So we can choose $c = O\left(\frac{1}{\epsilon}\right)$ to get $\var(\overline{T}_\epsilon^{(c)}) \leq \frac{\epsilon}{4} T\wt{T}$.

      Then, by Chebyshev's inequality, $\overline{T}_\epsilon^{(c)}$ is within $\frac{\epsilon}{2}\sqrt{T\wt{T}}$ of $(1 - \epsilon')T$ with probability $\frac{3}{4}$, and by Markov's inequality (and the fact that it is non-negative), $\overline{T}^{(c)} - \overline{T}_\epsilon^{(c)}$ is $< 6\epsilon'T \leq \frac{\epsilon}{2}T$ with probability $\frac{5}{6}$. So by taking a union bound, both of those hold with probability $\frac{2}{3}$. So if $\wt{T} \leq T$, $\overline{T}$ will be $(1 \pm \epsilon)T$ with probability $\frac{2}{3}$, and if $\wt{T} \geq T$, $\overline{T}$ will be within $\epsilon\wt{T}$ of $T$ with probability $\frac{2}{3}$. (so, by responding that $T < \wt{T}$ if $\overline{T} < (1-\epsilon)\wt{T}$, and returning $\overline{T}$ otherwise, it will either determine $T < \wt{T}$ correctly or return a $(1 \pm \epsilon)$ approximation to $T$)

      Then, by repeating this process $O\left(\log\frac{1}{\delta}\right)$ times and taking the median, we can guarantee that the above will hold with probability $1 - \delta$ instead, giving us our final result.
    \end{proof}

    \section{Counting Constant Size Subgraphs}
    We now consider the problem of counting subgraphs of a constant size in our graph stream. Let $A$ be a fixed subgraph, with $s = |A|$. We will attempt to estimate $M$, the number of subgraphs of $G$ that are isomorphic to $A$.

    We will use the same algorithm as for triangles, slightly generalized. $\widetilde{M}$ will be a lower bound on $M$. $\forall S \subseteq V(G), \mu(S)$ is the set of copies of $A$ in $G$ that contain $S$, $M_S = |\mu(S)|$, and $\forall i \in \lbrack s\rbrack$, $C_i = \sum_{S \subseteq V(G), |S| = i} T_S^2$. With $C_i^+$ as an upper bound on $C_i$, $\omega$ is then defined, analogously to the triangle case, as $\min\left\lbrace \frac{\left(\widetilde{M}\right)^2}{C_1^+}, \sqrt{k}\right\rbrace$.
    \subsection{First Pass}
    For $i = 1,2$:
    Let $r_{D,i} : E \rightarrow \lbrack 0,1\rbrack$ be a uniformly random hash function.

    Let $\mathcal{M}_i : V \rightarrow \lbrace 0,1\rbrace$ be a random hash function such that $\forall v \in V$:

    \begin{align*}
      \mathcal{M}_i(v) = \begin{cases}
        1 & \mbox{With probability $\frac{1}{\sqrt{k}}$.}\\
        0 & \mbox{Otherwise.}
      \end{cases}
    \end{align*}

    For each $v$ in $V(G)$, $a \in \mu(\lbrace v\rbrace)$, and $h \in \left\lbrace 0, \dots, \left\lceil \log \frac{\sqrt k}{\omega} \right\rceil\right\rbrace$, define $X_{v,a,i}^{(h)}$ to be 1 if the following hold:

    \begin{align*}
      \forall u \in V(a) \setminus \lbrace v\rbrace, & \mathcal{M}_i(u) = 1\\
      \forall \left\lbrace u \middle| vu \in E(a) \right\rbrace, & r_{D,i}(vu) < \frac{\omega 2^h}{\sqrt k}
    \end{align*}

    And 0 otherwise. We define $x_{v,a,i}^{(h)} = \mathbb{E}\left\lbrack X_{v,a,i}^{(h)}\right\rbrack$, so if $t = |\left\lbrace vu \middle| uv \in E(a) \right\rbrace|$, $x_{v,a,i}^{(h)} = \frac{\omega^t2^{ht}}{k^{\frac{s + t -1}{2}}}$. Then let $X_{v,i}^{(h)} = \sum_{a \cong A, v \in V(a)} \frac{X_{v,a,i}^{(h)}}{x_{v,a,i}^{(h)}}$.

    We then define $H_{v,i} = \left\lbrace h \in \left\lbrace 0, \dots, \left\lceil \log \frac{\sqrt k}{\omega} \right\rceil\right\rbrace \middle| X_{v,i}^{(h)} \geq \frac{\widetilde{M}}{\omega 2^h}\right\rbrace$. $\mathcal{M}_{v,i}$ is then defined as follows:

    \begin{align*}
      \mathcal{M}_{v,i} = \begin{cases}
        \min H_{v,i} & \mbox{If $H_{v,i} \not= \emptyset$.}\\
        \lo & \mbox{Otherwise.}
      \end{cases}
    \end{align*}

    Then, for each $v \in V(G)$, we define $\mathcal{M}_v$ to be $\lo$ if $\mathcal{M}_{v,1}$ and $\mathcal{M}_{v,2}$ are $\lo$, and otherwise to be the smallest numerical value amongst $\mathcal{M}_{v,1}, \mathcal{M}_{v,2}$.

    \subsection{Second Pass}
    \subsubsection{Splitting the Graph}
    Let $V_{L} = \lbrace v \in V | \mathcal{M}_v = \lo\rbrace$, and let $G_L$ be the subgraph induced by $V_L$. Then we define $M_L$ as the number of copies of $A$ in $G_L$, and $M_H = M - M_L$.

    We will compute estimates $\overline{M_L},\overline{M_H}$ of $M_L,M_H$, and estimate $M$ as $\overline{M} = \overline{M_L} + \overline{M_H}$.

    \subsection{Estimating \texorpdfstring{$M_L$}{M\_L}}
    We will estimate $M_L$ by sampling the vertices of $V_L$ with probability $\frac{1}{\sqrt k}$ each and calculating the number of copies of $A$ in the resulting graph.

    Let $c : V \rightarrow \lbrace 0,1\rbrace$ be a random hash function such that $\forall v \in V$:

    \begin{align*}
      c(v) = \begin{cases}
        1 & \mbox{With probability $\frac{1}{\sqrt k}$.}\\
        0 & \mbox{Otherwise.}
      \end{cases}
    \end{align*}

    Let $V_L' = \lbrace v \in V_L | c(v) = 1\rbrace$, and let $G_L'$ be the subgraph of $G$ induced by $V_L'$. Then $\overline{M_L}$ is the number of copies of $A$ in $G_L'$, multiplied by $k^\frac{s}{2}$.

    \subsection{Estimating \texorpdfstring{$M_H$}{M\_H}}
    We will estimate $M_H$ by considering every vertex in $V\setminus V_L$ separately. We will achieve this by sampling vertices $v$ with probability $2^{-\mathcal{M}_v}$, and then sampling edges incident to $v$ with probability proportional to $2^{\mathcal{M}_v}$.

    Let $h : V \rightarrow \lbrace 0,1\rbrace$ be a random hash function such that $\forall v \in V$:

    \begin{align*}
      h(v) = \begin{cases}
        i & \mbox{With probability $\frac{1}{\omega2^i}$ for each $i \in \lbrace 0,\dots,\left\lceil\log \frac{\sqrt k}{\omega}\right\rceil\rbrace$.}\\
        -\infty & \mbox{Otherwise.}
      \end{cases}
    \end{align*}

    And let $r_C : V \rightarrow \lbrack 0,1\rbrack$ be a uniformly random hash function.

    Then, for each $v\in V_H = V\setminus V_L$, we allow $v$ to contribute to $\overline{M_L}$ iff $h(v) = \mathcal{M}_v$. If it does, we calculate its contribution $\overline{M_v}$ as follows:

    We count a subgraph $a \subseteq$ such that $a \cong A$ and $v \in V(a)$ iff:

    \begin{align*}
      \forall u \in V(a) \setminus \lbrace v\rbrace, & c(u) = 1\\
      \forall e \in \left\lbrace u \middle| vu \in E(a) \right\rbrace, & r_C(e) < \frac{\omega 2^h}{\sqrt k}
    \end{align*}

    Then, for each such $a$, let $t = |\left\lbrace vu \middle| u \in E(a) \right\rbrace|$. We add $\frac{k^{\frac{s + t - 1}{2}}}{\omega^{t-1} 2^{\mathcal{M}_v(t-1)}} \times \frac{1}{|\lbrace x \in V(a) | \mathcal{M}_x \not = \lo\rbrace |}$ to $\overline{M_v}$. (with the second term then compensating for the fact that a motif could potentially be counted at multiple different vertices) 

    We then define $\overline{M_H} = \sum_{v \in V_H,\ h(v) = \mathcal{M}_v} \overline{M_v}$.

    \subsection{Final Output}
    We then output our estimate $\overline{M}$ of $M$ as $\overline{M} = \overline{M_L} + \overline{M_H}$.

    \section{Analysis of Generalized Algorithm}
    \subsection{First Pass}
    \subsubsection{Relation of \texorpdfstring{$\mathcal{M}_v$}{dv} to \texorpdfstring{$M_v$}{Mv}}
    \begin{lemma} 
      \label{xvariance2}
      For any $v\in \lbrack n\rbrack$, let $X_{v}^{(h)}$ be as defined previously. Then, $\mathbb{E}\left\lbrack X_{v}^{(h)} \right\rbrack = M_v$ and $\var\left(X_{v}^{(h)}\right) \leq \sum_{l = 2}^s \sum_{S \subseteq V(G), |S| = l, v \in S} M_S^2 \left(\frac{k}{\omega2^h}\right)^{l - 1}$.
    \end{lemma}
    \begin{proof}
      $X_{v}^{(h)} = \sum_{a \in \mu(\lbrace v\rbrace)} \frac{X_{v,a}^{(h)}}{\mathbb{E}\left\lbrack X_{v,a}^{(h)}\right\rbrack}$, so $\mathbb{E}\left\lbrack X_{v}^{(h)}\right\rbrack = M_v$.

      Then, to bound the variance, we start by bounding $\mathbb{E}\left\lbrack \left(X_{v}^{(h)}\right)^2\right\rbrack = \sum_{a_1,a_2 \in \mu(\lbrace v\rbrace)}\frac{\mathbb{E}\left\lbrack X_{v,a_1}^{(h)}X_{v,a_2}^{(h)}\right\rbrack}{x_{v,a_1}^{(h)}x_{v,a_2}^{(h)}}$. We will consider the contribution of these terms for each possible value of $l = |V(a_1) \cap V(a_2)|$. (noting that, for any such subgraph $a$, if $t = |\left\lbrace vu \middle| u \in E(a) \right\rbrace|$, $t \leq s - 1$):

      Letting $t_i = |\left\lbrace vu \middle| vu \in E(a_i) \right\rbrace|$ for $i = 1,2$,  $x_{v,a_i}^{(h)} = \frac{\omega^{t_i}2^{ht_i}}{k^{\frac{s + t_i -1}{2}}}$. 

      Then, for  $X_{v,a_1}^{(h)}X_{v,a_2}^{(h)} = 1$ to hold, we need the $2s - l - 1$ vertices  $u \in V(a_1) \cup V(a_2) \setminus \lbrace v\rbrace$ to have $d(u) =1$, and the $\geq t_1 + t_2 - (l - 1)$ edges $e \in \left\lbrace vu \middle| vu \in E(a_1) \right\rbrace \cup \left\lbrace vu \middle| vu \in E(a_2) \right\rbrace$ to have $r_D(e) < \frac{\omega 2^h}{\sqrt k}$. So $\mathbb{E}\left\lbrack X_{v,a_1}^{(h)}X_{v,a_2}^{(h)}\right\rbrack \leq \frac{(\omega2^h)^{t_1 + t_2 + 1 - l}}{k^{(s - l) + \frac{t_1 + t_2}{2}}}$. 

      So $\frac{\mathbb{E}\left\lbrack X_{v,a_1}^{(h)}X_{v,a_2}^{(h)}\right\rbrack}{x_{v,a_1}^{(h)}x_{v,a_2}^{(h)}} \leq \left(\frac{k}{\omega2^h}\right)^{l - 1}$. There are $\leq \sum_{S \subseteq V(G), |S| = l, v \in S} M_S^2$ such pairs for each $l \in \lbrace 2,s\rbrace$, and so the total contribution to the sum is $\leq \sum_{S \subseteq V(G), |S| = l, v \in S} M_S^2 \left(\frac{k}{\omega2^h}\right)^{l - 1}$.

      This gives us:

      \begin{align*}
        \var\left(X_{v}^{(h)}\right) & = \mathbb{E}\left\lbrack \left(X_{v}^{(h)}\right)^2\right\rbrack - \mathbb{E}\left\lbrack \left(X_{v}^{(h)}\right)^2\right\rbrack\\
        & \leq  \sum_{i =1}^s \sum_{S \subseteq V(G), |S| = l, v \in S} M_S^2 \left(\frac{k}{\omega2^h}\right)^{l - 1} - M_v^2\\
        & =  \sum_{l = 2}^s \sum_{S \subseteq V(G), |S| = l, v \in S} M_S^2 \left(\frac{k}{\omega2^h}\right)^{l - 1}
      \end{align*}

    \end{proof}
    \begin{lemma}
      \label{loerror2}
      If $M_v \geq 2\frac{\widetilde{M}}{\sqrt{k}}$, then $\mathbb{P}\left\lbrack \mathcal{M}_v = \lo \right\rbrack \lesssim \frac{1}{M_v^2}\sum_{l = 2}^s \sum_{S \subseteq V(G), |S| = l, v \in S} M_S^2 k^\frac{l - 1}{2}$.

    \end{lemma}
    \begin{proof}
      For $\mathcal{M}_v = \lo$ to hold, we need that $X_{v}^{(h)} < \frac{\widetilde{M}}{\omega 2^h}$ for all $h \in \left\lbrack\left\lceil\log \frac{\sqrt k}{\omega}\right\rceil\right\rbrack$, and so in particular $X_{v}^{\left(\left\lceil\log \frac{\sqrt k}{\omega}\right\rceil\right)} < \frac{\widetilde{M}}{\sqrt{k}}$.

      By Lemma \ref{xvariance2}, $\mathbb{E}\left\lbrack X_{v}^{\left(\left\lceil \log \frac{\sqrt k}{\omega}\right\rceil\right)}\right\rbrack = M_v \geq 2 \frac{\widetilde{M}}{\sqrt k}$, and \[
      \var\left(  X_{v}^{\left(\left\lceil \log \frac{\sqrt k}{\omega}\right\rceil\right)}\right) \lesssim \sum_{l = 2}^s \sum_{S \subseteq V(G), |S| = l, v \in S} M_S^2 k^\frac{l - 1}{2}\]
      . So by Chebyshev's inequality, this holds with probability $\lesssim \frac{1}{M_v^2}\sum_{l = 2}^s \sum_{S \subseteq V(G), |S| = l, v \in S} M_S^2 k^\frac{l - 1}{2}$.
    \end{proof}

    \begin{lemma}
      \label{hierror2}
      $\forall v \in V, \mathbb{P}\left\lbrack \mathcal{M}_v \not = \lo \right\rbrack \lesssim \frac{\sqrt{k} M_v}{\widetilde{M}}$.
    \end{lemma}
    \begin{proof}
      For $\mathcal{M}_v \not = \lo$ to hold, we need there to be at least one $h\in  \left\lbrack\left\lceil\log \frac{\sqrt k}{\omega}\right\rceil\right\rbrack$ such that $X_{v}^{(h)} \geq \frac{\widetilde{M}}{\omega 2^h}$. By Lemma \ref{xvariance2}, $\mathbb{E}\left\lbrack X_{v}^{(h)} \right\rbrack = M_v$,  so by Markov's inequality this occurs with probability $\leq \frac{M_v}{\widetilde{M}}\omega2^h$. So the probability that it holds for any $h$ is $\leq \sum_{h=\left\lceil \log \frac{\sqrt k}{\omega}\right\rceil}^1 \frac{M_v}{\widetilde{M}}\omega2^h \lesssim \sum_{i=0}^\infty\frac{M_v\sqrt{k}}{\widetilde{M}} 2^{-i}\lesssim \frac{\sqrt{k} M_v}{\widetilde{M}}$.
    \end{proof}

    \begin{lemma}
      \label{derror2}
      $\forall v \in V, l \geq 1, \left\lceil\log \frac{\widetilde{M}}{\omega M_v}\right\rceil \geq 1$, 
      \[\mathbb{P}\left\lbrack \mathcal{M}_v \geq \max\left\lbrace\left\lceil\log \frac{\widetilde{M}}{\omega M_v}\right\rceil + j, 0\right\rbrace\right\rbrack \lesssim \left(\frac{1}{M_v^2}\sum_{l = 2}^s \sum_{S \subseteq V(G), |S| = l, v \in S} M_S^2 \left(\frac{kM_v}{2^j\widetilde{M}}\right)^{l - 1}\right)^2\]
    \end{lemma}
    \begin{proof}
      For $ \mathcal{M}_{v,i} \geq \max\left\lbrace\left\lceil\log \frac{\widetilde{M}}{\omega M_v}\right\rceil + j, 0\right\rbrace$ to hold, $X_{v}^{\max\left\lbrace\left\lceil\log \frac{\widetilde{M}}{\omega M_v}\right\rceil + j - 1, 1\right\rbrace} < \frac{M_v}{2}$ must hold. By Lemma \ref{xvariance2} and Chebyshev's inequality, this happens with probability \[\lesssim \frac{1}{M_v^2}\sum_{l = 2}^s \sum_{S \subseteq V(G), |S| = l, v \in S} M_S^2 \left(\frac{kM_v}{2^j\widetilde{M}}\right)^{l - 1}\]. So it holds for both $i = 1,2$, and therefore for $\mathcal{M}_v$, with probability \[\lesssim \left(\frac{1}{M_v^2}\sum_{l = 2}^s \sum_{S \subseteq V(G), |S| = l, v \in S} M_S^2 \left(\frac{kM_v}{2^j\widetilde{M}}\right)^{l - 1}\right)^2\]
    \end{proof}

    \subsection{Second Pass}
    \subsubsection{\texorpdfstring{$\overline{M_L}$}{M\_L}}
    \begin{lemma}
      \label{tlexpec2}
      $\mathbb{E}\left\lbrack \overline{M_L} | M_L \right\rbrack = M_L$.
    \end{lemma}
    \begin{proof}
      $M_L$ is the number of copies of $A$ in the random graph $G_L$, and $\overline{M_L}$ is $k^\frac{s}{2}$ times the number of triangles in the random graph $G_L'$. Any subgraph $a \subseteq G$ is in $G_L'$ iff $\forall v \in V(a), c(v) = 1$, which occurs with probability $\frac{1}{k^\frac{s}{2}}$ for subgraphs $a$ such that $a \cong A$. So the expected number of copies in $G_L'$ is $\frac{M_L}{k^\frac{s}{2}}$, and so $\mathbb{E}\left\lbrack \overline{M_L} | M_L \right\rbrack = M_L$.
    \end{proof}
    \begin{lemma}
      \label{tlvar2}
      $\var(\overline{M_L}) \lesssim M^2 + \sum_{l = 2}^sC_lk^\frac{l}{2}$.
    \end{lemma}
    \begin{proof}
      Let $M_L'$ be the number of copies of $A$ in $G_L'$. (so $\overline{M_L} = k^\frac{s}{2} M_L'$) A subgraph $a \cong A$ in $G$ is in $G_L'$ iff $\forall v \in V(a), c(v) = 1$ and $\mathcal{M}_v =  \lo$. We proceed by bounding $\mathbb{E}\left\lbrack M_L'^2\right\rbrack$, by considering, for any pair of subgraphs $a_1,a_2 \cong A$, the probability they are both in $G_L'$. We will bound this for each value of $l = V(a_1) \cap V(a_2)$, treating $l = 1$ as a special case.

      \begin{description}
        \item[$l \not = 1$:] $\mathbb{P}\left\lbrack\forall v \in V(a), c(v) = 1 \right\rbrack = k^{-\frac{|V(a_1) \cap V(a_2)|}{2}} = k^{-\frac{2s - l}{2}}$. So the total contribution to $\mathbb{E}\left\lbrack M_L'^2\right\rbrack$ from these pairs is $\leq C_l k^\frac{l-2s}{2}$.

        \item[$l = 1$:] Let $v$ be the unique vertex in $V(a_1) \cap V(a_2)$. Then in this case, we will also make use of the fact that $\mathcal{M}_v = \lo$ must hold. As above, $\mathbb{P}\left\lbrack\forall v \in V(a), c(v) = 1 \right\rbrack = k^{-\frac{2s - 1}{2}}$.

          So for vertices $v$ s.t. $M_v \leq \frac{2\widetilde{M}}{\sqrt{k}}$, we can bound the contribution to $\mathbb{E}\left\lbrack M_L'^2\right\rbrack$ from these cases by $k^{-\frac{2s - 1}{2}}\sum_{v \in V(G)} M_v^2 \lesssim k^{-\frac{2s - 1}{2}}\frac{\left(\widetilde{M}\right)^2}{\sqrt k} = k^{-s}\left(\widetilde{M}\right)^2 \leq k^{-s} M^2$.

          We can then consider vertices $v$ s.t. $M_v > \frac{2\widetilde{M}}{\sqrt{k}}$. By Lemma \ref{loerror2}, \[
          \mathbb{P}\left\lbrack \mathcal{M}_v = \lo \right\rbrack \lesssim \frac{1}{M_v^2}\sum_{i = 2}^{s} \sum_{S \subseteq V(G), |S| = i, v \in S} M_S^2 k^{\frac{i - 1}{2}}\]. 
          As $\mathcal{M}_v$ is independent of the hash function, $c$, the probability of $a_1$ and $a_2$ both being in $G_L'$ is $\lesssim \frac{1}{M_v^2}\sum_{i = 2}^{s} \sum_{S \subseteq V(G), |S| = i, v \in S} M_S^2 k^\frac{i - 2s}{2}$, and so the total contribution to $\mathbb{E}\left\lbrack M_L'^2\right\rbrack$ is $\lesssim \sum_{v \in V(G)}\sum_{i = 2}^{s} \sum_{S \subseteq V(G), |S| = i, v \in S} M_S^2 k^\frac{i - 2s}{2} \lesssim \sum_{i = 2}^{s} C_i k^\frac{i - 2s}{2}$.
      \end{description}

      So by summing these cases together:

      \begin{align*}
        \mathbb{E}\left\lbrack M_L'^2\right\rbrack & \lesssim k^sM^2 + \sum_{l = 2}^{s} C_i k^\frac{l - 2s}{2}
      \end{align*}

      Which then gives us our variance bound:

      \begin{align*}
        \var(\overline{M_L}) & \leq \mathbb{E}\left\lbrack M_L^2\right\rbrack\\
        & \leq k^s\mathbb{E}\left\lbrack M_L'^2\right\rbrack\\
        & \lesssim M^2 + \sum_{l = 2}^{s} C_l k^\frac{l}{2}
      \end{align*}

    \end{proof}

    \subsubsection{\texorpdfstring{$\overline{M_H}$}{M\_H}}
    \begin{lemma}
      \label{thexpec2}
      If $\wt{M} \leq M$, $\mathbb{E}\left\lbrack \overline{M_H} | M_H \right\rbrack = M_H$.
    \end{lemma}
    \begin{proof}
      $\overline{M_H} = \sum_{v\in V_H} \overline{M_v}$, and $M_H$ is the number of copies of $A$ in $G$ that are not in $G_L$. If a subgraph in $G$ is in $G_L$, it can never contribute to a $T_v$, as every vertex $u$ it uses will have $\mathcal{M}_u = \lo$. If a subgraph $a \cong A$ in $G$ is not in $G_L$, it has $l \in \lbrack s\rbrack$ vertices $u$ s.t. $\mathcal{M}_u \not = \lo$. It will therefore have $l$ vertices $v$ such that it can contribute to $M_v$.

      At each of those vertices $v$, a subgraph $a \cong A$ with $t = |\lbrace u | vu \in E(a)\rbrace|$ will contribute $\frac{k^\frac{s + t - 1}{2}}{\omega^{t - 1} 2^{\mathcal{M}_v(t - 1)}}\frac{1}{l}$ to $\overline{M_v}$ iff:

      \begin{align*}
        & h(v) = \mathcal{M}_v\\
        \forall u \in V(a) \setminus \lbrace v\rbrace, & c(u) = 1\\
        \forall e \in \left\lbrace u \middle| vu \in E(a) \right\rbrace, & r_C(e) < \frac{\omega 2^{\mathcal{M}_v}}{\sqrt k}
      \end{align*}

      This happens with probability $\frac{\left(\omega 2^\mathcal{M}_v\right)^t}{\omega2^{\mathcal{M}_v}k^\frac{s + t - 1}{2}} $, so the expected contribution of $a$ to $\overline{M_v}$ is $\frac{1}{l}$. Therefore, as there are $l$ vertices where $a$ can contribute, its expected contribution to $\overline{M_H}$ is 1.

      Therefore,  $\mathbb{E}\left\lbrack \overline{M_H} | M_H\right\rbrack$ is precisely the number of copies of $A$ in $G$ that are not in $G_L$, which is $M_H$.
    \end{proof}

    \begin{lemma}
      \label{thvar2}
      If $\wt{M} \leq M$, $\var(\overline{M_H}) \lesssim M^2 + \sum_{l = 2}^s C_l \frac{k^{l-1}}{\omega^{l-2}}$
    \end{lemma} 
    \begin{proof}
      We now care, for any subgraph $a \cong A$, which vertex we are counting it at (and so which $M_v$ it may contribute to. We will therefore use $a^v$ to denote the subgraph $a$, counted at the vertex $v$. With $t = |\lbrace u | vu \in E(a)\rbrace|$, we then define $Y_{a^v}$ as $\frac{k^\frac{s + t - 1}{2}}{\omega^{t - 1} 2^{\mathcal{M}_v(t - 1)}}$ if $a$ is counted at $v$ and 0 otherwise.

      Then, $\overline{M_v} \leq \sum_{a \subseteq G, a \cong A, v \in V(a)} Y_{a^v}$. So with $Y = \sum_{a,v , a \cong A, v \in V(a)} Y_{a^v}$, $\mathbb{E}\left\lbrack \overline{M}^2 \right\rbrack \leq \mathbb{E}\left\lbrack Y^2 \right\rbrack$. We will bound  $ \mathbb{E}\left\lbrack Y^2 \right\rbrack$ by bounding $\mathbb{E}\left\lbrack Y_{a_1^u}Y_{a_2^v} \right\rbrack$ for each pair $a_1^u, a_2^v$.

      We will bound these with respect to $l = |V(a_1) \cap V(a_2)|$, treating $l = 1$ as a special case, and treating $u = v$ and $u \not = v$ separately. In each case, let $t_1 = |\left\lbrace uw \middle| uw \in E(a_1) \right\rbrace|, t_2 = |\left\lbrace vw \middle| vw \in E(a_2) \right\rbrace|$.

      \begin{description}
        \item[$l \not = 1, u \not = v$:] We need $\mathcal{M}_u = h(u)$, $\mathcal{M}_v = h(v)$, which happens with probability $\omega^{-2}2^{-\mathcal{M}_u - \mathcal{M}_v}$. We need $\forall w \in V(a_1) \setminus \lbrace u\rbrace \cup V(a_2) \setminus \lbrace v\rbrace, c(w) = 1$. There are at least $2s - l - 2$ vertices in this set, so this happens with probability $\leq k^{-s + 1 +\frac{l}{2}}$.  We need $\forall e \in \left\lbrace uw \middle| uw \in E(a_1) \right\rbrace r_C(e) < \frac{\omega 2^{\mathcal{M}_u}}{\sqrt k}$ and $\forall e \in \left\lbrace vw \middle| vw \in E(a_2) \right\rbrace, r_C(e) < \frac{\omega 2^{\mathcal{M}_v}}{\sqrt k}$. These sets can overlap in at most one edge ($uv$), so this happens with probability $\leq \left(\frac{\omega}{\sqrt k}\right)^{t_1 + t_2 -1} 2^{t_1\mathcal{M}_u + t_2\mathcal{M}_v - \max\lbrace \mathcal{M}_u, \mathcal{M}_v\rbrace}$. Furthermore, the overlap in $uv$ can only occur if $u \in V(a_2)$ and $v\in V(a_1)$, and in this case we will also need $c(u) = 1$ and $c(v) = 1$. So this either reduces the probability by a factor of $\frac{\omega2^{\max \lbrace \mathcal{M}_u, \mathcal{M}_v\rbrace}}{\sqrt k}$ or $\frac{1}{k}$, and so in either case by at least a factor of $\frac{\omega2^{\max \lbrace \mathcal{M}_u, \mathcal{M}_v\rbrace}}{\sqrt k}$.

          So as these three conditions are independent, and multiplying by the values $Y_{a_1^u}, Y_{a_2^v}$ take when $a_1^u, a_2^v$ are counted, $\mathbb{E}\left\lbrack Y_{a_1^u}Y_{a_2^v}\right\rbrack \leq k^\frac{l}{2}$. So the contribution to $\mathbb{E}\left\lbrack Y^2\right\rbrack$ from such pairs is $\lesssim C_l k^\frac{l}{2}$.

        \item[$l \not = 1, u = v$:] We need $\mathcal{M}_v = h(v)$, which happens with probability $\omega^{-1}2^{- \mathcal{M}_v}$. We need $\forall w \in V(a_1) \setminus \lbrace u\rbrace \cup V(a_2) \setminus \lbrace v\rbrace, c(w) = 1$. There are at least $2s - l - 1$ vertices in this set, so this happens with probability $\leq k^{-s +\frac{l + 1}{2}}$. We need \[
          \forall e \in \left\lbrace uw \middle| uw \in E(a_1) \right\rbrace \cup \left\lbrace vw \middle| vw \in E(a_2) \right\rbrace, r_C(e) < \frac{\omega 2^{\mathcal{M}_v}}{\sqrt k}\]
          . As this set has at least $\max\lbrace t_1,t_2\rbrace$ elements, this happens with probability $\leq \left(\frac{\omega 2^{\mathcal{M}_v}}{\sqrt k}\right)^{\max\lbrace t_1,t_2\rbrace}$.

          So as these three conditions are independent, and multiplying by the values $Y_{a_1^u}, Y_{a_2^v}$ take when $a_1^u, a_2^v$ are counted, $\mathbb{E}\left\lbrack Y_{a_1^u}Y_{a_2^v}\right\rbrack \leq $\\$k^{\frac{l + \min\lbrace t_1,t_2\rbrace - 3}{2}}\omega^{1 - \min\lbrace t_1,t_2\rbrace}2^{1 -\mathcal{M}_v(\min\lbrace t_1,t_2\rbrace)}$. So as $\frac{\omega 2^{\mathcal{M}_v}}{\sqrt k} \leq 1$, and $t_1, t_2 \leq l -1$, this gives us $\mathbb{E}\left\lbrack Y_{a_1^u}Y_{a_2^v}\right\rbrack \leq \frac{k^{l -1}}{\omega^{l-2}2^{\mathcal{M}_v(l - 2)}} \leq \frac{k^{l -1}}{\omega^{l-2}}$. So the contribution to $\mathbb{E}\left\lbrack Y^2\right\rbrack$ from such pairs is $\lesssim C_l \frac{k^{l -1}}{\omega^{l-2}}$.

        \item[$l = 1, u \not = v$:] As in the $l \not = 1$ case, $\mathbb{E}\left\lbrack Y_{a_1^u}Y_{a_2^v}\right\rbrack \leq k^\frac{l}{2} = \sqrt{k}$. So we seek to bound the number of such pairs.

          For any $w \in V(G)$, the number of such pairs intersecting at $w$ is $\leq \left(\sum_{v \in V(G), \mathcal{M}_v \not = \lo} M_{wv}\right)^2$. Now let $L = |\lbrace v \in V(G) | \mathcal{M}_v \not= \lo\rbrace |$. Suppose $L = r$. By Cauchy-Schwartz, this means the number of such pairs intersecting at $w$ is $\leq r \sum_{v \in V(G)} M_{wv}^2$. So by summing across all $w$, the total number of such pairs is $\leq rC_2$. So the contribution to the expectation conditioned on $L = r$ is $\leq \sqrt{k} rC_2$. As our bound on $\mathbb{E}\left\lbrack Y_{a_1^u}Y_{a_2^v}\right\rbrack$ is independent of the values of the $\mathcal{M}_v$, we can then bound the unconditional contribution to $\mathbb{E}\left\lbrack Y^2\right\rbrack$ by:

          \begin{align*}
            \sum_r \sqrt{k}rC_2 \mathbb{P}\left\lbrack L = r \right\rbrack & = \mathbb{E}\left\lbrack L \right\rbrack \sqrt{k}C_2\\
            & = \sqrt{k}C_2\sum_{v \in V(G)} \mathbb{P}\left\lbrack \mathcal{M}_v = \lo\right\rbrack\\
            & = \sqrt{k}C_2\sum_{v \in V(G)}\frac{\sqrt{k}M_v}{\widetilde{M}} & \mbox{By Lemma \ref{hierror2}.}\\
            & = C_2k\frac{M}{\widetilde{M}}
          \end{align*}

        \item[$l = 1, u = v$:]  As in the $l \not = 1$ case, $\mathbb{E}\left\lbrack Y_{a_1^u}Y_{a_2^v}\right\rbrack \leq \frac{k^{l -1}}{\omega^{l-2}2^{\mathcal{M}_v(l -2 )}} = \omega 2^{\mathcal{M}_v}$. So at any vertex $v$, the contribution to the expectation, conditioned on $\mathcal{M}_v$, is $\leq M_v^2 \omega 2^{\mathcal{M}_v}$. 

          We will consider three (exhaustive, but not necessarily mutually exclusive) cases: $\mathcal{M}_v \leq \left\lceil\log\frac{\widetilde{M}}{\omega M_v}\right\rceil$, $\mathcal{M}_v = 0$, and $\mathcal{M}_v = \max\left\lbrace\left\lceil\log \frac{\widetilde{M}}{\omega M_v}\right\rceil + i, 0\right\rbrace$ for some $i \geq 1$.

          In the first case, $\mathbb{E}\left\lbrack Y_{a_1^u}Y_{a_2^v}\right\rbrack \leq  \frac{\widetilde{M}}{M_v}$, so the total contribution to the expectation from such vertices is $\leq \sum_{v \in V(G)} M_v\widetilde{M} \lesssim M^2$.

          In the second case, $\mathbb{E}\left\lbrack Y_{a_1^u}Y_{a_2^v}\right\rbrack \leq \omega$, so the total contribution to the expectation from such vertices is $\leq \sum_{v \in V(G)} M_v^2\omega = C_1\omega \leq M^2$.

          In the third case, $\mathbb{E}\left\lbrack Y_{a_1^u}Y_{a_2^v}\right\rbrack \leq  \frac{\widetilde{M}2^i}{M_v}$, and the probability of $\mathcal{M}_v$ being at least this high is:

          \begin{align*}
            & \lesssim \min\left\lbrace 1, \left(\frac{1}{M_v^2}\sum_{l = 2}^s \sum_{S \subseteq V(G), |S| = l, v \in S} M_S^2 \left(\frac{kM_v}{2^i\widetilde{M}}\right)^{l - 1}\right)^2 \right\rbrace & \mbox{By Lemma \ref{derror2}}
          \end{align*}

          Now let $x \in \mathbb{R}$ be the unique solution to \[
            \frac{1}{M_v^2}\sum_{l = 2}^s \sum_{S \subseteq V(G), |S| = l, v \in S} M_S^2 \left(\frac{kM_v}{2^i\widetilde{M}}\right)^{l - 1} = 1
          \]  For $i \leq x$, $1 \leq 2^{-|i - x|} \frac{1}{M_v^2}\sum_{l = 2}^s \sum_{S \subseteq V(G), |S| = l, v \in S} M_S^2 \left(\frac{kM_v}{2^i\widetilde{M}}\right)^{l - 1}$. For $i \geq x$, \[
            \left( \frac{1}{M_v^2}\sum_{l = 2}^s \sum_{S \subseteq V(G), |S| = l, v \in S} M_S^2 \left(\frac{kM_v}{2^i\widetilde{M}}\right)^{l - 1}\right)^2 \leq 2^{-|i - x|}\frac{1}{M_v^2}\sum_{l = 2}^s \sum_{S \subseteq V(G), |S| = l, v \in S} M_S^2 \left(\frac{kM_v}{2^i\widetilde{M}}\right)^{l - 1}
          \] So in either case, the contribution to  $\mathbb{E}\left\lbrack Y^2\right\rbrack$ is bounded by:

          \begin{align*}
            &\lesssim \frac{\widetilde{M}}{M_v}\sum_{i = \max\left\lbrace 0, -\left\lceil\log \frac{\widetilde{M}}{\omega M_v}\right\rceil \right\rbrace}^{\left\lceil \log \frac{\sqrt k}{\omega}\right\rceil} 2^{i - |i - x|} \sum_{l = 2}^s \sum_{S \subseteq V(G), |S| = l, v \in S} M_S^2 \left(\frac{kM_v}{2^i\widetilde{M}}\right)^{l - 1} \\
            & \leq \sum_{i = \max\left\lbrace 0, -\left\lceil\log \frac{\widetilde{M}}{\omega M_v}\right\rceil \right\rbrace}^{\left\lceil \log \frac{\sqrt k}{\omega}\right\rceil} 2^{- |i - x|} \sum_{l = 2}^s \sum_{S \subseteq V(G), |S| = l, v \in S} M_S^2 k^{l-1}\left(\frac{M_v}{2^i\widetilde{M}}\right)^{l - 2}\\
            & \leq \sum_{i = 0}^{\infty} 2^{-i}\sum_{l = 2}^s \sum_{S \subseteq V(G), |S| = l, v \in S} M_S^2 \frac{k^{l-1}}{\left(2^i\omega\right)^{l-2}}\\
            & \lesssim \sum_{l = 2}^s \sum_{S \subseteq V(G), |S| = l, v \in S} M_S^2 \frac{k^{l-1}}{\omega^{l-2}}
          \end{align*}

          And so, by summing across all $v$, the contribution is:

          \begin{align*}
            \lesssim \sum_{l = 2}^s C_l \frac{k^{l-1}}{\omega^{l-2}}
          \end{align*}

      \end{description}

      Now, by summing across all the cases and using the fact that, with $\omega \leq \sqrt{k}$ and $l \geq 2$, $k^\frac{l}{2} \leq \frac{k^{l-1}}{\omega^{l-2}}$ (and noting that $C_l \leq M^2$), this gives us:

      \begin{align*}
        \var(\overline{M_H}) & \leq \mathbb{E}\left\lbrack Y^2 \right\rbrack\\
        &\lesssim M^2 + \sum_{l = 2}^s C_l \frac{k^{l-1}}{\omega^{l-2}}
      \end{align*}

    \end{proof}

    \subsection{\texorpdfstring{$\overline{M}$}{M}}
    \begin{lemma}
      $\mathbb{E}\left\lbrack\overline{M}\right\rbrack = M$.
    \end{lemma}
    \begin{proof}
      By Lemma \ref{thexpec2}, $\mathbb{E}\left\lbrack\overline{M} | M_L \right\rbrack = \mathbb{E}\left\lbrack\overline{M_L} | M_L \right\rbrack + \mathbb{E}\left\lbrack\overline{M_H} | M_L \right\rbrack = \mathbb{E}\left\lbrack\overline{M_L} | M_L \right\rbrack + \mathbb{E}\left\lbrack\overline{M_H} | M_H \right\rbrack = M_L + M_H = M$. So $\mathbb{E}\left\lbrack\overline{M}\right\rbrack = M$.
    \end{proof}

    \begin{lemma}
      \label{mvariance}
      $\var\left(\overline{M}\right) \lesssim M^2 + \sum_{l = 2}^s C_l\left(k^\frac{l}{2} + k\left(\frac{C_1^+}{M^2}k\right)^{l - 2}\right)$ 
    \end{lemma}

    \begin{proof}
      $\var(M) \lesssim \var(M_L) + \var(M_H)$, so from Lemmas \ref{thvar2} and \ref{tlvar2}, using the fact that $\frac{k^{l-1}}{\omega^{l-2}} \geq k^\frac{l}{2}$ for $\omega \leq \sqrt{k}$, $l \geq 2$, we get  $\var\left(\overline{M}\right) \lesssim M^2 + C_2k\log k + \sum_{l = 3}^s C_l \frac{k^{l-1}}{\omega^{l-2}}$. The result then follows from $\omega = \min \left\lbrace \frac{\left(\widetilde{M}\right)^2}{C_1^+}, \sqrt{k} \right\rbrace$. 
    \end{proof}

    \subsection{Single-Pass Algorithm}
    \begin{lemma}
      \label{medgecost}
      The first and second pass calculations can be performed while storing no more than $O\left(\frac{m\log \frac{M_V\sqrt{k}}{M^2}}{k}\right)$ edges in expectation.
    \end{lemma}
    \begin{proof}
      We will now demonstrate how both conceptual passes can be calculated in a single pass, by only calculating the value $\mathcal{M}_v$ for a vertex $v$ when necessary.

      We will store the following edges (noting that edges in this graph are undirected, and so an edge $e = uv$ is stored if it would be stored as either $uv$ or $vu$):

      \begin{align*}
        E_0 = \lbrace uv \in E & | \exists i, \mathcal{M}_i(u) = 1, \mathcal{M}_i(v) = 1\rbrace\\
        E_1 = \lbrace uv \in E & | \exists i, c(u) = 1, \mathcal{M}_i(v) = 1\rbrace\\
        E_2 = \lbrace uv \in E & | c(u) = 1, c(v) = 1\rbrace\\
        E_3 = \lbrace uv \in E & | \exists i, r_{D,i}(uv) < \frac{2^{h(u)}\omega}{\sqrt{k}}, d(v) = 1\rbrace\\
        E_4 = \lbrace uv \in E & | r_C(uv) < \frac{2^{h(u)}\omega}{\sqrt{k}}, c(v) = 1\rbrace
      \end{align*}

      Then, $\forall uv \in E$:

      \begin{align*}
        \mathbb{P}\left\lbrack uv \in E_0 \right\rbrack & \lesssim \frac{1}{k}\\
        \mathbb{P}\left\lbrack uv \in E_1 \right\rbrack & \lesssim \frac{1}{k}\\
        \mathbb{P}\left\lbrack uv \in E_2 \right\rbrack & = \frac{1}{k}\\
        \mathbb{P}\left\lbrack uv \in E_3 \right\rbrack & = \frac{1}{\sqrt{k}} \sum_{i = 1}^2\sum_{h = 0}^{\left\lceil \log \frac{\sqrt{k}}{\omega} \right\rceil} \mathbb{P}\left\lbrack r_{D,i}(uv) < \frac{2^{h(u)}\omega}{\sqrt{k}} \middle| h(u) = h \right\rbrack \mathbb{P}\left\lbrack h(u) = h \right\rbrack\\
        & \lesssim \frac{1}{\sqrt{k}} \sum_{h = 0}^{\left\lceil \log \frac{\sqrt{k}}{\omega} \right\rceil} \frac{2^{h}\omega}{\sqrt{k}} \frac{1}{\omega 2^h}\\
        & \lesssim \frac{\log \frac{\sqrt{k}}{\omega}}{k}\\
        \mathbb{P}\left\lbrack uv \in E_4 \right\rbrack & \lesssim\frac{\log \frac{\sqrt{k}}{\omega}}{k}\\ 
      \end{align*}

      So these sets will contain $O\left(\frac{m \log  \frac{\sqrt{k}}{\omega}}{k}\right)$ edges in expectation. Then, as $\omega = \min \left\lbrace \frac{(\wt{M})^2}{M_V^+},\sqrt{k}\right\rbrace$, and $M_V^+$ can be any upper bound on $M_V$, we can achieve this by storing $O\left(\frac{m \log  \frac{M_V\sqrt{k}}{M^2}}{k}\right)$ edges.

      We will then use these to calculate $\overline{M_L}$ and $\overline{M_H}$ as follows:

      \begin{description}
        \item[$\overline{M_L}$:] Recall that $\overline{M_L}$ is the number of copies of $A$ in $G_L'$, the subgraph of $G$ induced by $V_L' = \lbrace u \in V(G) | c(u) = 1, \mathcal{M}_u \not = \lo$. $E_2$ will contain every edge in $E(G_L')$, so we can calculate $\overline{M_L}$ provided we can calculate $\mathcal{M}_v$ for all $v$ s.t. $c(v) = 1$.

          To do this, we will need to calculate $\sum_{a \in \mu(\lbrace v\rbrace)} X_{v, i}^{(h,a)}$ for $h = \left\lceil \log \frac{\sqrt k}{\omega} \right\rceil$, and each $i \in \lbrack 2\rbrack$. For each subgraph $a \in \mu(\lbrace v\rbrace)$, $X_{v, i}^{(h,a)} = 1$ iff 

          \begin{align*}
            \forall u \in V(a) \setminus \lbrace v \rbrace, \mathcal{M}_i(u) & = 1\\
            \forall \lbrace u | vu \in E(a) \rbrace, r_{D,i}(vu) & < \frac{\omega 2^h}{\sqrt k}
          \end{align*}

          And 0 otherwise. The third and fourth conditions always hold when $h = \left\lceil \log \frac{\sqrt k}{\omega} \right\rceil$, so we need to know how many such subgraphs $a$ exist with $\forall u \in V(a) \setminus \lbrace v \rbrace, \mathcal{M}_i(u) = 1$. But in that case, $\lbrace uw \in E(a) | u \not = v \not = w\rbrace \subseteq E_0$, and $\lbrace vu | vu \in E(a) \rbrace \subseteq E_1$, so $\sum_{a \in \mu(\lbrace v\rbrace)} X_{v, i}^{(h,a)}$ will be equal to the number of subgraphs $a\in \mu(\lbrace v\rbrace)$ s.t. $\forall u \in V(a) \setminus \lbrace v \rbrace, \mathcal{M}_i(u) = 1$ in the edges we have sampled. So then $\mathcal{M}_{v,i} = \lo$ iff this number is $< \frac{\wt{M}}{k^\frac{3}{2}}$. 

          So we can compute $\mathcal{M}_v$ for each $v$ s.t. $c(v) = 1$ using our sampled edges, and therefore we can compute $\overline{M_L}$.

        \item[$\overline{M_H}$:] For any $v \in V(G)$, $\overline{M_v} = 0$ if $h(v) \not =  \mathcal{M}_v$. So it is sufficient to calculate $\mathcal{M}_v$ when $h(v) = \mathcal{M}_v$, and to know that $h(v) \not = \mathcal{M}_v$ otherwise. We can calculate  $\sum_{a \in \mu(\lbrace v\rbrace)} X_{v, i}^{(h,a)}$ for each $h \leq h(v)$, $i \in \lbrack 2 \rbrack$, as for any $a \in \mu(\lbrace v\rbrace)$ s.t. $\forall u \in V(a) \setminus \lbrace v \rbrace, \mathcal{M}_i(u)  = 1$, $\lbrace uw \in E(a) | u \not = v \not = w\rbrace \subseteq E_0$, and if $\forall \lbrace u | vu \in E(a) \rbrace, r_{D,i}(vu)  < \frac{\omega 2^h}{\sqrt k}$ for some $h \leq h(v)$, $\lbrace vu | vu \in E(a) \rbrace \subseteq E_3$. 

          So then, as $H_{v,i} = \left\lbrace h \in \left\lbrace 1, \dots, \left\lceil \log \frac{\sqrt k}{\omega} \right\rceil\right\rbrace \middle| X_{v,i}^{(h)} \geq \frac{\omega \wt{M}2^h}{k^2}\right\rbrace$ we can compute $H_{v,i} \cap \lbrace 0, \dots, h(v) \rbrace$. As $\mathcal{M}_{v,i} = \min H_{v,i}$, we can compute each $\mathcal{M}_{v,i}$ if it is $\leq h$, and determine that it is $> h(v)$ or $\lo$ otherwise. Then, if $\mathcal{M}_v \leq h(v)$, at least one of $\mathcal{M}_{v,1},\mathcal{M}_{v,2}$ is $\leq h(v)$, and so we can calculate it and therefore calculate $\mathcal{M}_v$, and if not we can determine that $\mathcal{M}_v$ is either $> h(v)$ or $\lo$. (although not necessarily which one)

          So if $\mathcal{M}_v \not = h(v)$, we know $\overline{M_v} = 0$, and then if $\mathcal{M}_v = h(v)$, we need to know how many $a \in \mu(\lbrace v\rbrace)$ there are such that:

          \begin{align*}
            \forall u \in V(a) \setminus \lbrace v \rbrace, c(u) & = 1\\
            \forall \lbrace u | vu \in E(a) \rbrace, r_{C}(vu) & < \frac{\omega 2^h}{\sqrt k}
          \end{align*}

          If these criteria hold, then $a \subseteq E_2 \cap E_4$. So we can calculate $\overline{M_v}$ by counting the number of subgraphs amongst our sampled edges such that these criteria hold.

          We then calculate $\overline{M_H}$ by summing $\overline{M_v}$ for each $v \in V(G)$.

      \end{description}
    \end{proof}

    We are now ready to prove our generalized theorem.
    \restate{thm:general}
    \begin{proof}
      By Lemma \ref{mvariance}, if we choose $\wt{M} \leq M$, we can obtain variance $\lesssim M^2$ with \[k = \sum_{l = 2}^s\left(\left(\frac{M^2}{C_l}\right)^\frac{2}{l} + \left(\frac{M^2}{C_l}\right)^\frac{1}{l-1}\left(\frac{M^2}{C_1}\right)^\frac{l-2}{l-1}\right) = \sum_{l = 2}^s\left(\left(\frac{M^2}{C_l}\right)^\frac{2}{l} + \frac{M^2}{\left(C_l C_1^{l -2}\right)^\frac{1}{l-1}}\right)\], by taking our upper bound $C_1^+$ within a constant factor of the true value of $C_1$. 

      So as by Lemma \ref{medgecost}, we need to sample $\frac{\log k}{k}$ edges to attain this, and we can then repeat the algorithm $O\left(\frac{1}{\epsilon^2}\log\frac{1}{\delta}\right)$ times, taking a median-of-means of our results, to get an estimate within $\epsilon M$ of $M$ with probability $1 - \delta$. 

      Our result then follows from the fact that $f_l = \Theta\left(\frac{C_l}{M^2}\right)$.
    \end{proof}

    \begin{corollary}
      We can obtain a constant-error approximation to $M$ while keeping
      \[
        O\left(m\sum_{l = 2}^s
        \left(f_l^\frac{2}{l} + f_l^\frac{1}{l -1}f_1^{1 - \frac{1}{l-1}}\right)\right)
      \]
      edges. 
    \end{corollary}

    \section{Transitivity coefficient}\label{app:transitivity}
    Let $P_2$ denote the number of length $2$ paths in the graph, so the
    transitivity coefficient $\alpha = 3T/P_2$.  The condition that the
    graph have no isolated edges implies that $P_2 \gtrsim m$, or $\alpha
    \lesssim T/m$.  Hence to show that the bound in~\cite{BFKP14} is
    weaker than the one in~\cite{PT12}, it suffices to show the following
    lemma.
    \begin{lemma}
      Consider any graph with $m$ edges, $T$ triangles, $P_2$ length-2
      paths, transitivity coefficient $\alpha = 3T/P_2$, and at most $\Delta_E$
      triangles sharing a common edge.  Then
      \[
        \frac{\sqrt{m}}{\alpha} + \frac{m \sqrt{m}}{T} \gtrsim m\left(\frac{\Delta_E}{T} + \frac{1}{\sqrt{T}}\right)
      \]
    \end{lemma}
    \begin{proof}
      First, consider the $\frac{m\Delta_E}{T}$ term on the right.  This is
      trivially bounded by the left if $\Delta_E \leq \sqrt{m}$.  Otherwise,
      note that because $\Delta_E$ triangles share a common edge, the
      corresponding vertices have degree at least $\Delta_E + 1$ so $P_2 \geq 2
      \binom{\Delta_E + 1}{2} \geq \Delta_E^2$.  Hence
      \[
        \frac{\sqrt{m}}{\alpha} = \frac{\sqrt{m}P_2}{3T} \geq
        \frac{\sqrt{m}\Delta_E^2}{3T} \geq \frac{m \Delta_E}{3T}
      \]
      from $\Delta_E > \sqrt{m}$, so $\frac{m\Delta_E}{T}$ is bounded by the LHS.

      Now, consider bounding the $m/\sqrt{T}$ term.  If $T \geq m$, then
      $P_2 \geq 3T \geq 3m$, so
      \[
        \frac{\sqrt{m}}{\alpha} = \frac{\sqrt{m}P_2}{3T}  \geq \frac{m}{\sqrt{T}} \sqrt{\frac{P_2}{3T}} \geq \frac{m}{\sqrt{T}}.
      \]
      On the other hand, for $T \leq m$, $m/\sqrt{T} \leq \frac{m
        \sqrt{m}}{T}$.

        Hence both terms of the RHS are bounded by the LHS, so the sum is
        bounded by twice the LHS.
      \end{proof}
      The result of~\cite{JSP13} for triangle counting is also implies by
      the~\cite{PT12} bound.  In~\cite{JSP13}, $\frac{m}{\eps^2\sqrt{T}}$
      space is used to learn $\alpha$ to $\pm \eps$; this gives a sample
      complexity of
      \[
        \frac{m}{\sqrt{T}}\cdot \left(\frac{P_2}{T}\right)^2
      \]
      for learning a constant multiplicative approximation to $T$.
      \begin{lemma}
        Consider any graph with $m$ edges, $T$ triangles, $P_2$ length-2
        paths, transitivity coefficient $\alpha = 3T/P_2$, and at most $\Delta_E$
        triangles sharing a common edge.  Then
        \[
          \frac{m}{\sqrt{T}}\cdot \left(\frac{P_2}{T}\right)^2 \gtrsim m\left(\frac{\Delta_E}{T} + \frac{1}{\sqrt{T}}\right)
        \]
      \end{lemma}
      \begin{proof}
        The $\frac{m}{\sqrt{T}}$ term on the right is trivial, since $T
        \lesssim P_2$.  For the other term, similarly to the previous proof,
        we have
        \[
          \frac{m}{\sqrt{T}}\cdot \left(\frac{P_2}{T}\right)^2 \gtrsim \frac{m}{\sqrt{T}}\cdot \left(\frac{P_2}{T}\right)^{.5}  = \frac{m\sqrt{P_2}}{T} \geq \frac{m \Delta_E}{T}.
        \]
      \end{proof}

      \end{document}

%% file: libtheorems.tex
\usepackage{etoolbox}

\makeatletter

\newcommand{\define}[4][ignore]{%
  \ifstrequal{#1}{ignore}{}{
  \@namedef{thmtitle@#2}{#1}}%
  \@namedef{thm@#2}{#4}%
  \@namedef{thmtypen@#2}{lemma}%
  \newtheorem{thmtype@#2}[theorem]{#3}%
  \newtheorem*{thmtypealt@#2}{#3~\ref{#2}}%
}

\newcommand{\state}[1]{%
  \@namedef{curthm}{#1}
  \@ifundefined{thmtitle@#1}{
  \begin{thmtype@#1}
    }{
  \begin{thmtype@#1}[\@nameuse{thmtitle@#1}]
  }
    \label{#1}
    \@nameuse{thm@#1}
  \end{thmtype@#1}
  \@ifundefined{thmdone@#1}{
  \@namedef{thmdone@#1}{stated}%
  }{}
}

\newcommand{\restate}[1]{%
  \@namedef{curthm}{#1}
  \@ifundefined{thmtitle@#1}{
    \begin{thmtypealt@#1}
    }{
  \begin{thmtypealt@#1}[\@nameuse{thmtitle@#1}]
  }
    \@nameuse{thm@#1}
  \end{thmtypealt@#1}
  \@ifundefined{thmdone@#1}{
  \@namedef{thmdone@#1}{stated}%
  }{}
}

\newcommand{\thmlabel}[1]{
  \@ifundefined{thmdone@\@nameuse{curthm}}{\label{#1}
    }{\tag*{\eqref{#1}}}
}
\makeatother

%% file: graphcommands.tex
    \pdfoutput=1
    \newcommand{\drawclique}{
      \begin{tikzpicture}[every node/.style={scale=0.3}]
        \graph [nodes={circle, fill=darkgray}, clockwise, empty nodes,simple] {subgraph K_n[n=6]
          ; 1 -!- 4; 2 -!- 6; 3 -!- 2; 5 -!- 1; 5-!-6;
        };
      \end{tikzpicture}
    }
    \newcommand{\drawheavy}{
      \begin{tikzpicture}[every node/.style={scale=0.3}]
        \foreach \i in {-5,-10,...,-15}{
          \node[fill=darkgray, circle, radius=0.04cm] (x\i) at (\i pt, 7pt) {};
          \node[fill=darkgray, circle, radius=0.04cm] (y\i) at (\i pt, -7pt) {};
          \draw (x\i) -- (y\i);
        }

        \node[fill=darkgray, circle, radius=0.04cm] (x) at (0, 7pt) {};
        \node[fill=darkgray, circle, radius=0.04cm] (y) at (0, -7pt) {};
        \draw (x) -- (y);
        \foreach \i in {3,8,14,22}{
          \node[fill=darkgray, circle, radius=0.04cm] (z\i) at (\i pt, 0) {};
          \draw (x) -- (z\i) -- (y);
        };
      \end{tikzpicture}
    }
    \newcommand{\drawhub}{
      \begin{tikzpicture}[every node/.style={scale=0.3}]
        \graph [nodes={circle, fill=darkgray}, empty nodes, n=8] { 
          x[xshift=1cm] -- 
          {[clockwise,simple] subgraph C_n; 1 -!- 2; 3 -!- 4; 5 -!- 6; 7 -!- 8;};
          x -- 2;x--4;x--6;x--8;
        };
      \end{tikzpicture}
    }
    \newcommand{\drawindep}{
      \begin{tikzpicture}[every node/.style={scale=0.3}]
        \foreach \i in {0,2.5,...,6}{
          \graph [nodes={circle, fill=darkgray,xshift=\i cm}, empty nodes, n=3, clockwise] { subgraph C_n};
        }
      \end{tikzpicture}
    }

%% file: hedge.tex
\pdfoutput=1
\tikzset{basevertex/.style={shape=circle, line width=0.5,
    minimum size=4pt, inner sep=0pt, draw}}
\tikzset{defaultvertex/.style={basevertex, fill=gray!40}}
\begin{figure}[h]
  \centering

    \begin{subfigure}[t]{\textwidth}
      \centering
      \begin{tikzpicture}[%
          VertexStyle/.style={defaultvertex}]
          \SetVertexNoLabel
          \Vertex{nw1}
          \SO(nw1){nw2}
          \NOEA(nw1){ne1}
          \EA(nw1){ne2}
          \EA(nw2){ne3}
          \SOEA(nw2){ne4}
          \Edge(nw1)(nw2)
          \Edges[color=red,lw=2pt](nw1,ne1,nw2)
          \Edges[color=red,lw=2pt](nw1,ne2,nw2)
          \Edges[color=red,lw=2pt](nw1,ne3,nw2)
          \Edges[color=red,lw=2pt](nw1,ne4,nw2)
          \WE(nw1){nww1}
          \WE(nww1){nwww1}
          \WE(nwww1){nwwww1}
          \WE(nw2){nww2}
          \WE(nww2){nwww2}
          \WE(nwww2){nwwww2}
          \Edge[color=blue,lw=2pt](nww1)(nww2)
          \Edge(nwww1)(nwww2)
          \Edge[color=blue,lw=2pt](nwwww1)(nwwww2)
        \end{tikzpicture}
        \caption*{Encoding the string $1010$ in $D_{1,4}$, Alice sends the edges in blue, connecting the $i^\text{th}$ pair of vertices if the $i^\text{th}$ bit of the string is 1. Bob then queries the $4^\text{th}$ position by adding the edges in red, connecting the $i^\text{th}$ pair to $d = 4$ wedges. The graph will contain $d$ triangles if the $i^\text{th}$ bit is 1, and 0 otherwise.}
      \end{subfigure}
      \caption{Heavy edge graph, showing that $\Omega(m
        \frac{\Delta_E}{T})$ is necessary for insertion streams.}
  \label{fig:lowerheavy}
\end{figure}

%% file: hub.tex
\pdfoutput=1
\tikzset{basevertex/.style={shape=circle, line width=0.5,
    minimum size=4pt, inner sep=0pt, draw}}
\tikzset{defaultvertex/.style={basevertex, fill=gray!40}}

\begin{figure}[h]
  \centering
      \begin{subfigure}[t]{0.48\textwidth}
        \centering
        \begin{tikzpicture}[%
            VertexStyle/.style={defaultvertex}]
            \SetVertexNoLabel

            \Vertex[x=0,y=0]{h}
            \Vertex[x=3.293,y=0.707]{s1}
            \Vertex[x=4,y=1]{s2}
            \Vertex[x=4.707,y=0.707]{s3}
            \Vertex[x=5,y=0]{s4}
            \Vertex[x=4.707,y=-0.707]{s5}
            \Vertex[x=4,y=-1]{s6}
            \Vertex[x=3.293,y=-0.707]{s7}
            \Vertex[x=3,y=0]{s8}
            
            \Edge[color=blue,lw=2pt](h)(s1)
            \Edge[color=blue,lw=2pt](h)(s5)
            \Edge[color=blue,lw=2pt](h)(s7)
            \Edge[color=red,lw=2pt](s1)(s4)
            \Edge[color=red,lw=2pt](s2)(s5)
            \Edge[color=red,lw=2pt](s3)(s8)
            \Edge[color=red,lw=2pt](s6)(s7)

          \end{tikzpicture}
          \caption{Encoding a ``trip''. }
        \end{subfigure}
        \begin{subfigure}[t]{0.48\textwidth}
          \centering
          \begin{tikzpicture}[%
              VertexStyle/.style={defaultvertex}]
              \SetVertexNoLabel
              \Vertex[x=0,y=0]{lc}
              \Vertex[x=-1,y=0.5]{lnw1}
              \Vertex[x=-0.5,y=1]{lnw2}
              \Vertex[x=1,y=0.5]{lne1}
              \Vertex[x=0.5,y=1]{lne2}
              \Vertex[x=1,y=-0.5]{lse1}
              \Vertex[x=0.5,y=-1]{lse2}
              \Vertex[x=-1,y=-0.5]{lsw1}
              \Vertex[x=-0.5,y=-1]{lsw2}
              \Edge(lc)(lnw1)
              \Edge[color=blue,lw=2pt](lc)(lnw2)
              \Edge(lc)(lne1)
              \Edge(lc)(lne2)
              \Edge(lc)(lse1)
              \Edge[color=blue,lw=2pt](lc)(lse2)
              \Edge(lc)(lsw2)
              \Edge[color=blue,lw=2pt](lc)(lsw1)
              \Edge[color=red,lw=2pt](lnw1)(lnw2)
              \Edge[color=red,lw=2pt](lne1)(lne2)
              \Edge[color=red,lw=2pt](lsw1)(lsw2)
              \Edge[color=red,lw=2pt](lse1)(lse2)
          \end{tikzpicture}
          \caption{The same encoding, shown as a subgraph of $H_{1,4}$.}
        \end{subfigure}\vspace{.7em}

        \small{
          We have one hub vertex, and a set of $2rd$
spoke vertices. Alice receives a subset of spoke vertices; she then
creates the edges shown in blue, drawing an edge from the hub vertex
to each of her vertices. Bob receives a partial matching on the spoke
vertices; he adds these $d$ edges (shown in red).  Correctly counting
triangles in the resulting problem would solve Boolean Hidden Matching,
which takes $\Omega(r \sqrt{d})$ communication.}
  \caption{Hub graph, showing that $\Omega(m \frac{\sqrt{\Delta_V}}{T})$ is necessary for insertion
  streams.}
  \label{fig:lowerhub}
\end{figure}

%% file: indep.tex
\pdfoutput=1
\tikzset{basevertex/.style={shape=circle, line width=0.5,
    minimum size=4pt, inner sep=0pt, draw}}
\tikzset{defaultvertex/.style={basevertex, fill=gray!40}}

\begin{figure}[h!]
  \centering
      \begin{subfigure}[t]{\textwidth}
        \centering
        $\begin{tikzpicture}[%
            VertexStyle/.style={defaultvertex}]
            \SetVertexNoLabel 
            \Vertices[x=0,y=0,dir=\EA]{line}{b1,b2}
            \Vertex[x=0.5,y=0.866]{t1}
            \Edges(b1,b2,t1,b1)
        \end{tikzpicture} \rightarrow
        \left\lbrace
        \begin{tikzpicture}[%
          VertexStyle/.style={basevertex}]
          \SetVertexNoLabel 
          \Vertex[x=0,y=0,style={fill=blue}]{b1}
          \Vertex[x=1,y=0,style={fill=blue}]{b2}
          \Vertex[x=0.5,y=0.866,style={fill=blue}]{t1}
          \Edges(b1,b2,t1,b1)
         \end{tikzpicture},
        \begin{tikzpicture}[%
          VertexStyle/.style={basevertex}]
          \SetVertexNoLabel 
          \Vertex[x=0,y=0,style={fill=blue}]{b1}
          \Vertex[x=1,y=0,style={fill=red}]{b2}
          \Vertex[x=0.5,y=0.866,style={fill=blue}]{t1}
          \Edge(b1)(t1)
         \end{tikzpicture},
        \begin{tikzpicture}[%
          VertexStyle/.style={basevertex}]
          \SetVertexNoLabel 
          \Vertex[x=0,y=0,style={fill=red}]{b1}
          \Vertex[x=1,y=0,style={fill=red}]{b2}
          \Vertex[x=0.5,y=0.866,style={fill=blue}]{t1}
          \Edge(b1)(b2)
        \end{tikzpicture},
        \begin{tikzpicture}[%
          VertexStyle/.style={basevertex}]
          \SetVertexNoLabel 
          \Vertex[x=0,y=0,style={fill=red}]{b1}
          \Vertex[x=1,y=0,style={fill=blue}]{b2}
          \Vertex[x=0.5,y=0.866,style={fill=blue}]{t1}
          \Edge(b2)(t1)
         \end{tikzpicture}
      \right\rbrace$
      \caption{The four (equally likely) possible results for applying the first subgraphing scheme to a single triangle.}
    \end{subfigure}
      \begin{subfigure}[t]{\textwidth}
        \centering
        $\begin{tikzpicture}[%
            VertexStyle/.style={defaultvertex}]
            \SetVertexNoLabel 
            \Vertices[x=0,y=0,dir=\EA]{line}{b1,b2}
            \Vertex[x=0.5,y=0.866]{t1}
            \Edges(b1,b2,t1,b1)
        \end{tikzpicture} \rightarrow
        \left\lbrace
        \begin{tikzpicture}[%
          VertexStyle/.style={basevertex}]
          \SetVertexNoLabel 
          \Vertex[x=0,y=0,style={fill=blue}]{b1}
          \Vertex[x=1,y=0,style={fill=blue}]{b2}
          \Vertex[x=0.5,y=0.866,style={fill=blue}]{t1}
         \end{tikzpicture},
        \begin{tikzpicture}[%
          VertexStyle/.style={basevertex}]
          \SetVertexNoLabel 
          \Vertex[x=0,y=0,style={fill=blue}]{b1}
          \Vertex[x=1,y=0,style={fill=red}]{b2}
          \Vertex[x=0.5,y=0.866,style={fill=blue}]{t1}
          \Edges(b1,b2,t1)
         \end{tikzpicture},
        \begin{tikzpicture}[%
          VertexStyle/.style={basevertex}]
          \SetVertexNoLabel 
          \Vertex[x=0,y=0,style={fill=red}]{b1}
          \Vertex[x=1,y=0,style={fill=red}]{b2}
          \Vertex[x=0.5,y=0.866,style={fill=blue}]{t1}
          \Edges(b1,t1,b2)
        \end{tikzpicture},
        \begin{tikzpicture}[%
          VertexStyle/.style={basevertex}]
          \SetVertexNoLabel 
          \Vertex[x=0,y=0,style={fill=red}]{b1}
          \Vertex[x=1,y=0,style={fill=blue}]{b2}
          \Vertex[x=0.5,y=0.866,style={fill=blue}]{t1}
          \Edges(b2,b1,t1)
         \end{tikzpicture}
      \right\rbrace$
      \caption{The four (equally likely) possible results for applying the second subgraphing scheme to a single triangle.}
    \end{subfigure}
    \small{\newline\newline
      A sampling scheme that
      only looks at one or two of the edges will observe identical
      distributions under the two regimes, so a sampling scheme that
    distinguishes the regimes must sample at least one complete triangle.}
    \caption{Independent triangles, showing that $\Omega(m/T^{2/3})$ is necessary for sampling algorithms. }
  \label{fig:lowerindep}
\end{figure}